\documentclass[11pt]{article}
\usepackage[utf8]{inputenc}
\usepackage{geometry}
\usepackage[nodisplayskipstretch]{setspace}
\usepackage{xcolor}

\definecolor{dark1}{RGB}{157, 58, 103}
\definecolor{dark2}{RGB}{161, 67, 0}
\definecolor{dark3}{RGB}{115, 102, 0}
\definecolor{dark4}{RGB}{2, 120, 50}
\definecolor{dark5}{RGB}{0, 116, 122}
\definecolor{dark6}{RGB}{18, 100, 176}
\definecolor{darker6}{RGB}{13.5, 75, 132}
\definecolor{dark7}{RGB}{116, 75, 163}
\definecolor{mid1}{RGB}{225, 119, 163}
\definecolor{mid2}{RGB}{228, 128, 77}
\definecolor{mid3}{RGB}{182, 158, 21}
\definecolor{mid4}{RGB}{87, 182, 109}
\definecolor{mid5}{RGB}{0, 181, 190}
\definecolor{mid6}{RGB}{88, 162, 242}
\definecolor{mid7}{RGB}{177, 135, 229}
\definecolor{light1}{RGB}{255, 187, 231}
\definecolor{light2}{RGB}{255, 198, 151}
\definecolor{light3}{RGB}{247, 223, 104}
\definecolor{light4}{RGB}{152, 248, 171}
\definecolor{light5}{RGB}{72, 249, 255}
\definecolor{light6}{RGB}{164, 219, 255}
\definecolor{light7}{RGB}{244, 192, 255}

\definecolor{python0}{RGB}{31, 119, 180}
\definecolor{python1}{RGB}{255, 127, 14}
\definecolor{python2}{RGB}{44, 160, 44}
\definecolor{python3}{RGB}{214, 39, 40}
\definecolor{python4}{RGB}{148, 103, 189}
\definecolor{python5}{RGB}{140, 86, 75}
\definecolor{python6}{RGB}{227, 119, 194}
\definecolor{python7}{RGB}{127, 127, 127}
\definecolor{python8}{RGB}{188, 189, 34}
\definecolor{python9}{RGB}{23, 190, 207}

\usepackage{amsmath}
\usepackage{amssymb}
\usepackage{amsthm}
\usepackage{thmtools} %
\usepackage{thm-restate}

\usepackage[colorlinks]{hyperref}
\hypersetup{linkcolor = dark7, citecolor = dark5, urlcolor = dark2}
\usepackage{cleveref}

\newtheorem{lemma}{Lemma}
\newtheorem{proposition}{Proposition}

\newtheorem{claim}{Claim}

\theoremstyle{definition}
\newtheorem{definition}{Definition}

\newtheorem{example}{Example}
\newtheorem*{excont}{Example \protect\continuation}

\newcommand{\continuation}{??}
\newenvironment{examplecont}[1]{
  \renewcommand{\continuation}{\ref{#1}}
  \excont[Continued]
}{\endexcont}

\usepackage{algorithm}
\usepackage{algpseudocode}

\usepackage[ttscale=.875]{libertine}
\usepackage[scaled=0.785]{beramono} %
\usepackage[libertine]{newtxmath} %
\usepackage[T1]{fontenc}

\usepackage{optidef}
\usepackage{natbib}
\usepackage{siunitx}
\usepackage{diffcoeff}
\usepackage{booktabs}
\usepackage{array}
\usepackage{graphicx}
\usepackage{tikz}
\tikzstyle{vertex} = [circle, fill=black!10]
\tikzstyle{edge} = [->, > = latex, draw=black, thick]
\usepackage[noblocks]{authblk}
\usepackage{ushort}
\usepackage{eqparbox}

\usepackage{etoc}
\usepackage{pgfplots}
\usepackage[skip=3pt]{subcaption}

\usepackage{enumitem}
\setlist{topsep=2pt, itemsep=2pt, partopsep=0pt, parsep=0pt, leftmargin=*}

\usepackage{titlesec}
\titlespacing*{\paragraph}{0pt}{2ex plus 1ex minus .2ex}{1em}
\titleformat{\subparagraph}[runin]{\normalfont\normalsize\itshape}{\thesubparagraph}{1em}{}
\titlespacing*{\subparagraph}{0pt}{1.2ex plus 1ex minus .2ex}{1em}

\allowdisplaybreaks

\usepackage{amsmath}
\usepackage{xparse}
\usepackage{bm}
\usepackage{dsfont}

\newcommand{\vb}{{\bm b}}
\newcommand{\vc}{{\bm c}}
\newcommand{\vd}{{\bm d}}

\newcommand{\vg}{{\bm g}}

\newcommand{\vp}{{\bm p}}
\newcommand{\vq}{{\bm q}}

\newcommand{\vw}{{\bm w}}
\newcommand{\vx}{{\bm x}}
\newcommand{\vy}{{\bm y}}
\newcommand{\vz}{{\bm z}}
\newcommand{\vzero}{{\bm 0}}

\newcommand{\vone}{{\bm 1}}
\newcommand{\vonen}{{\bm 1}_{n}}
\newcommand{\vpi}{{\bm\pi}}
\newcommand{\vPi}{{\bm\Pi}}
\newcommand{\vphi}{{\bm\phi}}
\newcommand{\vbeta}{{\bm\beta}}
\newcommand{\vgamma}{{\bm\gamma}}
\newcommand{\vzeta}{{\bm\zeta}}
\newcommand{\veta}{{\bm\eta}}
\newcommand{\vchi}{\bm \chi}
\newcommand{\vrho}{{\bm\rho}}

\newcommand{\vtheta}{{\bm\theta}}
\newcommand{\vdelta}{{\bm\delta}}

\newcommand{\mA}{{\bm A}}
\newcommand{\mB}{{\bm B}}

\newcommand{\mI}{{\bm I}}

\newcommand{\mM}{{\bm M}}

\newcommand{\mW}{{\bm W}}

\newcommand{\gB}{{\mathcal B}}
\newcommand{\gC}{{\mathcal C}}
\newcommand{\gD}{{\mathcal D}}

\newcommand{\gF}{{\mathcal F}}

\newcommand{\gH}{{\mathcal H}}

\newcommand{\gJ}{{\mathcal J}}
\newcommand{\gK}{{\mathcal K}}
\newcommand{\gL}{{\mathcal L}}

\newcommand{\gS}{{\mathcal S}}
\newcommand{\gT}{{\mathcal T}}
\newcommand{\gU}{{\mathcal U}}

\newcommand{\sN}{{\mathbb N}}

\newcommand{\sR}{{\mathbb R}}

\newcommand{\E}{{\mathbb E}}
\newcommand{\R}{{\mathbb R}}

\newcommand{\trp}{\mathsf{T}}
\newcommand{\D}{\mathrm{D}}

\let\ab\allowbreak

\let\originalleft\left
\let\originalright\right
\renewcommand{\left}{\mathopen{}\mathclose\bgroup\originalleft}
\renewcommand{\right}{\aftergroup\egroup\originalright}

\DeclareDocumentCommand\p{d<> o o}{
  \operatornamewithlimits{P} \IfNoValueF{#1}{_{#1}}
  \IfNoValueF{#2}{
    \left(#2 \IfNoValueF{#3}{\,\middle|\,#3}\right)
  }
}

\RenewDocumentCommand\E{d<> o o}{
  \operatornamewithlimits{\mathbb E} \IfNoValueF{#1}{_{#1}}
  \IfNoValueF{#2}{
    \left[#2 \IfNoValueF{#3}{\,\middle|\,#3}\right]
  }
}

\NewDocumentCommand\V{d<> o o}{
  \operatornamewithlimits{Var} \IfNoValueF{#1}{_{#1}}
  \IfNoValueF{#2}{
    \left(#2 \IfNoValueF{#3}{\,\middle|\,#3}\right)
  }
}

\NewDocumentCommand\I{d<> o o}{
  \operatorname{\mathds 1} \IfNoValueF{#1}{_{#1}}
  \IfNoValueF{#2}{
    \left[#2 \IfNoValueF{#3}{\,\middle|\,#3}\right]
  }
}

\newcommand{\Iterations}{14}
\newcommand{\OptimalWelfare}{313102}
\newcommand{\NaiveWelfare}{256846}
\newcommand{\NaiveWelfareRatio}{82.0\%}
\newcommand{\NewtonMultiplier}{19.65}
\newcommand{\NewtonWelfare}{312422}
\newcommand{\NewtonWelfareRatio}{99.8\%}
\newcommand{\NewtonWelfareRatioGap}{0.2\%}
\newcommand{\NewtonLyapunov}{1e-06}

\newcommand{\EventDays}{5}

\newcommand{\SimpleWelfareRatio}{99.8\%}

\newcommand{\ObservedDrivers}{3368}
\newcommand{\BalancedDrivers}{4475}
\newcommand{\ObservedDriversRatio}{75.3\%}

\newcommand{\subij}{_{i,j}}
\newcommand{\ie}{i.e.\@ }
\newcommand{\eg}{e.g.\@ }
\newcommand{\st}{s.t.\@ }
\newcommand{\wrt}{with respect to }
\newcommand{\supth}{\textsuperscript{th}}

\newcommand{\supt}{^{(t)}}

\newcommand{\suptmo}{^{(t-1)}}
\newcommand{\one}[1]{\mathds{1} \{ #1\}}
\newcommand{\outcome}{$(\vx, \vy, \vp)$}
\newcommand{\economy}{$(m, \vd, \vc, \vq)$}

\newcommand{\qtilde}{\tilde{q}}
\newcommand{\qhat}{\hat{q}}
\newcommand{\vqhat}{\hat{\vq}}
\newcommand{\vqtilde}{\tilde{\vq}}
\newcommand{\phantomSlack}{e}

\newcommand{\barmA}{\bar{\mA}}
\newcommand{\bargamma}{\bar{\gamma}}
\newcommand{\barvgamma}{\bar{\vgamma}}

\newcommand{\xobs}{x^\mathrm{obs}}
\newcommand{\pobs}{p^\mathrm{obs}}

\newcommand{\PyplotScale}{0.75}
\newcommand{\MmaScale}{0.88}

\newcommand{\historyt}{h_t}
\newcommand{\historySett}{\gH_t}
\newcommand{\suptprime}{^{(t')}}

\title{Iterative Network Pricing for Ridesharing Platforms%
\footnotetext{The authors would like to thank Francisco Castro, Yeon-Koo Che, Peter Frazier, Daniel Freund, Nikhil Garg, Sergey Gitlin, Ramesh Johari, Yash Kanoria, Cinar Kilcioglu, Zhen Lian, Ilan Lobel, Elliot Lipnowski, Will Ma, Jake Marcinek, Hongseok Namkoong, Hamid Nazerzadeh, David C. Parkes, Amin Saberi, Shreyas Sekar, Peng Shi, Daniel Russo, Garrett van Ryzin, Carmen Wang, Chiwei Yan, and participants at 
various seminars and conferences for helpful comments and discussions.
}
}
\author{Chenkai Yu\footnote{Decision, Risk, and Operations Division, Columbia Business School, \texttt{cyu26@gsb.columbia.edu}} 
\hspace{2cm} 
Hongyao Ma\footnote{Decision, Risk, and Operations Division, Columbia Business School, \texttt{hongyao.ma@columbia.edu}}
}

\date{November 14, 2023}

\begin{document}

\maketitle

\begin{abstract}
  Ridesharing platforms match riders and drivers, using dynamic pricing to balance supply and demand. The origin-based “surge pricing”, however, does not take into consideration market conditions at trip destinations, leading to inefficient driver flows in space and incentivizes drivers to strategize. In this work, we introduce the \emph{Iterative Network Pricing} mechanism, addressing a main challenge in the practical implementation of optimal origin-destination (OD) based prices, that the model for rider demand is hard to estimate. Assuming that the platform's surge algorithm clears the market for each origin in real-time, our mechanism updates the OD-based price adjustments week-over-week, using only information immediately observable during the same time window in the prior weeks. For stationary market conditions, we prove that our mechanism converges to an outcome that is approximately welfare-optimal. %
  Using data %
  from the City of Chicago, we illustrate (via simulation) the iterative updates under our mechanism for morning rush hours, demonstrating substantial welfare improvements despite significant fluctuations of market conditions from early 2019 through the end of 2020.
\end{abstract}

\section{Introduction} \label{sec:intro}

Ridesharing platforms have %
revolutionized the way people commute and travel in recent years.
Central to their success is the emphasis placed on providing reliable transportation services to riders.%
\footnote{\label{fn:reliability_for_riders}%
  Lyft, for example, advertises itself as ``a ride whenever you need one'' (\url{https://www.lyft.com/}, accessed May 16, 2023), and Uber aims to provide ``transportation that is as reliable as running water'' (\url{https://www.uber.com/en-au/blog/melbourne/transportation-that-is-as-reliable-as-running-water/}, accessed May 16, 2023).}
When rider demand surpasses available driver supply in an area, %
the dynamic ``surge pricing'' adopted by these platforms increases the prices for trips originating from the area. This maintains sufficient density of driver supply in space, and
guarantees that rider wait times would not exceed a few minutes~\citep{rayle2014app,castillo2017surge,yan2020dynamic}.

When a platform's origin-based pricing algorithm does not appropriately factor
market conditions at the destination of the trips, however, the resulting spatial distribution of driver supply can be highly inefficient.
This issue becomes %
especially pronounced when demand patterns are driven more by trip destinations than origins. %
For instance, during morning rush hours in major cities, %
a substantial number of riders seek trips that end downtown, irrespective of whether they start in the city center or in a residential area. However, relatively few riders request trips that end in the residential neighborhoods. %

\begin{figure}[t!]
  \centering
  \subcaptionbox{Average flow imbalance of different neighborhoods. \label{fig:imbalance_heatmap}}{
    \begin{tikzpicture}
      \node (image) at (1.8, -0.3) {\includegraphics[width = 0.47 \textwidth]{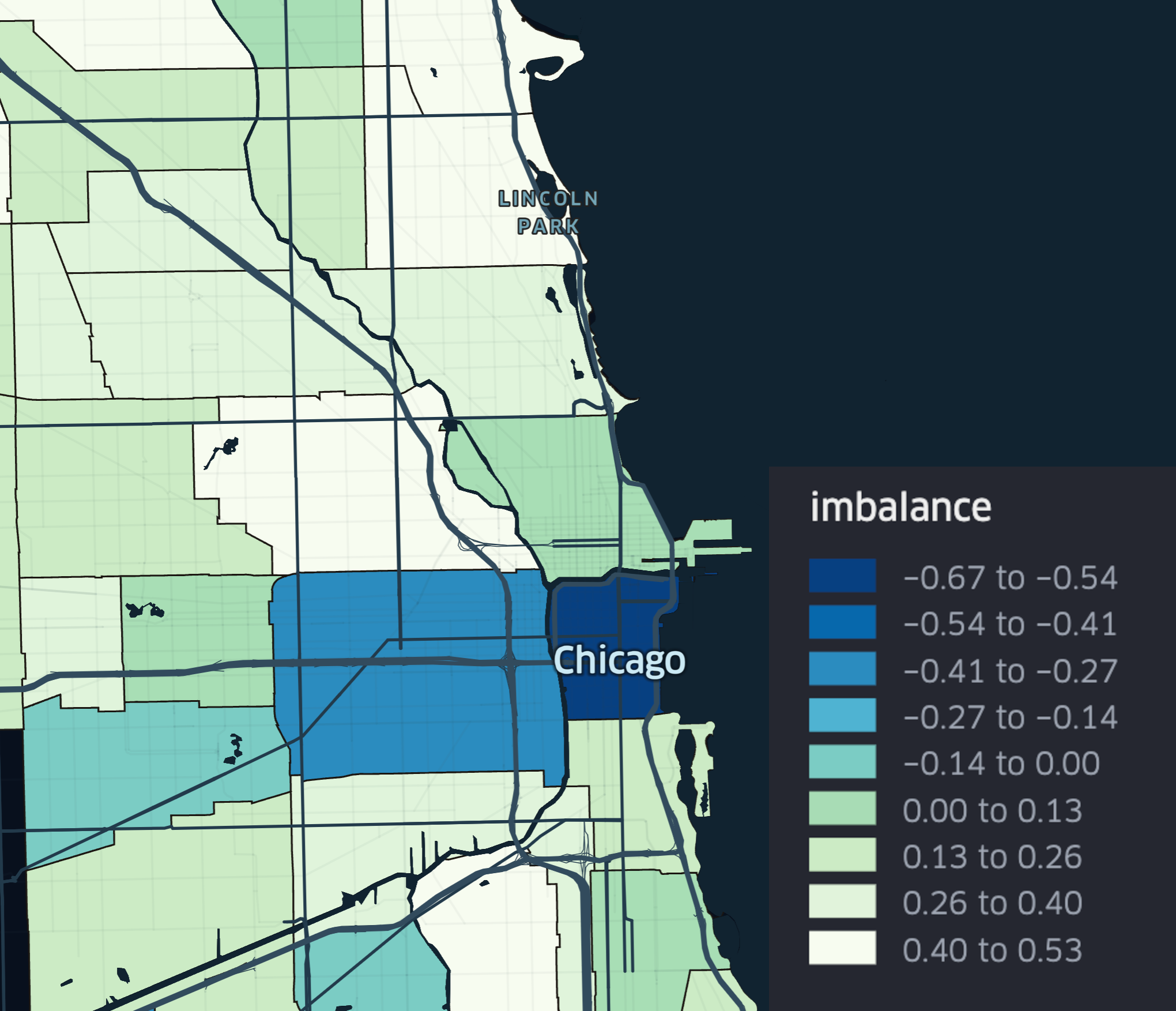}};
      \draw[-latex, line width = 2pt, python9] (0, 0) -- (1, 1.8);
      \draw[-latex, line width = 2pt, dashed, python3] (0, 0) -- (1.8, -1);
    \end{tikzpicture}
  }
  \hspace{1.5em}
  \subcaptionbox{Flow imbalance over time (weeks), 2019-2020.
    \label{fig:imbalance_over_week}}{\includegraphics[width = 0.42 \textwidth]{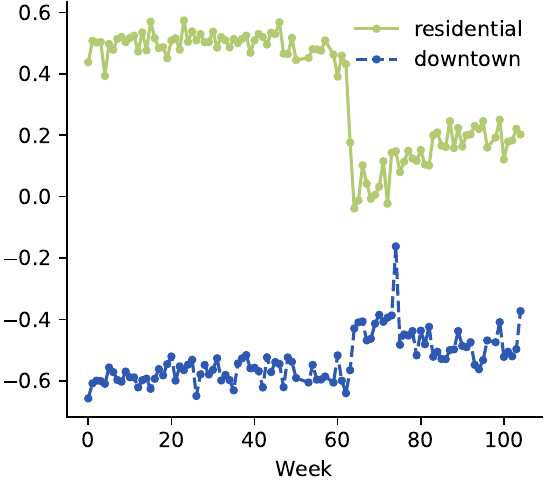}}
  \caption{The relative imbalance of trip flow
    during weekday morning rush hours %
    in Chicago: average of %
    first 9 weeks of 2020 (left) and temporal trends for two neighborhoods (right). %
    Week 60 corresponds to the COVID shutdown in March 2020, which substantially changed the %
    market dynamics and reduces the imbalance of the residential area due to work-from-home arrangements. See Section~\ref{sec:simulations} for more details.
  }
  \label{fig:chicago}
\end{figure}

\Cref{fig:imbalance_heatmap} illustrates the (relative) flow imbalance of different neighborhoods in Chicago during morning rush hours.
The flow imbalance of a location is defined as $(\text{outflow} - \text{inflow}) / (\text{outflow} + \text{inflow})$,
where the \emph{outflow} and \emph{inflow} refer respectively to the number of trips that originate from and end at the location per hour.
Downtown Chicago (usually referred to as ``The Loop'') is a highly popular destination.
Prior to the COVID shut-down, an average of over 4000 trips ending in this area per hour while just slightly more than 1000 trips originate here. In contrast, the North Side's %
Lake View neighborhood sees over 1000 trips originating per hour, while only approximately 300 trips end there %
in the same amount of time.

Under origin-based pricing, the two trips depicted in Figure~\ref{fig:imbalance_heatmap} will have identical prices since they start from the same location and cover the same distance. This %
incentivizes drivers to cherrypick, since %
a driver arriving in the residential area is likely to be dispatched a high surge trip immediately, while a driver who arrive in the downtown area may need to deadhead (\ie relocate without a rider) for a substantial period of time before receiving any dispatch.
This disparity in drivers' continuation earnings also suggests that origin-based pricing is inefficient in its use of drivers' time, and fails to properly distribute drivers towards areas where supply is more valuable.

\medskip

In the rapidly growing literature on ridesharing platforms, some existing work already studies origin-destination (OD) based pricing to provide efficient and reliable transportation. %
In particular, \citet{ma2022spatio} %
show that
it is not only welfare-optimal but also incentive-aligned (\ie it is to the drivers' best interest to accept all trip dispatches) to set OD-based prices of the following form: %
\[
  \begin{split}
    (\text{trip price}) = (\text{trip cost})
    & + (\text{marginal value of a driver at trip origin}) \\
    & - (\text{marginal value of a driver at trip destination}).
  \end{split}
\]

We refer to the difference between the marginal value of a driver at a trip's origin versus its destination as the \emph{OD-based adjustment}. This represents the network externality of completing the trip, and can also be interpreted as the value of teleporting a driver from the destination back to the origin.
Since a driver %
in the residential area is intuitively more ``valuable'' than one downtown, the trip to downtown will command a higher price under optimal OD pricing. This not only %
reduces driver supply in the already saturated downtown area (by reducing rider demand for trips heading downtown), but also compensates the driver for the lower continuation earnings from downtown onward with a higher earning from the current trip.

Despite these theoretical advancements, %
the practical implementation of optimal OD-based pricing presents significant challenges. In theory, the marginal values of drivers at different locations are derived from the dual variables of the welfare-optimization problem. Formulating the primal and dual programs, however,  necessitates a rider demand model (\ie the expected number of riders requesting a ride at any given price for each OD pair), which proves very difficult to estimate in practice for a number of reasons:
\begin{itemize} %
  \item \emph{Endogeneity}. Consider a trip typically priced at \$20. The platform is likely to have a solid understanding of rider demand at this price, %
        as well as the %
        local price sensitivity (due to natural price fluctuations resulting from fluctuating supply levels, or through local price exploration).
        However, estimating rider demand at %
        \$10 or \$30 presents considerable challenges.
        First, the trip may never have been priced as low as \$10. Furthermore, although prices might have surged to \$30 during inclement weather or special events, %
        such observations are not %
        as useful when %
        pricing trips under normal market conditions.
  \item \emph{Different short term and long term elasticity.} An abrupt reduction in the price of a trip from \$20 to \$10 might induce a modest demand increase in the short run, %
    but one could plausibly anticipate %
    that the demand will continue to increase for many weeks, 
    as riders react %
    by, for instance, deciding against renewing their monthly work parking passes. %
    Long-term %
    price exploration is, however, %
    not only costly but also inequitable (subjecting %
    rider groups to different prices for extended periods of time). %
  \item \emph{Heterogeneity in space and time.} The shape of the demand function may vary substantially across trip origin, destination, and time of the week. %
    These variations could be attributed to neighborhood-specific characteristics, public transit availability, fluctuating traffic conditions, etc. Consequently, assuming a universal functional form for all origins, destinations, and times of the week (whether in terms of prices or surge multipliers) is unlikely to yield accurate estimations of rider demand.%
\end{itemize}

\subsection{Our Results} \label{sec:intro_results}

In this work, we propose the \emph{iterative network pricing} (INP) mechanism, enabling a ridesharing platform to implement optimal OD-based pricing without the need for estimating a rider demand model.
Assuming that the platform’s %
dynamic pricing algorithm %
finds the origin-based market-clearing multipliers in real-time, our mechanism iteratively updates the OD-based additive price adjustments week-over-week, relying only on information immediately observable by the platform  %
during the preceding weeks.

Take the Chicago morning rush hour as an example. %
Given that a driver heading to %
the residential area has substantially higher continuation earnings compared to a driver transporting a rider to downtown, the marginal value of a driver downtown is intuitively lower, 
thus a trip to downtown should be priced higher. %
Without a demand model, the platform cannot %
compute the optimal price difference between the two trips. %
The platform %
can, however, %
set the downtown trip's price to be \$1 higher than the %
trip to the residential area for the following week.
With this additive adjustment, we are likely to observe %
a reduced difference in drivers' continuation earnings from trips to these two locations, 
once the platform's origin-based surge algorithm clears the market. %
A significant disparity %
is still likely to exist, at which point
the platform may increase this price difference to \$2, %
observe the market-clearing outcome for another week, and so forth.

In practice, iteratively updating the OD-based adjustments %
can be effective, as Figure~\ref{fig:imbalance_over_week}
suggests that the spatial imbalances of trip flows remain relatively stationary week-over-week (in the absence of turmoils such as COVID shutdown, which took place around week 62). %
Moreover, it is relatively easily to adapt the platform's online-control style %
origin-based pricing algorithm %
to accommodate additive adjustments. %
To operationalize this idea, %
we address address in this work the following questions:
\begin{itemize} %
  \item %
    Given a current origin-based market-clearing outcome, what would constitute a ``direction'' for updating the OD-based adjustments %
    for a metropolitan area with more than two neighborhoods?
  \item What should be an appropriate stepsize that the platform %
    takes each week in this desired direction?
  \item Will this iterative process converge? If so, will it converge to the optimal OD-based prices?
\end{itemize}

\paragraph{Model}

We study a continuous time, non-atomic setting, modeling the operation of a ridesharing platform during a particular time-window of the week, \eg morning rush hours on Wednesdays.%
\footnote{We focus on a particular time of the week because the demand pattern during the morning rush hours on a Wednesday is typically %
similar to that of the previous Wednesday. A platform may adopt multiple INP mechanisms in parallel %
for other time windows with different demand patterns, e.g. weekday mid-day, evening rush hours, or weekend bar hours. 
}
There is a fixed amount of driver supply and a set of discrete locations. 
Each origin-destination (OD) pair, representing a potential trip, is characterized by the duration of the trip, the cost incurred by a driver for fulfilling the trip, and a rider demand function which specifies the amount of riders requesting this trip (per unit of time) at different price levels.
The full rider demand functions are not known to the platform, but the platform can observe %
the demand levels and the (local) price sensitivities at %
the realized trip prices each week.

We assume that supply and demand are stationary during %
the focal time window in any given week, %
and that the market condition does not change week-over-week.%
\footnote{ Appendix~\ref{appx:Chicago_market_dynamics} discusses the %
stationarity assumption for the market conditions within our focal time window. 
Briefly, the duration of most trips is significantly shorter than the period over which market conditions remain relatively unchanged.
}
We also assume that during each week, the platforms' dynamic pricing algorithm determines origin-based multipliers that \emph{clear the market}, %

meaning that (i) riders are picked-up if and only if they are willing to pay the prices of their trips, (ii) all drivers are dispatched to %
pick-up a rider or relocate to some location, and (iii) the inflow and outflow of drivers are balanced for every location.

Finally, we assume that  drivers always follow the platform's dispatches. %

\paragraph{Main Results}

We first establish a welfare theorem (Lemma~\ref{lem:welfare_theorem}), that an outcome %
is welfare-optimal if and only if it can be supported by %
OD-based trip prices in competitive equilibrium (CE). %
We show that optimal CE prices consist of three parts: (i) the cost incurred by drivers to complete the trip; (ii) a multiplicative part (\ie the trip duration times a multiplier, which is uniform across all locations) reflecting the opportunity cost (of drivers' time) for fulfilling the trip; and (iii) an additive adjustment dependent on the origin and destination of the trip, representing the difference in the marginal values of a driver at the two locations. %

A platform is not able to directly compute the optimal %
OD-based adjustments without the rider demand model. %
When the OD-based adjustments are absent or suboptimal, the platform's %
dynamic pricing algorithm will need different origin-based multipliers to clear the market for all locations.
We show that for any given set of OD-based adjustments, there exists a unique origin-based market-clearing outcome, provided the platform coordinates the relocation of drivers who are not assigned rider trips (Lemma~\ref{lem:unique}). 
Furthermore, %
the welfare loss of such outcomes (relative to the optimal CE outcome) can be %
bounded by the differences in multipliers across different locations (Theorem 1), despite the fact that %
evaluating %
the achieved welfare or the optimal social welfare is not possible without a demand model.

An iterative network pricing (INP) mechanism starts with an %
initial origin-based market-clearing outcome with no OD-based adjustments. At each subsequent timestep (\eg the following week), after observing the outcomes of the previous timesteps, the mechanism updates the OD-based adjustments, and observes the corresponding new outcome resulting from the market-clearing process.
The trip throughput, driver earnings, and local price sensitivity from past %
outcomes allow us to evaluate the gradient of the dual objective (\ie the dual of the welfare-optimization problem). 
The natural idea of gradient descent, however, %
specifies a direction to update not only the OD-based adjustments but also the origin-based multipliers. This does not solve our problem, since the origin-based multipliers are results of the market-clearing process, %
instead of being controlled directly by the platform.

To identify a proper direction for updating the OD-based adjustments, a platform needs to anticipate how the market-clearing multipliers will change as the platform updates the OD-based adjustments. We show that this is possible using information immediately available %
from the past market-clearing outcomes (Lemma~\ref{thm:differentiable}). 
The INP mechanism we propose updates the OD-based adjustments towards the direction to equalize %
(the linear approximations of) the origin-based multipliers at all locations, and %
takes an appropriately sized step %
to ensure the multipliers don't change too drastically at any given location.
By establishing a correspondence between our mechanism and the Newton %
algorithm %
for solving systems of non-linear equations, we prove that our mechanism converges to an outcome that is approximately welfare-optimal when the market condition is stationary. %
Using data from Chicago's morning rush hours, we illustrate (via simulation) the iterative updates under our mechanism, %
demonstrating substantial welfare improvements despite significant fluctuations of market conditions from early 2019 through the end of 2020.

\subsection{Related Work} \label{sec:related_work}

The rapidly growing literature on ridesharing platforms covers both empirical studies of current platforms and theoretical analysis of market designs.
Empirical studies have demonstrated the effectiveness of %
origin-based surge pricing in improving reliability and efficiency~\citep{hall2015effects}, increasing driver supply during high-demand times~\citep{chen2015dynamic}, as well as creating incentives for drivers to relocate to higher surge areas~\citep{lu2018surge}.
A wide range of additional topics were also studied, including
consumer surplus~\citep{cohen2016using,castillo2020benefits},
the labor market of Uber drivers~\citep{hall2017labor},
the longer-term labor market equilibration~\citep{hall2016analysis},
the value of flexible work arrangements~\citep{chen2019value,chen2020reservation,xu2020empirical} and the gender earnings gap associated to ``learning-by-doing''~\citep{cook2018gender}.
To the best of our knowledge, no empirical study addresses the challenge of demand model estimation for ridesharing platforms in the presence of endogeneity, substantial heterogeneity, and different short term and long term demand elasticity.

A variety of policy and regulation related topics have also received attention, %
including work analyzing the optimal growth of two-sided platforms~\citep{lian2021optimal, fang2019prices}, competition between platforms~\citep{lian2021larger,ahmadinejad11920competition,fang2020loyalty}, operations in the presence of %
autonomous vehicles~
\citep{ostrovsky2019carpooling,lian2020autonomous}, and utilization-based minimum wage regulations~\citep{asadpour2023minimum}.

\medskip

In regard to dynamic matching in ridesharing platforms, \citet{ashlagi2018maximum}, \citet{dickerson2018allocation} and \citet{aouad2020dynamic} focus on matching between riders and drivers and the pooling of shared rides, taking into consideration the online arrival of supply and demand in space.
Further, \citet{kanoria2020blind}, \citet{qin2020ride} and \citet{ozkan2020dynamic} design policies that dispatch drivers from areas with relatively abundant supply, while \citet{cai2019role} and \citet{pang2017efficiency} look at the role of information availability and transparency in platform design.
\citet{castro2022randomized} demonstrate how to use drivers' waiting times (in airport driver queues) to align incentives and reduce inequity in earnings when some trips are necessarily more lucrative than the others.
State-dependent dispatching~\citep{banerjee2018state,castro2020matching}, driver admission control~\citep{afeche2023ride} and capacity planning~\citep{besbes2022spatial} are also studied using queueing-theoretical models.

On the pricing side, \citet{castillo2017surge} and \citet{yan2020dynamic} demonstrate the importance of origin-based dynamic %
pricing in maintaining supply density in space and reducing waiting times for riders;
\citep{BanJohRiq2015Pricing} demonstrate the robustness of dynamic pricing in comparison to static pricing when the platform may have imprecise information on the arrival rate of rider demand; \citet{bimpikis2019spatial} and \citet{besbes2020surge} study revenue-optimal pricing when supply and demand are imbalanced in space;
\citet{garg2019driver} show that additive instead of multiplicative %
pricing is more incentive aligned for drivers when prices need to be origin-based only;  \citet{rheingans2019ridesharing} study pricing in the presence of driver location preferences;
\citet{freund2021pricing} discuss the cycles of volatile driver supply when the platform adopts two prices %
when pricing trips, and \citet{yu2022price} analyzes the price cycles arising from divers collectively turning offline to drive up prices.

\medskip

The closest works to ours in the ridesharing literature are \citet{ma2019spatio} and \citet{cashore2022dynamic}, which we briefly discussed earlier in the introduction.
\citet{ma2019spatio} study a deterministic, complete  information model that is sufficiently flexible to model spatial imbalance and temporal variation of supply and demand. They demonstrate that various existing market failure in today's ridesharing platforms are symptoms of inefficient pricing, and prove that optimal OD-based pricing with the structure discussed earlier is both welfare-optimal and %
incentive-compatible (\ie it forms a subgame-perfect equilibrium for all drivers to accept all dispatches from the platform).
\citet{cashore2022dynamic} generalize the model to a stochastic setting, and prove that the same pricing structure is still %
optimal and incentive aligned in large markets. Both works, however, assume that the platform has access to the correct demand model, which has proven to be very difficult to estimate in practice. 

In this work, we %
propose mechanisms that %
optimizes origin-destination based pricing via iterative updates without using any model for rider demand.

\section{Preliminaries} \label{sec:preliminaries}

In this section, we introduce a continuous time, non-atomic model for the operation of a ridesharing platform for \emph{a particular time period of the week}, \eg %
morning rush hour on Wednesdays. During this period of time, we assume that driver supply and rider demand are stationary. We %
characterize welfare-optimal, competitive equilibrium (CE) outcomes for the time-window of interest, which serve as the target of the week-over-week iterative pricing process that we discuss in Sections~\ref{sec:origin_based_surge} and~\ref{sec:INP}. %

\medskip

There is a total of $m > 0$ units of driver supply, and a set of $n > 1$ discrete locations $\gL = \{1, 2, \dots, n\}$. For each pair of locations $i, j \in \gL$, $(i, j) \in \gL^2$ denotes a \emph{trip} from
$i$ to $j$. $d_{i, j} > 0$ is the \emph{duration} of the trip, \ie it takes a driver $d_{i, j}$ units of time to drive from $i$ to $j$. Any driver who completes a trip from $i$ to $j$ (with or without a rider) incurs a \emph{cost} $c_{i, j} \ge 0$, which models the cost of time, driving, fuel, wear-and-tear, etc. %
To allow different traffic patterns, trip durations do not have to be symmetric, \ie it is possible that $d_{i, j} \neq d_{j, i}$ for $i \neq j$. Similarly, it is not necessary that the trip costs be symmetric.

For each trip $(i,j) \in \gL^2$, $q_{i, j}: \R_{\ge 0} \to \R_{\ge 0}$ is the \emph{demand function} for the trip, with $q_{i, j}(r)$ being the amount of riders (per unit of time) traveling from $i$ to $j$ and willing to pay at least $r$.
We assume that for all trips $(i, j) \in \gL^2$, (i) $q_{i, j}(\cdot)$ is continuously differentiable and  %
strictly decreasing (\ie $q_{i, j}'(r) < %
  0$, $\forall r \in \R_{\ge 0}$), and (ii) $\int_0^\infty q_{i, j}(r) \dl r < \infty$, \ie the total rider value is finite.%
\footnote{
  We work with demand functions in this paper, which %
  determine riders' value %
  distributions, and vice versa.
  Let $Q\subij > 0$ be the number of riders (per unit of time) who would request a trip $(i,j) \in \gL^2$ at zero price,  %
  and $F\subij$ be the cumulative distribution function (CDF) of their value. We have $q\subij(r) = Q\subij \cdot (1 - F\subij(r))$ for all $r \geq 0$.
  $q\subij(\cdot)$ being continuously differentiable and strictly decreasing means that the %
  probability density function of rider values $\difs{F_{i, j}(r)}{r}$ is continuous and strictly positive for all $r \geq 0$.
  Also note that assumption (ii) $\int_0^\infty q_{i, j}(r) \dl r < \infty$ implies $\lim_{r \to +\infty} q_{i, j}(r) = 0$.
}

We assume the platform can observe the total driver supply $m$, trip duration $\vd = (d\subij)_{(i,j)\in \gL^2}$, and trip costs $\vc = (c\subij)_{(i,j)\in \gL^2}$.
The rider demand model $\vq = (q\subij)_{(i,j)\in \gL^2}$ is not known to the platform, though the platform is able to observe the demand and local price sensitivity %
at the current trip prices.
Formally, if a trip $(i,j) \in \gL^2$ is priced at $r$, %
the platform observes $q\subij(r)$ and $q'\subij(r)$. %
The tuple $(m, \vd, \vc, \vq)$ specifies an \emph{economy}, \ie the market conditions during this focal time window (\eg %
morning rush hour on Wednesdays).

\paragraph{Outcomes}
An \emph{outcome} describes the flow of drivers and riders in space, as well as payments for drivers and riders associated to each trip. %
Formally, an outcome is a triple $(\vx, \vy, \vp)$, where for all $(i, j) \in \gL^2$, %
$x\subij \geq 0$ is the rider flow rate from location $i$ to location $j$, $y_{i, j} \ge 0$ is the rate of driver flow from $i$ to $j$ (%
either completing a rider trip, or relocating without a rider), and $p \subij \ge 0$ is the price of the $(i,j)$ trip that is collected from all $(i,j)$ riders who are picked up, and paid to all drivers traveling from $i$ to $j$ with a rider.%
\footnote{ %
  By construction, the outcomes are strictly budget-balanced. All results presented in this paper still hold if the platform takes a fixed percentage of drivers' surplus from every trip as commission.%
}
An outcome is feasible and stationary if the followings hold:%
\begin{enumerate}[label = (F\arabic*)]
  \item \label{cond:feasibility_nonneg_price} Trip prices are non-negative, \ie $\forall i, j \in \gL,\ p_{i, j} \ge 0$.
  \item \label{cond:feasibility_demand_supply} %
        The rider flow does not exceed the driver flow for any trip, \ie $\forall i, j \in \gL,\ x_{i, j} \le y_{i, j}$.
  \item \label{cond:feasibility_supply_cnst}
        The %
        outcome does not require more drivers than %
        the available supply, \ie $\sum_{i, j \in \gL} d_{i, j} y_{i, j} \le m$. %
  \item \label{cond:feasibility_flow_balance}
        The driver flow into and out of every location is balanced, \ie $\forall i \in \gL,\ \sum_{j \in \gL} y_{i, j} = \sum_{j \in \gL} y_{j, i}$.
\end{enumerate}

Unless stated otherwise, when we mention an outcome in the rest of the paper, it is assumed to be feasible and stationary.
We also assume throughout the paper that riders requesting the same trip are picked up in decreasing order of their values.%
\footnote{
  This is without loss for welfare-optimization, and always holds %
  when rider best-response is satisfied, \ie riders are picked up if and only if they are willing to pay the trip price (see condition \ref{cond:CE_rider_br}).
  The mechanisms we propose in this paper satisfy rider best-response at all times. %
}
For each $(i,j) \in \gL^2$, %
let $v_{i, j}(\cdot)$ be %
the inverse function of $q_{i, j}(\cdot)$, \ie %
$v_{i, j}(s)$ the value of the $s$\supth {} %
rider traveling from $i$ to $j$.
In this way, %
when $x\subij \in [0, q_{i, j}(0)]$ riders are picked up for the $(i,j)$ trip, their total value can be written as $\int_0^{x_{i, j}} v_{i, j}(s) \dl s$.
Since $q_{i, j}(r)$ is continuous and strictly decreasing for each $(i,j) \in \gL^2$, its inverse $v\subij(s)$ is well defined and is strictly decreasing for all $s \in (0, q_{i, j}(0))$. %
For convenience, we define $v_{i, j}(s) = \infty$ for $s \le 0$ and $v_{i, j}(s) = -\infty$ for $s > q_{i, j}(0)$.

\paragraph{Cycles and Driver Surplus}
Under an outcome \outcome, $\vy$ specifies the overall flow of drivers in space. To describe the movements of individual drivers, as well as their payoffs, it is convenient to use the cycle representation of driver flows.
A \emph{directed cycle} is a sequence of %
trips $((i_1, i_2), (i_2, i_3), \dots, \ab (i_{\ell}, i_1))$ visiting %
$\ell$ distinct locations before returning to the origin of the first trip.
Let $\gC$ be the set of all directed cycles.
We say a cycle $\kappa \in \gC$ \emph{covers} a trip $(i,j) \in \gL^2$ if the cycle includes the trip from location~$i$ to location~$j$, and we denote this as $(i,j) \in \kappa$.

By the flow decomposition theorem (see Theorem~3.5 of \citet{ahuja1988network}), every balanced flow can be decomposed into a weighted sum of cycles, and vice versa. %
Formally, given %
driver flow $\vy \in \R_{\ge 0}^{n^2}$,
\begin{equation} \label{eq:balance=cycles}
  \forall i \in \gL, \ \sum_{j \in \gL} y_{i, j} = \sum_{j \in \gL} y_{j, i}
  \iff
  \exists \vw \in \R_{\ge 0}^{|\gC|}, \  \forall i, j \in \gL, \ \sum_{\kappa %
    \in \gC} w_\kappa \one{(i, j) \in \kappa } = y_{i, j}.
\end{equation}
Here $\one{\cdot}$ is the indicator function, and we refer to $\vw$ as a \emph{(cycle) decomposition} of $\vy$.
Intuitively, in any balanced driver flow, every driver's movement in space %
forms a directed cycle, and $w_\kappa$ is the %
amount of drivers %
traveling along cycle $\kappa$.
Given an outcome \outcome, for each directed cycle $\kappa \in \gC$, denote %
\begin{equation}
  \mu_\kappa \triangleq \frac{\sum_{(i, j) \in \kappa} (p_{i, j} - c_{i, j})}{\sum_{(i, j) \in \kappa} d_{i, j}}. \label{eq:cycle_surplus_rate}
\end{equation}
Consider a driver traveling along $\kappa$. If for every trip on the cycle, either (i) the driver picks up a rider, or (ii) the trip price is zero when the driver relocates without a rider, then $\mu_\kappa$ represents the surplus rate of the driver, \ie the total payments minus costs from all trips %
divided by the total trip duration.

\begin{example}[The Morning Rush]
  \label{ex:running-example}
  Consider an economy illustrated in Figure~\ref{fig:simple-pattern} with $n = 2$ locations: the residential neighborhood (location 1) and the downtown area (location 2). Every minute during the morning rush hours, 10 riders travel from the residential neighborhood to downtown, 20 riders travel within the downtown area, but no rider demands a trip ending in the residential neighborhood.

  \begin{figure}[t!]
    \centering
    \begin{tikzpicture}[shorten >= 1pt, scale = 2, auto]
      \foreach \pos/\name in {(0, 0)/1, (1.5, 0)/2}
      \node[vertex] (\name) at \pos {\name};
      \draw[edge, bend left] (1) to node{%
        10 riders/min} (2);
      \draw[edge, bend left, dashed] (2) to node{no demand} (1);
      \draw[edge, out = -150, in = 150, loop, dashed] (1) to node{no demand} (1);
      \draw[edge, out = 30, in = -30, loop] (2) to node{%
        20 riders/min} (2);
    \end{tikzpicture}
    \caption{%
      A stylized economy illustrating the demand pattern during morning rush hours.}
    \label{fig:simple-pattern}
  \end{figure}
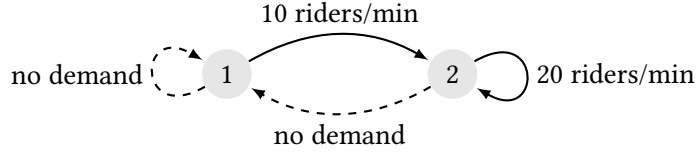

  The demand functions are given by
  $q_{1, 2}(r) = 10 \exp(-r / 40)$, %
  $q_{2, 2}(r) = 20 \exp(-r / 10)$,  %
  and $q_{1, 1}(\cdot) = q_{2, 1}(\cdot) = 0$. %
  This effectively assumes that riders' values are exponentially distributed with mean $40$ and $10$ for the $(1,2)$ and $(2,2)$ trips, respectively.
    Trips within each location takes 10 minutes ($d_{1, 1} = d_{2, 2} = 10$) and trips between locations take 20 minutes ($d_{1, 2} = d_{2, 1} = 20$).
  Finally, there are $m = 240$ available drivers, and we assume for simplicity that all trip costs are zero, \ie $c_{i, j} = 0$ for all $i, j \in \{1, 2\}$.

  \subparagraph{A Feasible Outcome}

  We can verify that %
  the outcome $(\vx, \vy, \vp)$ specified as follows satisfies conditions \ref{cond:feasibility_nonneg_price}--\ref{cond:feasibility_flow_balance}, and therefore is feasible and stationary:  %
  \begin{itemize} %
    \item the %
          the rider flow is given by $x_{1,2} = 1$, $x_{2,2} = 20$, and $x_{1,1} = x_{2,1} = 0$,
    \item the driver flow is given by $y_{1,2} = y_{2,1} = 1$, $y_{2,2} = 20$, and $y_{1,1} = 0$, and
    \item the trip prices are $p_{1,1} = 46.05$, $p_{1,2} = 92.10$ and $p_{2,1} = p_{2,2} = 0$.
  \end{itemize}

  \vspace{0.1em}

  Intuitively, %
  this outcome results from the platform's naively %
  ``clearing the market'' by origin, \ie setting a price rate %
  for each origin such that supply meets demand--- %
  location $2$ has excess supply thus is priced at zero; %
  location $1$ trips are priced at \$4.605 per minute (\ie $p_{1,1} = d_{1,1} \times 4.605 = 46.05$ and similarly $p_{1,2} = 92.10$); %
  riders are picked up if and only if their value is (weakly) above the price.
  The social welfare (\ie the total rider value minus driver cost) achieved under this outcome %
  is $\int_0^{x_{1, 2}} v_{1, 2}(s) \dl s + \int_0^{x_{2, 2}} v_{2, 2}(s) \dl s - 0 = 332.10$ dollars/min. This outcome is inefficient, since all riders for the $(2,2)$ trip are picked up, even those with values very close to zero, whereas many high-value $(1,2)$ riders did not get a ride.

  \subparagraph{Cycles and Driver Surplus}

  There are a total of three directed cycles in this economy: $\gC = \{\kappa_1, \kappa_2, \kappa_3\}$, where $\kappa_1 = ((1,1))$ $\kappa_2 = ((2,2))$ are ``self loops'' and  $\kappa_3 = ((1,2), (2,1))$ is the cycle between the two locations (note that cycles $((1, 2), (2, 1))$ and $((2, 1), (1, 2))$ %
  are the same). %
  The driver flow $\vy$ can be decomposed into cycles with weights $w_{\kappa_1} = 0$, $w_{\kappa_2} = 20$, and $w_{\kappa_3} = 1$. We can see that this outcome is very inequitable: those drivers on $\kappa_2$ (\ie fulfilling the $(2,2)$ trip) gets zero surplus, whereas the drivers on $\kappa_3$ looping around the two locations %
  gets a surplus of $(p_{1,2} + 0) / (d_{1,2} + d_{2,1}) = 2.30$ per minute.
\end{example}

\subsection{Optimal Outcome and %
  CE Prices} \label{sec:opt_welfare_and_CE}

We now formalize the optimization problem, characterize the structure of competitive equilibrium (CE) prices, and prove a welfare theorem, that an outcome is welfare-optimal if and only it is a CE.

\begin{definition}[Competitive Equilibrium (CE)] \label{def:CE}
  A CE is an outcome $(\vx, \vy, \vp)$ such that:
  \begin{enumerate}[label = (R\arabic*)]
    \item \label{cond:CE_rider_br} Riders are picked up if and only if their value is higher than the price: $\forall (i, j) \in \gL^2, \, x_{i, j} = q_{i, j}(p_{i, j})$.
  \end{enumerate}
  \begin{enumerate}[label = (D\arabic*)]%
    \item \label{cond:CE_relocation}
          Trips with relocating driver flow are priced at zero: $\forall (i, j) \in \gL^2, \ ( y_{i, j} > x_{i, j} \implies p_{i, j} = 0 )$.
    \item \label{cond:CE_exhaust_supply}
          All drivers are used-up, if there exist a cycle with a strictly positive %
          surplus rate:
          \begin{equation}
            \max_{\kappa \in \gC} \mu_\kappa > 0 \implies \sum_{i, j \in \gL} d_{i, j} y_{i, j} = m. \label{qu:CE_exhaust_supply}
          \end{equation}
    \item \label{cond:CE_do_nothing}
          No driver is dispatched, if the surplus rate is strictly negative for all cycles: %
          \begin{equation}
            \max_{\kappa \in \gC} \mu_\kappa < 0 \implies \forall (i, j) \in \gL^2, \, y_{i, j} = 0.
          \end{equation}
    \item \label{cond:cycle_max}
          The driver flow $\vy$ can be decomposed into cycles with weights $\vw \in \R_{\geq 0}^{|\gC|}$ such that every cycle with positive driver flow has the highest %
          surplus rate: %
          \begin{equation}
            \exists \vw, \ \forall \kappa \in \gC, \quad w_\kappa > 0 \implies \mu_\kappa = \max_{\kappa' \in \gC} \mu_{\kappa'}. \label{qu:CE_cycle_surplus}
          \end{equation}
  \end{enumerate}
\end{definition}

Intuitively, under a %
CE, all riders and drivers get their %
highest possible utility given the trip prices.
\ref{cond:CE_rider_br} can be interpreted as \emph{rider best-response}, \ie each rider decides whether she would like to be picked up given their trip price.
Conditions \ref{cond:CE_relocation}--\ref{cond:cycle_max} together represent \emph{driver best-response}, \ie each of the $m$ units of drivers is not worse-off than the scenario where she has the freedom to (i) choose any cycle and get paid for all trips on the cycle, and (ii) choose to not drive for the platform at all.%
\footnote{Technically, in condition \ref{cond:cycle_max}, it is equivalent if ``$\exists \vw$'' is replaced by ``$\forall \vw$''. This can be seen from the proof of \Cref{lem:welfare_theorem}. Moreover, for competitive equilibrium in two-sided markets, it is typically not necessary to require that the price be zero when supply exceeds demand, since a negative price implies a violation of the best-response on the supply side. %
  This is the case for our setting as well, but we keep condition \ref{cond:CE_relocation} for simplicity of notation and presentation.%
}

\smallskip

When the rider demand model $\vq = (q_{i, j}(\cdot))_{(i, j) \in \gL^2}$ (and thus the inverse functions $(v_{i, j}(\cdot))_{(i, j) \in \gL^2}$)
is known, the following convex program maximizes social welfare over rider and driver flows:
\begin{maxi!}
{\vx, \vy \in \R^{n^2}}{\sum_{i, j \in \gL} \biggl( \int_0^{x_{i, j}} v_{i, j}(s) \dl s - c_{i, j} y_{i, j} \biggr) \label{eq:primal-obj}}{\label{eq:primal}}{}
\addConstraint{x_{i, j}}{\le y_{i, j}, \quad \label{eq:x<=y}}{\forall i, j \in \gL}
\addConstraint{\sum_{i, j \in \gL} d_{i, j} y_{i, j}}{\le m \label{eq:supply-constraint}}
\addConstraint{\sum_{j \in \gL} y_{i, j}}{= \sum_{j \in \gL} y_{j, i}, \quad \label{eq:flow-balance}}{\forall i \in \gL.}
\end{maxi!}

The objective \eqref{eq:primal-obj} represents the total value of riders who are picked up minus the total cost incurred by the drivers, per unit of time. %
Given that $v_{i, j}(\cdot)$ is strictly decreasing for all $(i,j) \in \gL^2$, its integral is strictly concave, and thus the objective function \eqref{eq:primal-obj} is strictly concave. Constraints \eqref{eq:x<=y}--\eqref{eq:flow-balance} correspond to conditions \ref{cond:feasibility_demand_supply}--\ref{cond:feasibility_flow_balance}  under which the outcome is feasible and stationary.%
\footnote{Note that we did not impose non-negativity constraints on $\vx$ or $\vy$, or the constraint that $x\subij \leq q\subij(0)$ for all $(i,j)\in\gL$. This is because by construction, $v_{i, j}(s) = \infty$ for $s < 0$ and $v_{i, j}(s) = -\infty$ for $s > q_{i, j}(0)$, and thus these conditions are satisfied by any solution where the objective is not negative infinity.
}

\smallskip

Let $p_{i, j}$, $\omega$, and $\phi_i$ be the dual variables corresponding to constraints \eqref{eq:x<=y}, \eqref{eq:supply-constraint}, and \eqref{eq:flow-balance}, respectively. The dual problem of the welfare-optimization problem \eqref{eq:primal} is given as follows:
\begin{mini!}
{\vp \in \R^{n^2}, \, \omega \in \R, \, \vphi \in \R^n}{m \omega + \sum_{i, j \in \gL} \int_{0}^{q_{i, j}(p_{i, j})} (v_{i, j}(s) - p_{i, j}) \dl s %
\label{eq:dual_obj}}{\label{eq:dual}}{}
\addConstraint{p_{i, j}}{= c_{i, j} + d_{i, j} \omega + \phi_i - \phi_j, \quad \label{eq:dual_price_structure}}{\forall i, j \in \gL}
\addConstraint{p_{i, j}}{\ge 0, \quad}{\forall i, j \in \gL}
\addConstraint{\omega}{\ge 0. \label{eq:dual_omega>=0}}
\end{mini!}

\begin{restatable}[Welfare Theorem]{lemma}{lemWelfareThm} \label{lem:welfare_theorem}
  A feasible rider and driver flow $(\vx, \vy)$ is welfare-optimal if and only if there exist trip prices $\vp$ such that the outcome $(\vx, \vy, \vp)$ is a competitive equilibrium (CE).
  Moreover, for any CE outcome $(\vx, \vy, \vp)$, there exist $\omega%
    \ge 0$ and $\vphi %
    \in \R^n$ such that
  \begin{equation}
    p_{i, j} = c_{i, j} + d_{i, j} \omega %
    + \phi_i %
    - \phi_j%
    , \quad \forall i, j \in \gL. \label{eq:price_structure}
  \end{equation}
\end{restatable}

This structure of CE prices are aligned with the structure of optimal and incentive compatible prices in existing literature~\citep{ma2022spatio,cashore2022dynamic} (which we discussed in Section~\ref{sec:intro}).
We provide the proof of this Lemma in Appendix~\ref{appx:proof_welfare_thm}. Briefly, we use standard observations about the connection between duality and market equilibrium~\citep{shapley1971assignment,bertsekas1990auction,parkes2000iterative,ma2022spatio}. %
The dual variables $\vp$ and $\omega$ can be interpreted as trip prices and drivers' surplus rate, %
and the difference between the dual and the primal objectives corresponds to the total violation of the driver and rider best-response conditions \ref{cond:CE_rider_br}--\ref{cond:cycle_max}. %
An outcome is therefore welfare-optimal if and only if it is a competitive equilibrium.

Note that we work with alternative primal and dual formulations %
(see Appendix~\ref{appx:cycle_formulation}), which provide sufficient flexibility to map any %
set of prices $\vp$
to the dual variables.
Driver best-response, however, requires that drivers be indifferent towards all feasible paths between %
each pair of locations $(i, j) \in \gL^2$.
This allows us to prove that %
CE prices must have the structure as shown in \eqref{eq:price_structure}, consisting of three parts:%
\begin{itemize}
  \item $c\subij$: the cost incurred by the driver for completing each trip $(i,j) \in \gL^2$,
  \item $d_{i, j} \omega$: the opportunity cost of the %
        trip, which is multiplicative and scales with trip duration. %
        $\omega$ represents the marginal value of driver supply, and is referred to as the \emph{multiplier},
  \item $\phi_i - \phi_j$: the additive, origin-destination (OD) based adjustment, corresponding to the marginal value of a driver at the trip origin $i$ minus that of a driver
        at the trip destination $j$. Intuitively, this can  be interpreted as the value of ``teleporting'' a driver from $j$ back to $i$.
\end{itemize}

Note that there are $n-1$ instead of $n$ degrees of freedom in the OD-based adjustments, since shifting all $\phi_i$'s %
by the same constant has no impact on the dual problem \eqref{eq:dual}.
Technically, this is due to the fact that one of the flow-balance constraints \eqref{eq:flow-balance} is redundant, since the flow-balance of any $n-1$ locations imply that the remaining location is also balanced.
We prove in Proposition~\ref{prop:dual_uniqueness} in Appendix~\ref{appx:proof_dual_uniqueness} that the optimal dual solution is unique %
up to a constant shift in $\vphi$, and in the rest of this paper, we use $\omega^\ast$ and $\vphi^\ast$ to denote the
optimal multiplier and the optimal OD-based adjustments with $\phi_n^\ast = 0$.

\medskip

We now revisit the morning rush-hour example, %
and illustrate the welfare-optimal outcome.

\begin{examplecont}{ex:running-example} %
  Under the optimal outcome, the platform picks up (per minute) $x_{1,2} = 4$ units of riders going from the residential area to downtown, and $x_{2,2} = 8$ units of riders traveling within downtown.
  The optimal welfare is %
  approximately 459.91 dollars per minute, and
  the corresponding driver flow is $y_{1, 1} = 0$, $y_{1, 2} = y_{2, 1} = 4$, and $y_{2, 2} = 8$.
  Despite the fact that there is no demand from downtown to the residential area ($x_{2,1} = 0$), some drivers relocate without a rider ($y_{2, 1} = 4$) in order to satisfy the high-value demand on the opposite direction.

  \subparagraph{CE Prices}

  The CE prices are given by $p_{1,1} = p_{2,2} = 9.16$, $p_{1,2} = 36.65$ and $p_{2,1} = 0$, %
  corresponding to the optimal dual solution
  $\omega^\ast
    = 0.916$ (dollars per minute) and $\phi_1^\ast - \phi_2^\ast
    = 18.33$ (dollars).
    It is intuitive that $\phi_1^\ast> \phi_2^\ast$, \ie the marginal value of a driver in the residential area is higher than that in downtown.
  Adjusting trip prices by \emph{origin and destination} using $\vphi^\ast$,
  the platform charges $p_{2,2} = c_{2,2} + d_{2,2} \omega^\ast + \phi_2^\ast - \phi_2^\ast = 10\omega^\ast > 0$ (%
  so that the low-value downtown riders are not picked up),
  while setting %
  $p_{2,1} = c_{2,1} + d_{2,1}\omega^\ast + \phi_2^\ast - \phi_1^\ast %
    = 0$ %
  and relocate drivers from downtown %
  back to the residential area.

  \subparagraph{Equitable Driver Earnings}

  For drivers cycling between locations $1$ and $2$, each cycle takes $d_{1,2} + d_{2,1} = 40$ minutes, from which the drivers get a surplus of $(p_{1,2} + 0) / 40 = 0.916 = \omega^\ast$ per minute.
  The drivers who loop around location $2$ get the same surplus rate: each $(2,2)$ trip takes $d_{2,2} = 10$ minutes, and the surplus rate is also $p_{2,2} / 10 = 0.916 = \omega^\ast$.
\end{examplecont}

In practice, %
without access to the rider demand model, %
a platform is not able to directly compute %
$\vphi^\ast$, the marginal value of drivers by location.
Instead, under the outcome (discussed in Example~\ref{ex:running-example} earlier in this section) that results from the platforms' naively clearing the market by origin without any OD-based adjustments, the
platform achieves only 332.1/459.81 = 72\% of the optimal social welfare. %
Moreover, drivers earn substantially more from trips originating in the residential area in comparison to downtown.

This large discrepancy of earnings, however, suggests that drivers in the residential are more valuable than those in downtown (which is indeed the case, as $\phi_1^\ast > \phi_2^\ast$). It is therefore natural to consider trying some small adjustment, for example $\phi_1 - \phi_2
  = 1$, %
and then see what happens next week and move forward from there.
In next section, we characterize the outcome under  %
such ``surge pricing'' algorithms that clear the market %
using origin-based multipliers in the presence of additive OD-based adjustments. %

\section{Origin-Based Market Clearing} \label{sec:origin_based_surge}

In this section, we characterize the outcome under origin-based market-clearing with OD-based adjustments, \ie when the platform determines origin-based ``surge'' multipliers to balance supply and demand for each location.
We prove that the market-clearing outcome is unique when driver relocation is coordinated by the platform, and that the welfare-loss (in comparison to the optimal outcome discussed in Section~\ref{sec:preliminaries}) is bounded by the difference in the multipliers at different locations.
This characterization allows us to discuss in
Section~\ref{sec:INP} how a platform updates the OD-based adjustments week-over-week, based on the observed market-clearing outcomes in prior weeks.

\medskip

As we have shown in \Cref{lem:welfare_theorem}, given the optimal OD-based adjustments $\vphi^\ast$, a platform can clear the market and optimize welfare %
using a uniform %
multiplier $\omega^\ast$ for all locations.
In the absence of OD-based adjustment or when the adjustments are suboptimal, however, it is impossible to satisfy the best-response conditions for both drivers and riders. %
Instead, as in practice, the platforms %
determine \emph{origin-based multipliers} $\vpi = (\pi_i)_{i \in \gL}$ %
to balance supply and demand for each location, %
using prices of the form:
\begin{equation}
  p_{i, j} = c_{i, j} + d_{i, j} \pi_i + \phi_i - \phi_j, \quad \forall i, j \in \gL. \label{eq:price_structure_obmc}
\end{equation}

For instance, Example~\ref{ex:running-example} illustrates such an %
outcome for the morning rush hours with $\vphi = \vzero$.
With origin-based multipliers $\pi_1 = 4.605$ and $\pi_2 = 0$, (i) riders %
are picked up if and only if their values are above the price, and (ii) prices are non-negative but trips with excess supply are priced at zero.
When a platform's dynamic pricing algorithm sets
$\vpi$ to achieve the desired properties (i) and (ii), however, the resulting outcome may not be unique and will %
depend the actions of drivers upon arriving at a location without being assigned a rider trip.%
\footnote{For instance, %
  when a half of the drivers who are undispatched at location~$2$ drives to location~$1$ while the other half stay in location~$2$, %
  the market-clearing %
  multipliers are given by $\pi_1 = 5.99$ and $\pi_2 = 0$, and corresponding rider and driver flows are $x_{1,1} = x_{2,1} = 0$, $x_{1,2} = 0.5$, $x_{2,2} = 20$, and $y_{1,1}=0$, $y_{1,2} = y_{2,1} = 0.5$, $y_{2,2} = 22$. %
}
This %
lack of uniqueness in the potential market-clearing outcome imposes a challenge on the week-over-week update of the OD-based adjustments.

\paragraph{Coordinated Driver Relocation}
In this work, we assume that the platform coordinates %
the relocation of drivers who are not dispatched rider trips. This can be %
modeled using \emph{phantom demand functions} $\vqtilde = (\qtilde_{i, j}(\cdot))_{(i, j) \in \gL^2}$.
Specifically, %
if the price of a trip $(i,j) \in \gL$ is $p\subij \geq 0$, the platform dispatches $\qtilde_{i, j}(p_{i, j})$ units of drivers to relocate from $i$ to $j$ (per unit of time) without a rider, in addition to dispatching $q\subij (p_{i, j})$ drivers to rider trips. The platform has flexibility in choosing any phantom demand function $\vqtilde$ such that:
\begin{enumerate}[label = (\Roman*)]%
  \item \label{cond:phantom_fn} For all $i, j \in \gL$, $\qtilde_{i, j}(\cdot)$ is non-negative, %
        continuously differentiable, and monotonically decreasing.
  \item \label{cond:phantom_zero_price} For all $i, j \in \gL$, $d_{i, j} \qtilde_{i, j}(0) > m$.

\end{enumerate}

Condition \ref{cond:phantom_zero_price} guarantees that the %
phantom demand at zero price is sufficiently high for each trip, so that the platform %
will provide a relocation recommendation for every driver who is not dispatched to pick up a rider. %
When the %
prices are very
high, however, it is inefficient and inequitable to have drivers relocate without a rider (and as a result not getting paid %
for the trip). Define this ``slackness'' as of $\vqtilde$ as
\begin{equation}
  \phantomSlack_{i, j} \triangleq \sup_{r \ge 0} r \cdot \qtilde_{i, j}(r), \quad \forall i,j \in \gL, \label{eq:phantom_slackness}
\end{equation}
$\sum_{i, j \in \gL} \phantomSlack_{i, j}$ provides an upper bound of the %
welfare loss due to driver relocation at positive prices under $\vqtilde$.
Unless stated otherwise,
we assume that conditions \ref{cond:phantom_fn} and \ref{cond:phantom_zero_price} are satisfied %
in the rest of the paper. %

\begin{definition}[Origin-based market clearing] \label{def:MC}
  For any economy \economy, any phantom demand $\vqtilde$, and any OD-based additive adjustment $\vphi \in \R^n$, a set of origin-based multipliers $\vpi \in \R^n$ is \emph{market clearing} if there exists a feasible outcome $(\vx, \vy, \vp)$ with prices given by \eqref{eq:price_structure_obmc} such that:
  \begin{enumerate}[label = (MC\arabic*)]
    \item \label{cond:mc_rider_reloc_flow} %
          Riders best-respond to the trip prices (\ie $x_{i, j} = q_{i, j}(p_{i, j})$ for all %
          $(i, j) \in \gL^2$), and the driver relocation %
          is coordinated by the phantom demand (\ie $y_{i, j} - x_{i, j} = \qtilde_{i, j}(p_{i, j})$ for all $(i, j) \in \gL^2$).
    \item \label{cond:mc_total_supply}
          All drivers are dispatched, \ie $\sum_{i, j \in \gL} d_{i, j} y_{i, j} = m$.
  \end{enumerate}
\end{definition}

Note that by defining the origin-based market-clearing outcome, we effectively postulate capabilities and behaviors of the platform's origin-based dynamic pricing algorithm.
Condition \ref{cond:mc_total_supply}, for instance, implies that the platform dispatches all drivers who are on the road to either a rider or a relocation trip, and does not recommend any driver to stop driving and go offline.
Among the conditions for competitive equilibria, rider best-response \ref{cond:CE_rider_br} is still satisfied %
since it is imperative for the ridesharing platforms to guarantee reliable service for riders (see Footnote~\ref{fn:reliability_for_riders}). %
In contrast, the driver best-response conditions are relaxed, mirroring the current state of platforms, as without optimal OD-based adjustments a platform cannot ensure best-response for both sides of the market. %

\begin{restatable}[Existence and Uniqueness]
  {lemma}{lemUniqueness} \label{lem:unique}
  For any economy \economy, any phantom demand function $\vqtilde$, %
  and any OD-based price adjustment $\vphi \in \R^n$, there exists a unique %
  vector of origin-based multipliers $\vpi \in \R^n$ that clears the market.
\end{restatable}

We provide the proof of this lemma in Appendix~\ref{appx:proof_lem_uniqueness}.
Briefly, fixing the economy and the OD-based adjustments, the flow balance constraints (only $n-1$ of them are ``independent'', as discussed in Section~\ref{sec:preliminaries}) together with the condition on total supply \ref{cond:mc_total_supply} impose $n$ non-linear constraints on the origin-based multipliers $\vpi \in \R^n$.
By exploiting the structure of %
the flow-balance constraints, we show that the set of multipliers that induce balanced driver flows form a one-dimensional manifold, within which exactly one set of multipliers lead to an outcome that dispatch exactly $m$ units of drivers.

By Definition~\ref{def:MC} and Lemma~\ref{lem:unique}, we know that the origin-based multipliers $\vpi$ and %
the corresponding market-clearing outcome $(\vx, \vy, \vp)$ are fully determined by the OD-based adjustments $\vphi$. %
We refer to %
$(\vx, \vy, \vp, \vphi, \vpi)$ as the \emph{origin-based market-clearing outcome} associated with OD-based adjustment $\vphi \in \R^n$.

\medskip

One challenge for analyzing the social welfare achieved by an orogin-based market-clearing outcome using primal and dual formulations \eqref{eq:primal}--\eqref{eq:dual} is that the dual variables satisfying constraint \eqref{eq:dual_price_structure} does not provide sufficient flexibility to represent the trip prices \eqref{eq:price_structure_obmc} (which now have more than $n$ degrees of freedom).
As a result, we use the following equivalent primal formulation with additional constraints. %
\begin{maxi!}
{\vx, \vy \in \R^{n^2}, \, \vz \in \R^n}{\sum_{i, j \in \gL} \biggl( \int_0^{x_{i, j}} v_{i, j}(s) \dl s - c_{i, j} y_{i, j} \biggr) \label{eq:primal_obj_origin_based}}{\label{eq:primal_origin_based}}{}
\addConstraint{x_{i, j}}{\le y_{i, j}, \quad}{\forall i, j \in \gL}
\addConstraint{\sum_{j \in \gL} d_{i, j} y_{i, j}}{= z_i, \quad \label{eq:supply-local}}{\forall i \in \gL}
\addConstraint{\sum_{i \in \gL} z_i}{\le m \label{eq:supply-total}}
\addConstraint{\sum_{j \in \gL} y_{k, j}}{= \sum_{i \in \gL} y_{i, k}, \quad}{\forall k \in \gL.}
\end{maxi!}

Intuitively, replacing the supply constraint \eqref{eq:supply-constraint} with a \emph{local} supply constraint \eqref{eq:supply-local} in addition to the \emph{global} supply constraint \eqref{eq:supply-total}, the resulting dual problem below now %
includes additional variables $\pi_i$ (corresponding to  \eqref{eq:supply-local}) that can represent origin-based price multipliers:
\begin{mini!}
{\vp \in \R^{n^2}, \, \omega \in \R, \, \vpi \in \R^n, \, \vphi \in \R^n}{m \omega + \sum_{i, j \in \gL} \int_{0}^{q_{i, j}(p_{i, j})} (v_{i, j}(s) - p_{i, j}) \dl s %
\label{eq:dual_obj_origin_based}}{\label{eq:dual_origin_based}}{}
\addConstraint{p_{i, j}}{= c_{i, j} + d_{i, j} \pi_i + \phi_i - \phi_j, \quad}{\forall i, j \in \gL \label{eq:dual_price_cnst}}
\addConstraint{p_{i, j}}{\ge 0, \quad}{\forall i, j \in \gL \label{eq:dual_origin_based_price_nonneg}}
\addConstraint{\omega}{\ge \pi_i, \quad}{\forall i \in \gL \label{eq:dual_oribin_based_omega_pi}}
\addConstraint{\omega}{\ge 0.\label{eq:dual_oribin_based_omega_nonneg}}
\end{mini!}

For each location $i \in \gL$, $\pi_i$ can be interpreted as the marginal value of a unit of drivers' time for drivers who are on trip originating from location $i$. Intuitively, when the origin-based multipliers are higher for some locations than the others, or when some of the multipliers are negative, we are making inefficient use of drivers' time.
In the following result, we establish an upper bound on the welfare suboptimality of any origin-based market-clearing outcome (\ie the difference between the achieved and the highest possible social welfare), which consists of this inefficiency as well as the slack introduced by driver relocation.

\begin{restatable}{theorem}{thmOptimality} \label{thm:optimality}
  Given any %
  origin-based market-clearing outcome
  $(\vx, \vy, \vp, \vphi, \vpi)$,
  the welfare suboptimality %
  is upper bounded by
  \begin{equation} \label{eq:primal_dual_diff_at_mc}
    \sum_{i, j \in \gL} d_{i, j} y_{i, j} \Bigl( \max\Bigl\{ \max_{k \in \gL} \pi_k, \, 0 \Bigr\} - \pi_i \Bigr) + \sum_{i, j \in \gL} p_{i, j} (y_{i, j} - x_{i, j}),
  \end{equation}
  which is further upper bounded by
  \begin{equation} \label{eq:upper_bound_welfare_gap}
    m \Bigl( \max \Bigl\{\max_{k \in \gL} \pi_k, \, 0 \Bigr\} - \min_{i \in \gL} \pi_i \Bigr)
    + \sum_{i, j \in \gL} \phantomSlack_{i, j}.
  \end{equation}
\end{restatable}

Importantly, the bounds \eqref{eq:primal_dual_diff_at_mc} and \eqref{eq:upper_bound_welfare_gap} are computable using only information immediately available from a current market-clearing outcome: origin-based multipliers, trip prices, and driver and rider flows.
This is of practical relevance as it allows a platform to assess the welfare suboptimality even when the overall social welfare cannot be determined due to the absence of a demand model.

\begin{proof}%
  Given any origin-based market-clearing outcome $(\vx, \vy, \vp, \vpi, \vphi)$, let $z_i = \sum_{j \in \gL} d_{i, j} y_{i, j}$ for all $i \in \gL$, and denote $\omega = \max\left\lbrace \max_{i \in \gL} \{\pi_i\}, \, 0 \right\rbrace$.
  Intuitively, $z_i$ is the amount of drivers (at any moment of time) who are on trip from each location $i \in \gL$, and $\omega$ represents the ``best-response'' driver surplus rate (\ie when the driver has the freedom to choose any route or to not drive for the platform).
  It is straightforward to verify that $(\vx, \vy, \vz)$ and $(\vp, \omega, \vpi, \vphi)$ form feasible solutions of the primal \eqref{eq:primal_origin_based} and the dual \eqref{eq:dual_origin_based}, respectively.

  The primal objective \eqref{eq:primal_obj_origin_based} represents welfare, which is equal to the total driver and rider surplus under the market-clearing outcome.
  The dual objective \eqref{eq:dual_obj_origin_based}, on the other hand, can be interpreted as the total ``best response'' surplus of drives and riders,
  \ie every driver gets a surplus rate of $\omega$ per unit of time, and all riders who are willing to pay the price of their trips are picked-up.
  The difference %
  \begin{equation}
    \sum_{i, j \in \gL} \int_{x_{i, j}}^{q_{i, j}(p_{i, j})} (v_{i, j}(s) - p_{i, j}) \dl s
    + \sum_{i, j \in \gL} p_{i, j} (y_{i, j} - x_{i, j})
    + \sum_{i \in \gL} z_i (\omega - \pi_i)
    + \omega \Bigl(m - \sum_{i \in \gL} z_i \Bigr), \label{eq:duality_gap}
  \end{equation}
  and %
  can %
  is therefore the total violation of driver and rider best-response conditions:
  \begin{itemize}%
    \item The first term $\sum_{i, j \in \gL} \int_{x_{i, j}}^{q_{i, j}(p_{i, j})} (v_{i, j}(s) - p_{i, j}) \dl s$ represents the violation of rider best-response: $q_{i, j}(p_{i, j})$ riders who like to be picked-up for trip $(i,j)\in \gL^2$ when the price is given by $p\subij$, but only $x\subij$ of them are picked-up and the rest each looses $v_{i, j}(s) - p_{i, j}$ in comparison to their ``best-response'' outcome.
          Given \ref{cond:mc_rider_reloc_flow}, we know that this term is zero under any %
          market-clearing outcome.
    \item The second term $\sum_{i, j \in \gL} p_{i, j} (y_{i, j} - x_{i, j})$ is the violation of driver best-response due to drivers being dispatched to relocate without getting paid when prices are positive. This can be rewritten as $\sum_{i, j \in \gL} p_{i, j} \qtilde\subij(p_{i, j})$, which is bounded by $\sum_{i, j \in \gL} \phantomSlack_{i, j}$ by definition of the slack \eqref{eq:phantom_slackness}.
    \item The third term $\sum_{i \in \gL} z_i (\omega - \pi_i)$ is the violation of driver best-response due to the fact that
          at any moment in time $z_i$ drivers are \emph{on-trip} originating from location $i \in \gL$ but the the surplus rate $\pi_i$ these drivers get %
          may be lower than the highest possible surplus rate $\omega$.
    \item The last term $\omega \Bigl(m - \sum_{i \in \gL} z_i \Bigr)$ represents the violation of driver best-response due to the fact that some drivers may not be dispatched even when %
          some other drivers are getting a positive surplus of  $\omega$ per unit of time. This terms is always zero under market-clearing outcomes due to \ref{cond:mc_total_supply}.
  \end{itemize}

  The %
  welfare-suboptimality of the current market-clearing outcome (\ie the current primal objective) is bounded by the difference between the primal and dual objectives \eqref{eq:duality_gap}, since the dual objective is always weakly above the optimal primal objective.
  Combining the non-zero (\ie the 2nd and 3rd) terms of \eqref{eq:duality_gap} gives us \eqref{eq:primal_dual_diff_at_mc}, which we can see is upper bounded by \eqref{eq:upper_bound_welfare_gap}.
  \end{proof}

We now revisit the morning rush-hour example discussed in Section~\ref{sec:preliminaries}, %
and illustrate the origin-based market-clearing outcomes at different OD-based adjustments.

\begin{examplecont}{ex:running-example}

  Recall that in the morning rush hour example, many riders request trips to downtown (location~$2$) while no rider request any trip ending in the residential area (location~$1$).
  To coordinate driver relocation, %
  the platform adopts phantom demand function $\qtilde_{i, j}(r) = 24 (\max\{0, 1 - r / 5\})^4$ for all $i, j \in \{1, 2\}$, in which case relocation occurs only when the trip price falls below \$5.
  Fixing $\phi_2 = 0$, \Cref{fig:running_example_ODB} illustrates the origin-based market-clearing outcomes as we vary $\phi_1$. %

  \begin{figure}[t!]
    \centering
    \subcaptionbox{
      Multipliers.
      \label{fig:ex_surge_pi}}{\includegraphics[scale = \MmaScale]{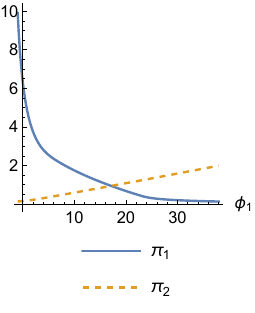}}
    \hfill
    \subcaptionbox{Trip prices. \label{fig:ex_surge_price}}{\includegraphics[scale = \MmaScale]{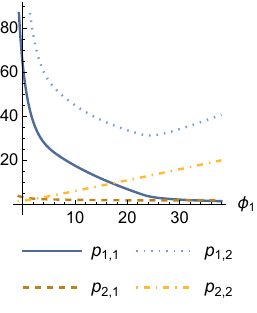}}
    \hfill
    \subcaptionbox{Rider flow rates. \label{fig:ex_surge_x}}{\includegraphics[scale = \MmaScale]{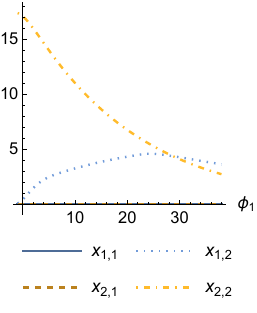}}
    \hfill
    \subcaptionbox{
      Objectives.
      \label{fig:ex_surge_welfare}}{\includegraphics[scale = \MmaScale]{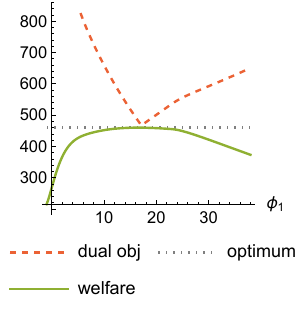}}
    \caption{
      The origin-based market clearing outcome for different $\phi_1$ in \Cref{ex:running-example} (fixing $\phi_2 = 0$). %
    }
    \label{fig:running_example_ODB}
  \end{figure}

  When $\phi_1=0$, which is the case of the na\"ive origin-based pricing, $\pi_1 \gg \pi_2$, meaning that the origin-based multiplier for location~$1$ has to be substantially higher than that at location~$2$ to due to the supply demand imbalance. %
  The resulting %
  outcome is very inefficient, since drivers spend time serving many $(2,2)$ riders with very low values, %
  while many very high value $(1,2)$ riders were not picked up.

  As $\phi_1$ increases, $p_{2,1}$ decreases relative to $p_{2,2}$. This leads to higher $x_{1,2}$ and lower $x_{2,2}$, since more drivers at location $2$ %
  relocate to location $1$ to pick-up the $(1,2)$ riders (the relocation flow $y_{2,1} = \qtilde(p_{2,1})$ is not shown in Figure~\ref{fig:running_example_ODB} but is equal to $x_{1,2}$ for small $\phi_1$),
  instead of pick-up the lower value riders at location $2$. This reduces the gap in the %
  multipliers at the two locations, and improves welfare.

  Starting from the market-clearing outcome with $\phi_1 = 0$, $\pi_1 \gg \pi_2$ indicates that the marginal value of drivers is higher at location~$1$ than that at location $2$.
  It is intuitive that the platform should increase $\phi_1$ until the %
  $\pi_1-\pi_2$ is sufficiently small, in which case welfare is approximately optimal (Theorem~\ref{thm:optimality}).

\end{examplecont}

With more than two locations, however, it is unclear what would be a good direction for updating the OD-based adjustment $\vphi$, how big the stepsizes should be, and whether the process will converge to the optimal outcome (or at all). We answer these questions in the next section.

\section{Iterative Network Pricing} \label{sec:INP}

In this section, we introduce the Iterative Network Pricing (INP) mechanism, which %
observes the market-clearing outcomes during a specific time window (\eg Wed. morning rush hour) in preceding timesteps, and updates the OD-based additive price adjustments to be adopted in this time window for the upcoming timestep.%
\footnote{%
  A platform may adopt multiple INP mechanisms in parallel to update prices for other time windows with different demand patterns, \eg the weekday evening rush hours, or the weekend bar hours.%
  The timesteps represent the frequency at which the platform updates the OD-based adjustments.
  When a timestep corresponds to a week, for example, the platform determines the OD-based adjustments for Wednesday morning rush hour each week using market-clearing outcomes from the prior weeks during the same time window.
}
We prove the main result of this paper, that when the market condition is stationary over time, %
the INP mechanism converges to an outcome that is approximately welfare-optimal. %

\medskip

We first define a general iterative pricing mechanism. Let $t = 0, 1, \dots$ be the index of timesteps. Time $t=0$ represents the status quo with $\vphi= \vzero$, \ie the platform %
clears the market by origin without any OD-based adjustments.
Time $t=1$ and $t=2$ represent the first and second timesteps (\eg weeks), respectively, for the mechanism to start updating the OD-based adjustments, and so forth.

We make the following assumptions:
\begin{enumerate}[label = (A\arabic*)]
  \item \label{assn:stationarity} The economy is stationary %
        over time, \ie %
        $(m, \vd, \vc, \vq)$ remains the same for all $t \geq 0$.
        The phantom demand $\vqtilde$ used by the platform to %
        coordiante relocation also stays the same for all $t \ge 0$.%
        \footnote{ %
          Assumption \ref{assn:stationarity} is used to establish the optimality of the INP mechanism in theory. The mechanism is well defined %
          when the market condition is not stationary (see Section~\ref{sec:nonstationary_simulations} for simulation results for Chicago morning rush hours, 2019--2020).
        }
  \item \label{assn:market-clearance} During each timestep $t \geq 0$, the platform is able to
        observe the origin-based market-clearing outcome $(\vx\supt, \vy\supt, \vp\supt, \vphi\supt, \vpi\supt)$ corresponding to OD-based adjustments $\vphi\supt$, as well as the (local) price sensitivity $(q'\subij(p \subij \supt))_{(i,j)\in \gL^2}$ at the current prices. %
\end{enumerate}

We denote the \emph{history} %
up to time $t$ as $h_t = ((\vx\suptprime, \vy\suptprime, \vp\suptprime, \vphi\suptprime, \vpi\suptprime, (q'\subij(p \subij \suptprime))_{(i,j)\in \gL^2} ))_{t' \leq t-1}$, and denote the set of all possible histories under an iterative pricing mechanism up to time $t$ as $\historySett$.

\begin{definition}[Iterative Pricing Mechanisms] \label{def:iterative_pricing}
  An iterative pricing mechanism starts with no OD-based additive adjustment $\vphi^{(0)} = \vzero$ at %
  time $t = 0$.
  At each time $t \ge 1$, the mechanism determines %
  $\vphi^{(t)}$ based on the observable market-condition $(m, \vd, \vc)$ %
  and the history of past market-clearing outcomes $\historyt$. %
  \end{definition}

Given history $\historyt$, a seemingly straightforward %
method for updating the OD-based adjustment $\vphi^{(t)}$ is to use the gradient of the dual objective.
Recall from %
Section~\ref{sec:origin_based_surge} that given %
any %
$\vphi$ and the corresponding market-clearing $\vpi$, we can construct a feasible solution to the dual problem \eqref{eq:dual_origin_based}: %
$\vp$ is determined by constraint
\eqref{eq:dual_price_cnst}, and %
$\omega = \max\{\max_{i\in \gL} \pi_i, \, \ab 0\}$ achieves the smallest dual objective once $(\vp, \vpi, \vphi)$ is fixed.
In this way, we can rewrite the dual objective \eqref{eq:dual_obj_origin_based} %
as
\begin{equation}
  \mathsf{Dual}(\vphi, \vpi) \triangleq m \max \Bigl\{ \max_{i\in \gL} \pi_i, \, 0 \Bigr\} + \sum_{i, j \in \gL} \int_{c_{i, j} + d_{i, j} \pi_i + \phi_i - \phi_j}^\infty q_{i, j}(r) \dl r. \label{eq:dual_obj_as_function_of_phi_pi}
\end{equation}
It is straightforward to verify that the gradient (or a subgradient, at non-differentiable points) of $\mathsf{Dual}$ \wrt $\vpi$ and $\vphi$ %
can be computed using the trip prices, distances, and rider %
flow at current prices, %
quantities that are immediately observable by the platform at the current market-clearing outcome.%
\footnote{This is similar in spirit to many other settings such as that of iterative combinatorial auctions~\citep{parkes2000iterative}, where the market designer can is able to identity a direction to update the dual variables that is guaranteed to improve the dual objective, despite the fact that neither the primal nor the dual objective can be evaluated.}
However, this natural idea of %
gradient descent \wrt %
$(\vphi, \vpi)$ %
does not %
solve our problem for two reasons:
\begin{itemize} %
  \item First, $\nabla_\vphi \mathsf{Dual}(\vphi, \vpi) = \vzero$ at any outcome where the rider flow is balanced, in which case the gradient does not provide a %
        direction for updating $\vphi$, even when the outcome is far from optimal. %
          \item Moreover, while the gradient \wrt the origin-based multipliers $\vpi$ %
        is non-zero, %
        $\vpi$ is the result of the market-clearing process, %
        and is not determined by the platform ahead of time.
\end{itemize}

Since $\vphi$ is the only dual variable directly controlled by the platform,
a first challenge a %
mechanism needs to address is to ``anticipate'' how the market-clearing $\vpi$ will change %
as the platform updates $\vphi$.
The following result shows that this is possible using only information immediately available to the platform at any %
market-clearing outcome.
Recall Lemma~\ref{lem:unique}, that there exists a unique market-clearing %
$\vpi \in \R^n$ given any %
$\vphi \in \R^n$.
Since shifting all $\phi_i$'s by the same constant has no impact (see %
Section~\ref{sec:preliminaries}), we fix $\phi_n = 0$ for the rest of the paper, and treat $\vphi \in \R^{n-1}$ as a vector of $n - 1$ dimensions.
In this way, we denote $\vPi: \R^{n-1} \to \R^n$ as the mapping from the OD-based adjustments to the corresponding market-clearing multipliers.

\begin{restatable}{lemma}{thmDifferentiable} \label{thm:differentiable}
  The mapping $\vPi: \R^{n-1} \to \R^n$ is continuously differentiable. The Jacobian matrix at %
  the time $t-1$ market-clearing outcome $\D\vPi(\vphi^{(t-1)}) \in \R^{n \times (n - 1)}$ can be computed using history $\historyt$.  %
\end{restatable}

The proof of this lemma is provided in Appendix~\ref{appx:proof_thm_differentiable}. Briefly, %
the market-clearing conditions (Definition~\ref{def:MC})  provide $n$ equality constraints on the OD based adjustments $\vphi$ and the corresponding market-clearing $\vpi$ (which implicitly define the mapping $\vPi$ from $\vphi$ to $\vpi$).
We show that the violation of these $n$ constraints (which is a $n$-dimensional function) is continuously differentiable, and that its Jacobian matrix \wrt $\vpi$ is invertable. The implicit function theorem %
then implies that $\vPi$ is continuously differentiable. It is straightforward to verify that its Jacobian matrix depends only on the trip durations $(d\subij)_{(i,j) \in \gL^2}$ and the local slope of the demand function $(q_{i, j}'(p_{i, j}))_{(i,j) \in \gL^2}$, both of which
can be observed by the platform.

\medskip

With Lemma~\ref{thm:differentiable},
we can consider the dual objective \eqref{eq:dual_obj_as_function_of_phi_pi} as a function of only the OD-based adjustments $\mathsf{Dual}(\vphi, \vPi(\vphi))$.
The function is, however, not convex in $\vphi$.
Figure~\ref{fig:ex_surge_welfare} illustrates this %
for an economy %
with two locations. Moreover, Example~\ref{exmp:dual_obj_non_convex} in Appendix~\ref{appx:additional_discussions} shows that %
quasi-convexity does not hold either.
Consequently, we need to address the challenge of finding the global optimum using only local information, and in the absence of convexity.

\subsection{The Iterative Network Pricing Mechanism} \label{sec:INP_definition}

We define the INP mechanism by providing the direction towards which it updates the OD-based adjustments, the stepsize it takes along that direction, as well as a policy for ``backtracking''.

\paragraph{The Direction}

Recall from Theorem~\ref{thm:optimality} that the inefficiency of a %
market-clearing outcome %
is bounded by the differences in the origin-based multipliers at different locations. %
We prove that at any market-clearing outcome, there exists a unique direction %
for equalizing the (linear approximations of the) origin-based multipliers $\{\Pi_i(\vphi)\}_{i \in \gL}$, and we choose this direction for updating the OD-based adjustments $\vphi$ under the INP mechanism.

Formally, let $\vonen$ be the all-one vector of dimension $n$, and denote $\vchi_{1:n-1} \in \R^{n-1}$ as the first $n-1$ entries of any vector $\vchi \in \R^n$. %
Given any %
OD-based adjustments $\vphi \in \R^{n-1}$, %
the direction for updating $\vphi$ is given by
\begin{equation}
  \vdelta = \bigl( \begin{bmatrix} -\D\vPi(\vphi) & \vonen \end{bmatrix}^{-1} \vpi \bigr)_{1:n-1}, \label{eq:INP_direction_no_t}
\end{equation}
where $\vpi = \vPi(\vphi)$ is the corresponding market-clearing multipliers, $\D\vPi(\vphi) \in \R^{n \times (n - 1)}$ is the Jacobian matrix of $\vPi$ evaluated at the market-clearing outcome, and the %
matrix $\begin{bmatrix} -\D\vPi(\vphi) & \vonen \end{bmatrix}$ is guaranteed to be full-rank (see \Cref{clm:full_rank} in Appendix~\ref{appx:the_update_direction}).
In this way, there exists %
$\xi \in \R$ \st $ \vpi + \D\vPi(\vphi) \vdelta = \xi \vone_n$.%
\footnote{To see this, observe that $-\D\vPi(\vphi) \vdelta + \xi \vonen$ can be written as $\begin{bmatrix} -\D\vPi(\vphi) & \vonen \end{bmatrix} \begin{bmatrix} \vdelta \\ \xi \end{bmatrix}$.}
Note that $\vdelta$ %
corresponds to the update direction under Newton's method for solving the following system of $n-1$ non-linear equations:
\begin{equation}
  \Pi_1(\vphi) = \Pi_2(\vphi) = \cdots = \Pi_n(\vphi). \label{eq:equalize_pi}
\end{equation}

\paragraph{The Stepsize}
It is desirable in practice to avoid volatile %
changes in trip prices and driver earning rates, since as we have discussed in Section~\ref{sec:intro}, both drivers and riders learn and adapt slowly to changes in the marketplace.%
Given any $\vphi \in \R^{n-1}$ and the corresponding update direction $\vdelta$, we choose stepsize
\begin{align}
  \alpha = \min \{ 1, \, \tau / \|\D\vPi(\vphi) \vdelta \|_\infty \},
  \label{eq:INP_stepsize_no_t}
\end{align}
where $\tau > 0$ is a constant and $\|\cdot \|_\infty$ is the $\ell^\infty$ norm of vectors. %
In this way, based on the current linear approximation of %
$\vPi$, %
the origin-based multiplier $\pi_i$ for each location $i \in \gL$ will change by at most $\tau$ under the new market-clearing outcome at $\vphi + \alpha \vdelta $. %

Moreover, define the following auxiliary function (also known as a Lyapunov function) representing the sum of the squared differences between the origin-based multipliers and their mean:
\begin{equation} \label{eq:quadratic_objective}
  f(\vphi)
  \triangleq
  \sum_{i \in \gL} \Bigl( \Pi_i(\vphi) - \frac{1}{n} \sum_{j \in \gL} \Pi_j(\vphi) \Bigr)^2.
\end{equation}
To guarantee convergence in theory, our mechanism uses $f$ for \emph{backtracking line search}, \ie repeatedly shrinks the step size by a constant factor until
a sufficient decrease of %
$f$ is observed, relative to the expected amount of decrease based on the step size and the local gradient of %
$f$. The gradient of $f$ at any market-clearing outcome can be evaluated using the origin-based multipliers and the Jacobian of $\vPi$.

\medskip

We now define the INP mechanism, which updates the OD-based adjustments in the direction for equalizing the origin-based multipliers, chooses the stepsize to target a maximum change of the multipliers,
and reduces the stepsize if necessary via backtracking line search using the Lyapunov function $f$.

\begin{definition}[Iterative Network Pricing] \label{def:INP}
  An \emph{iterative network pricing} (INP) mechanism is an iterative mechanism %
  parameterized by %
  $\tau > 0$, $\beta \in (0, 1)$, $\sigma \in (0, 1)$. %
  At each time $t > 0$, given history $\historyt$:
  \begin{itemize} %
    \item The mechanism determines $\vphi \supt = \vphi \suptmo + \alpha \supt \vdelta \supt$ according to %
          \begin{align}
            \vdelta^{(t)} & = \bigl( \begin{bmatrix} -\D\vPi(\vphi\suptmo) & \vonen \end{bmatrix}^{-1} \vpi^{(t-1)} \bigr)_{1:n-1}, \label{eq:INP_direction} \\
            \alpha^{(t)}  & = \min \{ 1, \, \tau / \|\D\vPi(\vphi\suptmo) \vdelta^{(t)} \|_\infty \}, \label{eq:INP_stepsize}
          \end{align}
          if $t = 1$ or sufficient progress in $f$ was made (\ie if $f(\vphi^{(t-1)}) < f(\vphi^{(t')}) + \sigma \nabla f(\vphi^{(t')}) ^\trp %
            \alpha^{(t-1)} \vdelta^{(t-1)}$ %
          where $t'$ %
          is the timestep prior to the most-recent update of direction according to \eqref{eq:INP_direction}).
    \item If $t > 1$ and not enough progress in $f$ was made (\ie if
          $f(\vphi^{(t-1)}) \geq f(\vphi^{(t')}) + \sigma \nabla f(\vphi^{(t')}) ^\trp \alpha^{(t-1)} \vdelta^{(t-1)}$), the mechanism  backtracks, setting
          $\vphi \supt = \vphi \suptprime + \alpha \supt \vdelta^{(t)}$ with $\vdelta^{(t)} = \vdelta^{(t-1)}$ and $\alpha \supt = \beta \alpha\suptmo$.  %
  \end{itemize}
\end{definition}

We now state the main result of the present paper, that the INP mechanism converges to an outcome that is approximately welfare-optimal.

\begin{restatable}{theorem}{thmINP} \label{thm:INP}
  Assuming \ref{assn:stationarity} and \ref{assn:market-clearance}, given any economy, the iterative network pricing mechanism converges %
  to an origin-based market-clearing outcome where all multipliers are equal, \ie $\exists \omega^\star \in \sR$ \st $\lim_{t \to \infty} \pi_i\supt = \omega^\star$ for all $i \in \gL$. %
  The welfare %
  suboptimality of the limit outcome %
  is upper bounded by
  \begin{equation}
    m \max\{0, -\omega^\star\} + \sum_{i, j \in \gL} e_{i, j}. \label{eq:suboptimality_thm_INP}
  \end{equation}
  \end{restatable}

Intuitively, we prove that under assumptions \ref{assn:stationarity} and \ref{assn:market-clearance}, the INP mechanism can be interpreted as an iterative algorithm that is guaranteed to converge to the unique $\vphi^\star \in \sR^{n-1}$ for which \eqref{eq:equalize_pi} holds.
When the platform is under-supplied,  %
$\omega^\star \geq 0$ %
and the only inefficiency in the limit outcome is due to drivers' relocating when prices are strictly positive (see \eqref{eq:phantom_slackness}).
When the platform is over-supplied, it is possible for %
$\omega^\star$ to be strictly negative %
if all $m$ units of drivers continue to show-up despite the fact that they all make a negative surplus.
If we assume that drivers %
gradually leave the platform if the surplus %
is consistently negative, the suboptimality will again converge to $\sum_{i, j \in \gL} e_{i, j}$.

\begin{proof}
  First, observe that fixing any economy \economy, the phantom demand function $\vqtilde$ satisfying conditions \ref{cond:phantom_fn} and \ref{cond:phantom_zero_price}, and assuming \ref{assn:stationarity} and \ref{assn:market-clearance},
  the INP mechanism %
  can be interpreted as an iterative algorithm for solving
  \eqref{eq:equalize_pi}, which we formally describe in Algorithm~\ref{alg:Newton_Armijo}.

  \begin{algorithm}[htb]
    \caption{Iterative Network Pricing %
    } \label{alg:Newton_Armijo}
    \begin{algorithmic}[1]
      \Require $\tau > 0$, $\beta \in (0, 1)$, $\sigma \in (0, 1)$. %
      \State $\vphi^{(0)} \gets \vzero$ %
      \State $\vpi^{(0)} \gets \vPi(\vphi^{(0)})$ \label{line:market_clears_0} %
      \For{$t = 1, 2, \dots$}
      \If{$t = 1$ \textbf{or} (%
      $f(\vphi^{(t-1)}) < f(\vphi^{(t')}) + \sigma \nabla f(\vphi^{(t')}) ^\trp \alpha^{(t-1)} \vdelta^{(t-1)}$)}
      \State $\vdelta^{(t)} \gets \bigl( \begin{bmatrix} -\D\vPi(\vphi^{(t-1)}) & \vone \end{bmatrix}^{-1} \vpi^{(t-1)} \bigr)_{1:n-1}$
      \State $\alpha^{(t)} \gets \min \{ 1, \, \tau / \|\D \vPi(\vphi^{(t-1)}) \vdelta^{(t)} \|_\infty \}$ %
      \State $t' \gets t - 1$ \label{line:great_step} %
      \Else
      \State $\vdelta^{(t)} \gets \vdelta^{(t-1)}$
      \State $\alpha^{(t)} \gets \beta \alpha^{(t-1)}$
      \EndIf
      \State $\vphi^{(t)} \gets \vphi^{(t')} + \alpha^{(t)} \vdelta^{(t)}$
      \State $\vpi^{(t)} \gets \vPi(\vphi^{(t)})$  \label{line:market_clears}
      \EndFor
    \end{algorithmic}
  \end{algorithm}

  Lines~\ref{line:market_clears_0} and~\ref{line:market_clears} in Algorithm~\ref{alg:Newton_Armijo}  %
  correspond to the platform observing the %
  market-clearing outcome given the OD-based adjustments for each timestep. The observed information is then used to evaluate $\nabla f$ and $\D \vPi$.
  Line~\ref{line:great_step} marks the timesteps at which a sufficiently small stepsize for the current update direction is found, and the market-clearing outcomes at these timesteps are used to compute the new direction %
  in the following timestep.

  By constructing another auxiliary optimization problem and its dual problem, we prove in Lemma~\ref{lem:alternative_obj} in Appendix~\ref{appx:proof_lem_alternative_obj} that there exists a unique set of OD-based adjustments $\vphi^\star \in \R^{n-1}$ solving \eqref{eq:equalize_pi}, at which $\Pi_i(\vphi^\star) = \omega^\star$ for all $i \in \gL$. %
  Once we establish that the %
  algorithm converges to $\vphi^\star$, %
  the welfare suboptimality is a straightforward application of Theorem~\ref{thm:optimality}. What is left to prove is that $\vphi\supt$ converges to $\vphi^\star$ as $t \to \infty$.

  First, observe that $\vphi^\star \in \R^{n-1}$ is also the unique minimizer of the Lyapunov function $f$, since $f$ is non-negative and achieves the global minimum of $0$ only when \eqref{eq:equalize_pi} is satisfied.
  $f$ is not necessary convex (or quasi-convex; see Figure~\ref{fig:nonconvex_f}), but we prove in Appendix~\ref{appx:proof_lem_alternative_obj} that %
  $f$ does not have any saddle points or spurious local minima (Proposition~\ref{prop:critical_point_is_global_min}), and %
  that $f$ is coercive, meaning that every sublevel set of $f$ is bounded (Proposition~\ref{prop:coercive}).

  Let $t_0, t_1, t_2, \dots$ be the sequence of values that $t'$ in Line~\ref{line:great_step} of \Cref{alg:Newton_Armijo} ever takes, \ie the timesteps at which a %
  new update direction is determined using the current market-clearing outcome.
  We prove that the %
  sequence of update directions is \emph{gradient related} (see Proposition 1.2.1 of \cite{bertsekas2016nonlinear}), which in conjunction with Propositions~\ref{prop:critical_point_is_global_min} and~\ref{prop:coercive} imply that %
  $\vphi^{(t_k)}$ converges to $\vphi^\star$ as $k \to \infty$ (\Cref{prop:coarse_convergence}).
  Finally, we show in \Cref{prop:delta->0} that $\vdelta^{(t_k)}$ converges to $\vzero$ as $t \to \infty$, %
  which implies that $\vdelta^{(t)}$ also converges to $\vzero$. %
  For any $k \geq 0$ and any $t \in (t_k, t_{k+1}]$, we know $\|\vphi^{(t)} - \vphi^\star\| = \|\vphi^{(t_k)} + \alpha^{(t)} \vdelta^{(t)} - \vphi^\star\| \le \|\vphi^{(t_k)} - \vphi^\star\| + \|\alpha^{(t)} \vdelta^{(t)}\| \le \|\vphi^{(t_k)} - \vphi^\star\| + \|\vdelta^{(t)}\|$. Therefore, $\limsup_{t \to \infty} \|\vphi^{(t)} - \vphi^\star\| \le \limsup_{k \to \infty} \|\vphi^{(t_k)} - \vphi^\star\| + \limsup_{t \to \infty} \|\vdelta^{(t)}\| = 0$. This completes the proof of this theorem.
\end{proof}

\medskip

We now revisit our running example %
and illustrate the iterative updates under the INP mechanism.

\begin{examplecont}{ex:running-example}
  Recall %
  from Section~\ref{sec:opt_welfare_and_CE} that ixing $\phi_2 = 0$, the optimal OD-based adjustment for the morning rush hour example has $\phi_1 = 18.33$ (\ie driver supply at location~1 (the residential area) is more valuable than that at location~2 (downtown).
  Starting from from $\phi_1^{(0)} = 0$ with $\pi_1 ^{(0)} \gg \pi_2^{(0)}$, \Cref{fig:running_example_iterations} illustrates the iterative updates %
  under the INP mechanism, when the maximum change of each $\pi_i$ between timesteps is limited to $\tau = 1$. %

  We can see that %
  the difference between the origin-based multipliers reduces over time (Figure~\ref{fig:ex_Newton_pi}), and that both the primal and the dual objectives converge to the optimal welfare (Figure~\ref{fig:ex_Newton_welfare}).
  Moreover, the prices of the $(1,1)$ and $(2,2)$ trips converge to the same level %
  since the two trips take the same amount of time (Figure~\ref{fig:ex_Newton_prices}).
  We also observe superlinear convergence of $f$ when $\phi_1$ is sufficiently close to the optimum.

  \renewcommand{\MmaScale}{0.89}

  \begin{figure}[t!]
    \centering
    \subcaptionbox{OD-based adjustment.%
      \label{fig:ex_Newton_phi}}{\includegraphics[scale = \MmaScale]{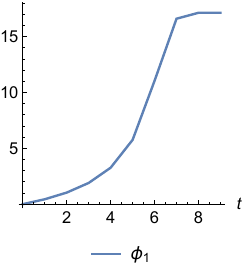}}
    \hfill
    \subcaptionbox{Multipliers. \label{fig:ex_Newton_pi}}{\includegraphics[scale = \MmaScale]{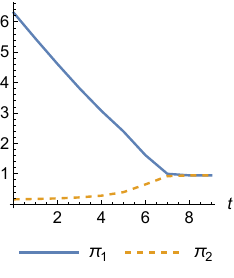}}
    \hfill
    \subcaptionbox{Primal and dual objs. \label{fig:ex_Newton_welfare}}{\includegraphics[scale = \MmaScale]{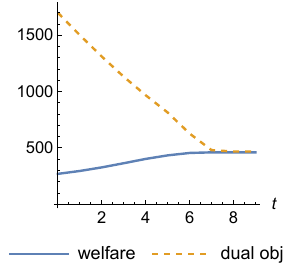}}
    \hfill
    \subcaptionbox{Objective \eqref{eq:quadratic_objective}. \label{fig:ex_Newton_f}}{\includegraphics[scale = \MmaScale]{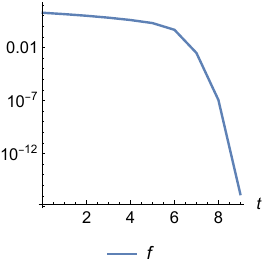}}

    \vspace{0.2em}

    \subcaptionbox{Trip prices. \label{fig:ex_Newton_prices}}{\includegraphics[scale = \MmaScale]{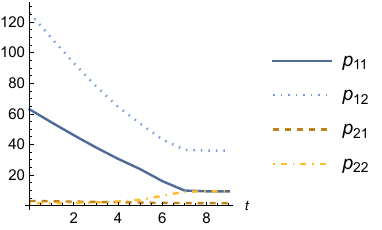}}
    \hspace{2em}
    \subcaptionbox{Rider flow. \label{fig:ex_Newton_x}}{\includegraphics[scale = \MmaScale]{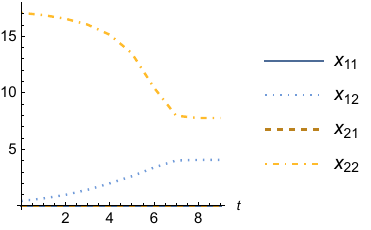}}

    \caption{Outcomes under the INP mechanism over iterations, for the economy in \Cref{ex:running-example}. %
    }
    \label{fig:running_example_iterations}
  \end{figure}
  \end{examplecont}

\subsection{Discussions} \label{sec:algorithmic_discussions}

\paragraph{Alternative Update Directions}

Under the INP mechanism, the OD-based adjustments $\vphi$ are iteratively updated towards the direction \eqref{eq:INP_direction_no_t} for equalizing the origin-based multipliers, which corresponds to the update direction under Newton's method for solving the system of %
equations \eqref{eq:equalize_pi}.
Note that this is different from Newton's method for optimization, which will determine a different update
direction %
and requires additional higher-order information (\eg the second order derivatives of the demand functions).

Since the Lyapunov function $f$ is coercive (\Cref{prop:coercive}) and has a single stationary point corresponding to the unique global optimum (\Cref{prop:critical_point_is_global_min}), gradient descent \wrt $f$ (\ie taking $-\nabla f$ as the update direction) in conjunction with backtracking line search %
has the same convergence guarantees.
This method, however, converges slowly (see simulation results in Appendix~\ref{appx:GD_simulations}) and does not require less information in comparison to the INP mechanism since the Jacobian %
$\D\vPi$ is needed for evaluating %
$\nabla f$.

A simple mechanism %
that requires very little information is to directly take %
last iteration's market-clearing multipliers $\vpi\suptmo$ as the update direction.
In Appendix~\ref{appx:pi_as_direction}, we %
show that %
simply setting $\vphi\supt = \vphi \suptmo + \alpha \vpi \suptmo$ (where $\alpha$ is some small constant)
performs very well in simulation.  %
Intuitively, moving $\vphi$ in the direction of $\vpi$ may reduce $f$ because increasing $\phi_i$ for some location $i \in \gL$ with a high multiplier $\pi_i$ will likely lead to lower prices for trips ending at location $i$ for the next iteration, thereby increasing driver supply for this location and decreasing $\pi_i$. However, $\vpi$ is not necessarily a descent direction of $f$, and we %

  are not yet able to prove or disprove the convergence of this simple mechanism.

\paragraph{Non-Stationary Market Conditions}

Backtracking line search \wrt $f$ is integrated in the design of the INP mechanism in order to guarantee convergence in theory. In simulations, for appropriately chosen
$\tau$ and $\sigma$, our mechanism converges without ever backtracking for a stationary setting (see Section~\ref{sec:stationary_simulation}). %
The practical implementation of INP should not include the backtracking mechanism since the market conditions fluctuate week-over-week %
(see Section~\ref{sec:nonstationary_simulations} for simulation results under a non-stationary setting).

By design, the INP mechanism does not include a termination criteria. In theory, with stationary market conditions, we may terminate the algorithm once the origin-based multipliers are sufficiently close.
In practice, however, the INP mechanism should be ``always on'' even after substantially improving the outcomes in order to track potential shifts in the overall market conditions. We illustrate this in simulations in the next section.

\section{Simulation Results} \label{sec:simulations}

In this section, we compare the iterative pricing
mechanism with various benchmarks via simulations,
for scenarios representing the morning rush hours in Chicago.
We fit the market conditions using trip-level data made public by the City of Chicago,%
\footnote{\url{https://data.cityofchicago.org/Transportation/Transportation-Network-Providers-Trips-2018-2022-/m6dm-c72p}, accessed May 16, 2023. The dataset contains trips from transportation network companies including both Uber and Lyft, and provides for each trip information including the starting and ending timestamps (rounded to the nearest 15 minutes), price (rounded to the nearest \$2.5), the duration (seconds), distance (miles), and the pickup and dropoff locations. %
}
and take the 77 community areas in Chicago as the set of locations (\ie $n = |\gL| = 77$).%
\footnote{\url{https://en.wikipedia.org/wiki/Community_areas_in_Chicago} (accessed May 16, 2023) provides a map of the community areas broken down by region.
  The pick-up (or drop-off) location of a trip is not provided for trips starting (or ending) outside the city boundary.
  The locations may also be omitted due to privacy considerations.
  Our simulations use only trips with both the pick-up and drop-off community areas, which consist of 86\% of all trips provided in the dataset.
  We do not expect any qualitative change in the simulation results if all trips in the greater Chicago area are taken into consideration.
}
Detailed descriptions of %
the simulation setup, the overall market dynamics in Chicago, %
and additional simulation results are provided in Appendix~\ref{appx:additiona_simulations}.

\subsection{Stationary Market Condition} \label{sec:stationary_simulation}

We first illustrate the iterative updates under the INP mechanism when the market condition is stationary and corresponds to 7--8 a.m.\@ on Wednesdays during the first 9 weeks of 2020.
For each %
OD pair $(i, j) \in \gL^2$, we set the distance $d_{i, j}$ to be the average duration of all recorded trips from %
$i$ to %
$j$ in terms of \emph{hours},
and assume that %
the trip cost is \$20/hour, \ie $c_{i, j} = 20 d_{i, j}$.%
\footnote{This is roughly aligned with a proposal from Uber in 2019 (see \url{https://p2a.co/H9gttWA}, accessed September 14, 2020) for ensuring drivers are paid an average of \$21 per hour while on trip. The earnings per hour online could be %
  slightly lower, depending on the average utilization level.}
For rider demand, we assume riders' values are exponentially distributed with mean \$60/hour for all OD pairs, and we determine the total %
demand for each OD to match the observed average trip volume given the observed average prices.
For total driver supply $m$, we %
compute the minimum amount of drivers that can fulfill the observed %
trips $(\xobs_{i, j})_{(i, j) \in \gL^2}$ and maintain flow balance.
See \Cref{appx:stationary_setting} for a more detailed description of this simulation setup.

\paragraph{Optimal OD-Based Adjustments}

Given the economy constructed above, we first %
illustrate the optimal OD-based adjustments %
(Figure~\ref{fig:chicago_optimal_phi_edited}) and the origin-based multipliers when the platform naively clears the market without any OD-based adjustments (Figure~\ref{fig:chicago_origin_based_surge_edited}). %
Observe that the naive origin-based multipliers resemble the %
the flow imbalance as shown in Figure~\ref{fig:imbalance_heatmap}, and vary substantially in space--- %
drivers can make over \$120/hour from trips originating from residential neighborhoolds in the North Side, whereas the multiplier for downtown is close to $-15$ dollars/hour, meaning that drivers are paid only around \$5/hour since the trip cost is \$20/hour.
In contrast, the optimal OD-based adjustments $\vphi^\ast$ as shown in Figure~\ref{fig:chicago_optimal_phi_edited}
has a substantially smaller spread, %
and is also appropriately smooth in space, reflecting the fact that the marginal values for divers in neighboring locations do not differ substantially.

\begin{figure}[t]
  \centering
  \subcaptionbox{Optimal OD-based additive adjustments $\vphi^\ast$.%
    \label{fig:chicago_optimal_phi_edited}}[0.48 \textwidth]{\includegraphics[width = 0.471 \textwidth]{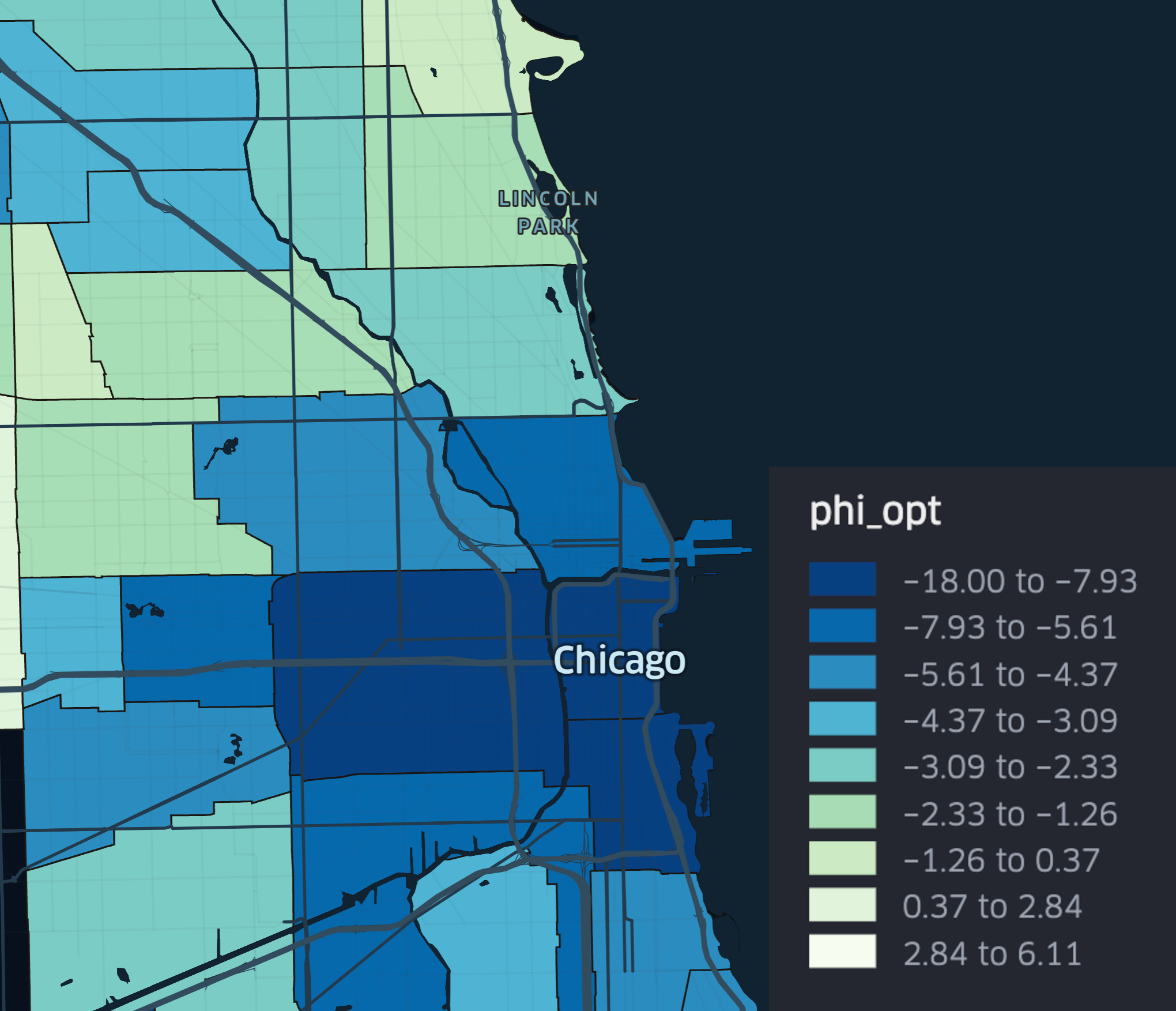}}
  \hfill
  \subcaptionbox{Origin-based multipliers $\vpi$ corresponding to $\vphi = \vzero$.%
    \label{fig:chicago_origin_based_surge_edited}}[0.48 \textwidth]{\includegraphics[width = 0.471 \textwidth]{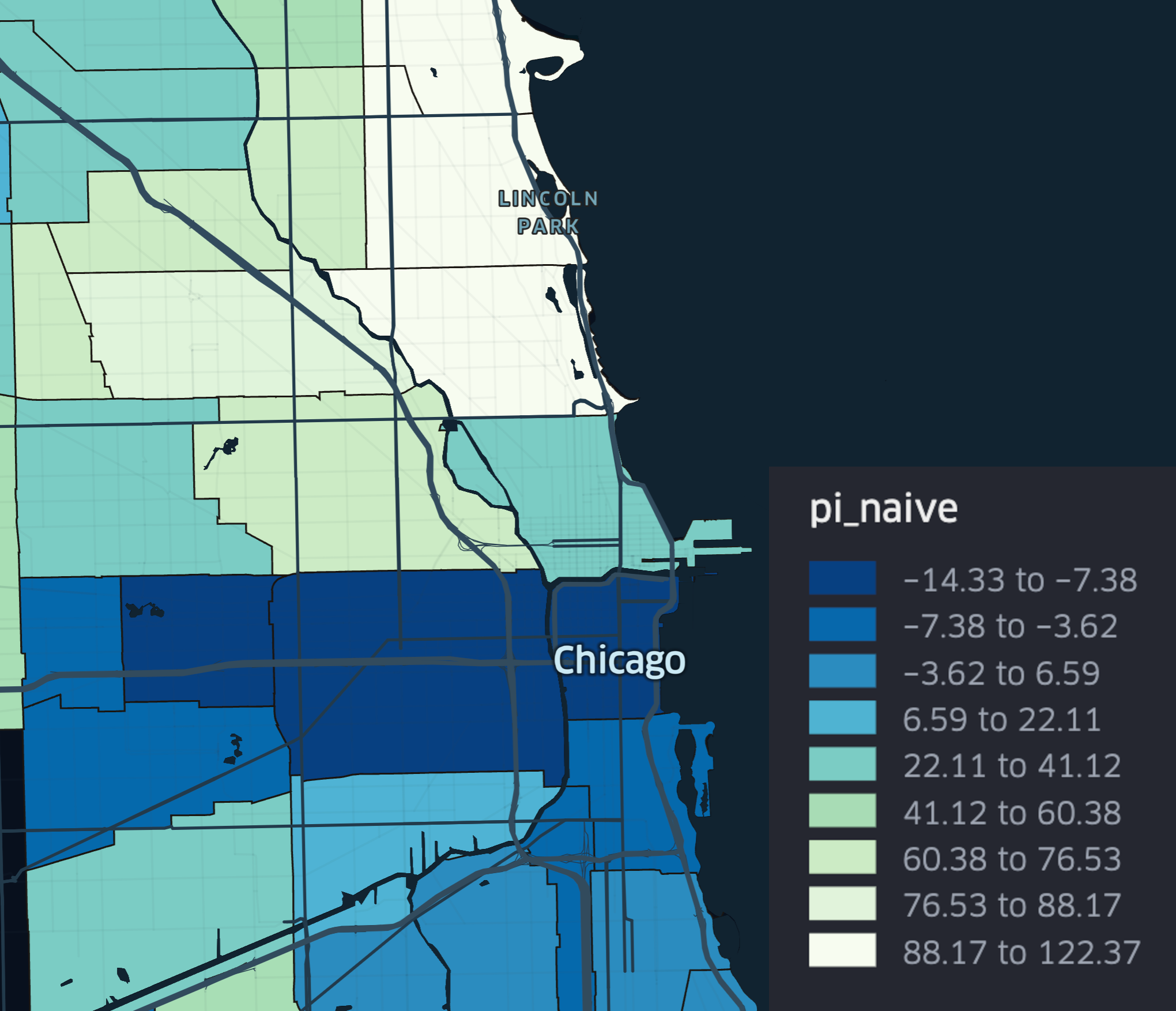}}
  \caption{%
    Optimal OD-based adjustments and naive origin-based multipliers, Chicago morning rush hours.%
  }
  \label{fig:static_heatmap}
\end{figure}

\paragraph{Iterative Updates Under INP}

In Figure~\ref{fig:static_pi_phi}, we present the OD-based additive adjustments $\vphi\supt$ and the corresponding market-clearing multipliers $\vpi\supt$ over time under the %
INP mechanism. %
Each curve in the figures represents a location (\ie a community area), and we highlight the residential area (Lake View) and the downtown area (the Loop), %
the destinations of the two trips that we discussed in Figure~\ref{fig:chicago} in Section~\ref{sec:intro}.
The shade of the remaining curves correspond to the volume of trips originating from each location.

\begin{figure}[t]
  \centering
  \subcaptionbox{The additive adjustments $\vphi^{(t)}$.\label{fig:Newton_phi_2}}{\includegraphics[scale = \PyplotScale]{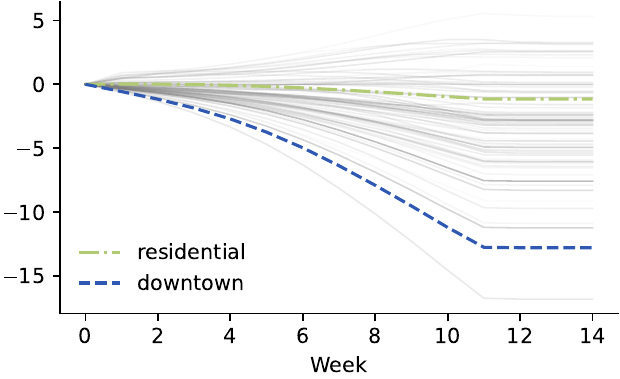}}
  \hfill
  \subcaptionbox{The multipliers $\vpi^{(t)}$.\label{fig:Newton_pi_2}}{\includegraphics[scale = \PyplotScale]{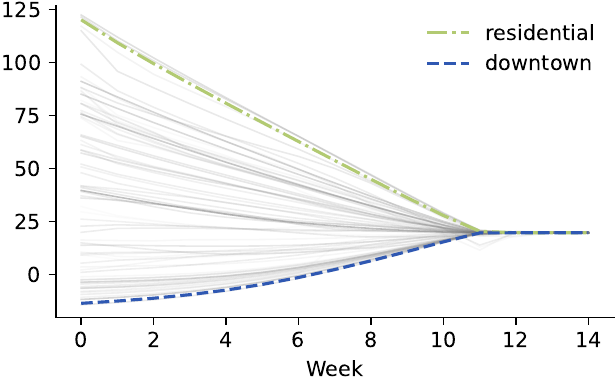}}
  \caption{$\vphi^{(t)}$ and $\vpi^{(t)}$ over iterations $t$, with the residential area and downtown highlighted.}
  \label{fig:static_pi_phi}
\end{figure}

We can see that the OD-based adjustments start from $\vphi^{(0)} = \vzero$ and %
smoothly ``spread out'' over time.
On the other hand, the origin-based multipliers %
start from $\vpi^{(0)}$ as shown in Figure~\ref{fig:chicago_origin_based_surge_edited}, and each converges to \$\NewtonMultiplier/hour after 13 iterations.
With $\tau = 10$, the mechanism aims to limit the week-over-week change of $\pi_i$ for each location $i\in \gL$ within \$10/hour.
This is based on the linear approximation of the mapping from $\vphi$ to $\vpi$, and we do observe in Figure~\ref{fig:Newton_pi_2} larger changes for some locations at certain times.%

\medskip

\begin{figure}[t]
  \centering
  \subcaptionbox{Trip prices $\vp^{(t)}$.}{\includegraphics[scale = \PyplotScale]{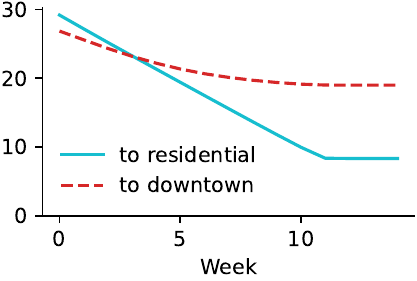}}
  \hspace{4em}
  \subcaptionbox{Rider flow rate $\vx^{(t)}$.}{\includegraphics[scale = \PyplotScale]{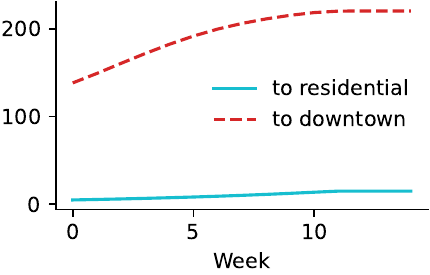}}
  \caption{Trip prices and rider flow rates for the two trips as shown in Figure~\ref{fig:chicago}, ending in the residential area and in downtown, respectively.\label{fig:static_p_x}}
\end{figure}

Figure~\ref{fig:static_p_x} plots the trip prices and %
rider flow rates for the two trips discussed in Figure~\ref{fig:chicago} in Section~\ref{sec:intro}.
We can see that the prices for both trips decrease over time, and as a result, more riders are picked up for both trips.
Initially, both prices are very high because the origin of the two trips is %
very under-supplied during morning rush hours, and therefore %
has a very high %
multiplier when market is naively cleared by origin. %
Proper OD-based price adjustments increase the supply to the location, thereby reducing the multiplier and the prices of the two trips.
Moreover, %
the price for the trip ending in the residential area drops faster than that to downtown.
This is because the OD-based adjustment is higher for the residential area in comparison, and this difference increases over time (see Figure~\ref{fig:Newton_phi_2}). %

\medskip

\begin{figure}[t]
  \centering
  \subcaptionbox{Social welfare (\$/hour).\label{fig:Newton_welfare}}{\includegraphics[scale = \PyplotScale]{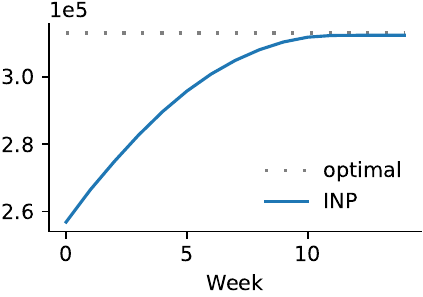}}
  \hfill
  \subcaptionbox{Dual objective.\label{fig:Newton_dual}}{\includegraphics[scale = \PyplotScale]{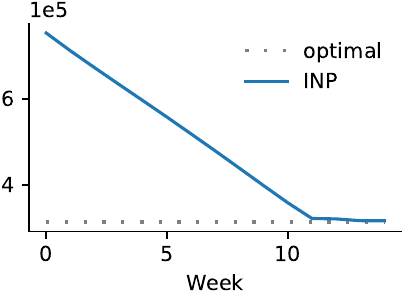}}
  \hfill
  \subcaptionbox{Function $f$.\label{fig:Newton_log_f}}{\includegraphics[scale = \PyplotScale]{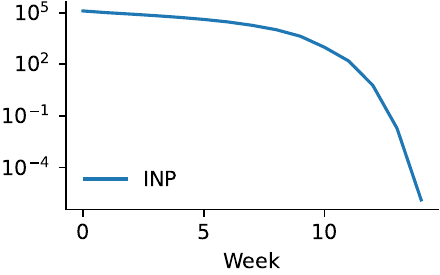}}
  \caption{Social welfare, dual objective, and function $f$ over iterations $t$.}
  \label{fig:static_obj}
\end{figure}

In \Cref{fig:static_obj}, we plot the social welfare (\ie the primal objective \eqref{eq:primal_obj_origin_based}), the dual objective \eqref{eq:dual_obj_origin_based}, and the Lyapunov function $f$ (defined in \eqref{eq:quadratic_objective}) %
corresponding to the market-clearing outcome at each time $t$.
In this setting, the optimal social welfare is \$\num{\OptimalWelfare} per hour.
The initial market clearing outcome without OD-based adjustments achieves a welfare of \$\num{\NaiveWelfare} per hour, which is \NaiveWelfareRatio{} of optimum.
Within 13 iterations, the welfare increases to \$\num{\NewtonWelfare} per hour (\NewtonWelfareRatio{} of optimum), and the gap of \NewtonWelfareRatioGap{} is a result of drivers relocating without a rider when the trip prices are positive.

The Lyapunov function $f$ defined in \eqref{eq:quadratic_objective} achieves a value of \num{\NewtonLyapunov} at time $t = \Iterations$, and we observe superlinear local convergence in Figure~\ref{fig:Newton_log_f}.
With $\tau = 10$, $\beta = 1/2$ and $\sigma = 10^{-3}$, sufficient progress in $f$ was made at every timestep and the mechanism never has to backtrack.

  In Appendix~\ref{appx:undampened_Newton}, we illustrate the iterative updates when the mechanism does not limit week-over-week changes in multipliers (\ie %
  when $\tau = +\infty$). We can see that %
  in addition to limiting the volatility of prices for riders and earnings rates for drivers (which are desirable in practice), limiting the stepsize also leads to faster convergence, even though in theory this is not necessary for convergence due to backtracking using $f$.

\subsection{Non-Stationary Market Condition} \label{sec:nonstationary_simulations}

\noindent{}In this section, we %
demonstrate the robustness of the INP mechanism under non-stationary market conditions representing the Chicago morning rush hours from the beginning of 2019 through the end of 2020. %
\Cref{fig:dynamic_volume_imbalance} presents the total trip volume and the overall flow imbalance for 7--8 a.m.\@ on Wednesdays over the two years.%
\footnote{
  The \emph{overall imbalance} %
  is defined as the average of the absolute %
  value of the flow imbalance (defined in Section~\ref{sec:intro}) of all locations, %
  \ie $\frac{1}{n} \sum_{i \in \gL} \bigl| \sum_{j \in \gL} \xobs_{i, j} - \sum_{j \in \gL} \xobs_{j, i} \bigr| \big/ \bigl(\sum_{j \in \gL} \xobs_{i, j} + \sum_{j \in \gL} \xobs_{j, i}\bigr)$, where $\xobs_{i, j}$ is the observed trip volume from $i$ to $j$.
}
We can observe the seasonality (lower trip volume during the summer, weeks 20--40), week-over-week fluctuations and the substantial shift in the overall market dynamics due to the COVID-19 pandemic.
The shutdown %
starts at week $t = 60$,
after which the total trip volume %
drops from around 16 thousand to around 2 thousand per hour.
The overall imbalance decreases from 0.3 to 0.2, %
since the %
demand from residential areas %
drops substantially %
as %
people start to work from home.

\begin{figure}[t]
  \centering
  \subcaptionbox{Trip volume.}{\includegraphics[scale = \PyplotScale]{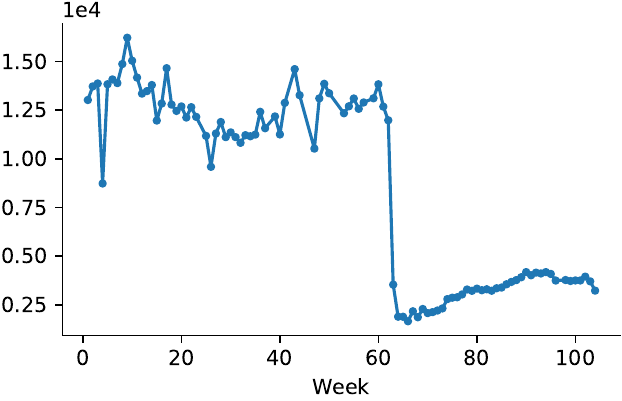}}
  \hfill
  \subcaptionbox{Overall flow imbalance.\label{fig:dynamic_imbalance}}{\includegraphics[scale = \PyplotScale]{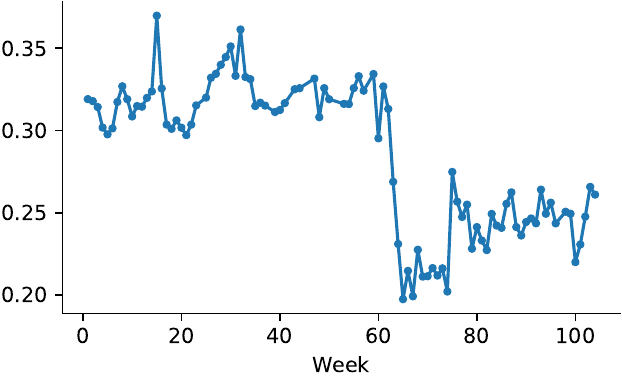}}
  \caption{Trip volume and overall flow imbalance during 7--8 a.m.\@ on Wednesdays, 2019-2020. %
  }
  \label{fig:dynamic_volume_imbalance}
\end{figure}

\Cref{appx:nonstationary_setting} provides a detailed description of the construction of the economies for each week.
Briefly, there main differences (in comparison to the scenario in Section~\ref{sec:stationary_simulation}) are (i) we use the observed trip duration and prices of each week (instead of the average over all weeks) to fit the rider demand model, and (ii) we take the minimum amount of drivers needed to fulfill the realized trip flow of each week as the driver supply for that week.%
\footnote{With (i), we effectively attribute all %
  week-over-week variability in the data to the changes of underlying market conditions. This is  likely an overestimation, since even when demand and supply are fully stationary, the realized trip flows at different weeks will be different due to stochasticity.
  (ii) effectively assumes that the driver supply level responds very fast to changes in market conditions. Our mechanism does not explicit model drivers' decisions on whether to driver for a platform. With an over 80\% drop in trip volume, however, it is not meaningful to assume a fixed level of total supply. In practice, the total driver supply equilibrates in a few weeks after a substantial change in driver earnings level~\citep{hall2017labor}.}
Moreover, we excluded %
a total of \EventDays{} ``event days'', during which the McCormick Place (Chicago's main convention center) held large events %
that draw substantial rider demand into the neighborhood immediately south of the Chicago Loop during the morning rush hours.%
\footnote{Large events are typically arranged months ahead of time, thus in practice, the platform can maintain a separate set of OD-based adjustments %
  corresponding to the scenario
  where substantially more riders demand
  for trips ending in south of loop. %
}

\paragraph{Simulation Results}

In Figure~\ref{fig:dynamic_phi_pi}, we present the OD-based additive adjustments $\vphi\supt$ and the corresponding market-clearing multipliers $\vpi\supt$ under the %
INP mechanism with $\tau = 5$, without backtracking line search.
\Cref{fig:dynamic_welfare} %
compares the social welfare achieved by INP mechanism, the na\"ive origin-based surge (\ie the market clearing outcome in the absence of OD-based adjustments), and the welfare-optimal outcome for each week.
Note that the optimal outcome represents the highest possible welfare \emph{in hindsight} (\ie using the optimal OD-based adjustments for each particular week) and cannot be achieved in practice.
Overall, the INP mechanism %
consistently outperforms the na\"ive origin-based pricing, achieving over 95\% of the optimal welfare for most of the weeks.

\begin{figure}[t]
  \centering
  \subcaptionbox{The OD-based additive adjustments $\vphi^{(t)}$. \label{fig:dynamic_Newton_phi}}{\includegraphics[scale = \PyplotScale]{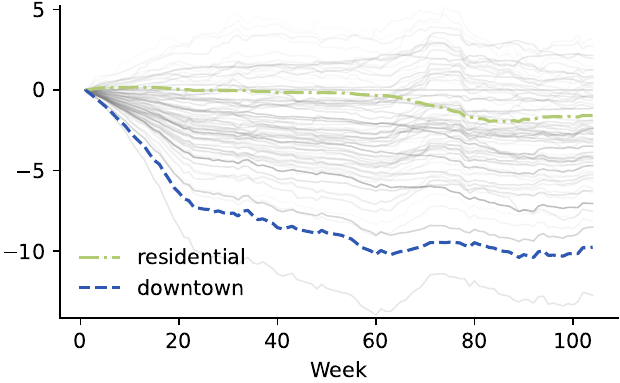}}
  \hfill
  \subcaptionbox{The multipliers $\vpi^{(t)}$.}{\includegraphics[scale = \PyplotScale]{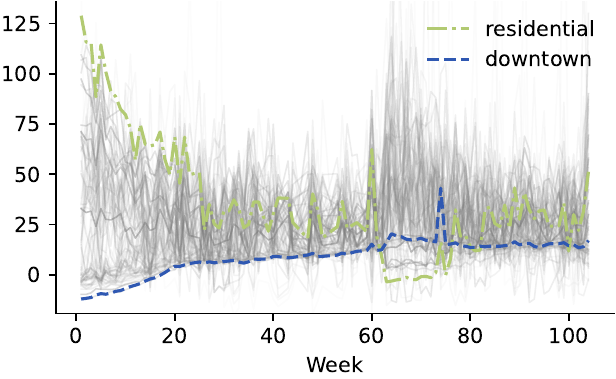}}
  \caption{$\vphi^{(t)}$ and $\vpi^{(t)}$ over weeks under the INP mechanism.}
  \label{fig:dynamic_phi_pi}
\end{figure}

\begin{figure}[t]
  \centering
  \subcaptionbox{Social welfare.}{\includegraphics[scale = \PyplotScale]{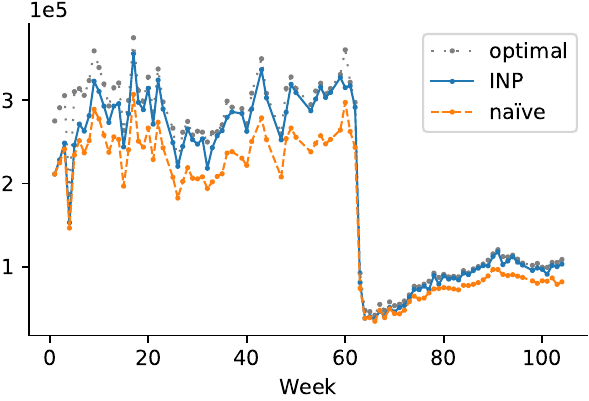}}
  \hfill
  \subcaptionbox{The welfare ratio.}{\includegraphics[scale = \PyplotScale]{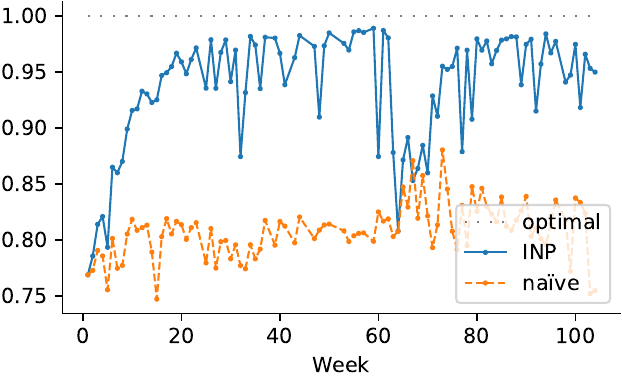}}
  \caption{Social welfare and the welfare ratio (relative to the optimal welfare of each week)} over weeks.
  \label{fig:dynamic_welfare}
\end{figure}

During the first 25 weeks, %
the welfare ratio achieved by the INP mechanism increases from 80\% to over 95\%, as the OD-based adjustments approaches the levels of the pre-COVID ``steady state''.
The difference between the OD-based adjustments of the residential area and downtown smoothly increases to around \$10, and the difference between the multipliers %
decreases from over \$100/hour to close to 0.
The na\"ive origin-based pricing without %
OD-based adjustments %
achieves around 80--85\% of the optimum welfare.
There %
are three weeks around week~$40$ where the INP mechanism achieves around 90\% of the optimal welfare. These are days during which the %
south of loop location (where convention center resides) observes larger than normal trip inflows, but did not reach the threshold to be excluded from the simulations. %

After the COVID shutdown, we observe an abrupt drop of the multiplier $\pi_i$ and a gradual decrease of the OD-based adjustment $\phi_i$ for the residential area. This is a result of the reduced imbalance for the residential area (\ie the area becoming less under-supplied), as shown in Figure~\ref{fig:imbalance_over_week}.
Despite of the substantial change in the overall market conditions after the COVID shutdown,
the INP mechanism maintains a social welfare higher than 90\% of the optimum, and recovers to over 95\% within %
a month's time. %
The performance of the na\"ive %
origin-based pricing increases %
slightly after the shutdown. This improvement %
is a result of an overall less imbalanced demand pattern (see \Cref{fig:dynamic_imbalance}), and diminishes as the economy returns to normal in 2021-2022.

\section{Conclusion} \label{sec:conclusion}

We introduce %
the \emph{iterative network pricing} (INP) mechanism, addressing a main challenge in the practical implementation of optimal origin-destination (OD) based prices, that the model for rider demand is hard to estimate.
Our mechanism, assuming real-time market clearance by origin through dynamic ``surge'' pricing, updates OD-based additive adjustments week-over-week, using only the data immediately available from previous weeks' market outcomes.
By mapping the mechanism to an %
iterative algorithm that gradually %
reduces the differences %
in driver surplus rates at different locations,
we prove that the INP mechanism converges to an %
approximately optimal outcome when the market condition is stationary. 
Using data from Chicago’s morning rush hours from early 2019 to late 2020, we demonstrate through simulations the robustness of our mechanisms, even amidst substantial fluctuations in market conditions.

\bibliographystyle{plainnat}
\bibliography{ref}

\newpage

\appendix

\begin{center}
\LARGE \textbf{Appendices}
\end{center}

\bigskip

\noindent We provide in Appendix~\ref{appx:proofs} the proofs omitted from the body of the paper.
Appendix~\ref{appx:additional_discussions} provides additional discussions and examples.
We include in Appendix~\ref{appx:additiona_simulations} detailed description of the data from
the City of Chicago, the market dynamics, the simulation setup as well as additional simulation results.

\section{Proofs} \label{appx:proofs}

We provide in this appendix proofs that are omitted from the body of the paper.

\localtableofcontents

\subsection{A Cycle Formulation} \label{appx:cycle_formulation}

Before proceeding with the proofs, we first
provide an alternative formulation of the welfare-optimization problem and its dual, which %
we use to prove Lemma~\ref{lem:welfare_theorem}. %
\begin{maxi!}
{\vx, \vy \in \R^{n^2}, \, \vw \in \R^{|\gC|}}{\sum_{i, j \in \gL} \biggl( \int_0^{x_{i, j}} v_{i, j}(s) \dl s - c_{i, j} y_{i, j} \biggr) \label{eq:primal_cycle_obj}}{\label{eq:primal_cycle}}{}
\addConstraint{x_{i, j}}{\le y_{i, j}, \quad \label{eq:primal_cycle_demand}}{\forall i, j \in \gL}
\addConstraint{\sum_{i, j \in \gL} d_{i, j} y_{i, j}}{\le m \label{eq:primal_cycle_total_supply}}
\addConstraint{\sum_{\kappa \in \gC} \one{(i, j) \in \kappa} w_\kappa}{= y_{i, j}, \quad \label{eq:primal_cycle_decomposition}}{\forall i, j \in \gL}
\addConstraint{w_\kappa}{\ge 0, \quad \label{eq:primal_cycle_w>=0}}{\forall \kappa \in \gC.}
\end{maxi!}

Compared to the flow formulation \eqref{eq:primal}, this \emph{cycle} formulation introduces additional decision variables $\vw \in \R^{|\gC|}$ representing the amount of driver flow on each directed cycle.
It is straightforward to see that the two formulations are equivalent.
The objective function \eqref{eq:primal_cycle_obj} and the constraints on demand and supply \eqref{eq:primal_cycle_demand}--\eqref{eq:primal_cycle_total_supply} are the same as those in the flow formulation \eqref{eq:primal}.
Constraints \eqref{eq:primal_cycle_decomposition}--\eqref{eq:primal_cycle_w>=0} are equivalent to constraint \eqref{eq:flow-balance} because, by the flow decomposition theorem, the driver flow is balanced if and only if it can be decomposed into cycles with non-negative weights, as formally stated in \eqref{eq:balance=cycles}.

\begin{claim} \label{claim:flow-cycle_primal_equivalence}
  The primal problem of the cycle formulation \eqref{eq:primal_cycle} is equivalent to the primal problem of the flow formulation \eqref{eq:primal} in the sense that for any \emph{feasible} solution to one problem, there exists a feasible solution to the other problem with the same objective value.
\end{claim}

Let $p\subij$, $\omega$ and $\eta\subij$ be the dual variables corresponding to constraints \eqref{eq:primal_cycle_demand}--\eqref{eq:primal_cycle_decomposition}, respectively.
The dual problem can be written as:
\begin{mini!}
{\vp \in \R^{n^2}, \, \omega \in \R, \, \veta \in \R^{n^2}}{m \omega + \sum_{i, j \in \gL} \int_{0}^{q_{i, j}(p_{i, j})} (v_{i, j}(s) - p_{i, j}) \dl s}{\label{eq:dual_cycle}}{}
\addConstraint{p_{i, j}}{= c_{i, j} + d_{i, j} \omega - \eta_{i, j}, \quad \label{eq:p=c+dpi-eta}}{\forall i, j \in \gL}
\addConstraint{\sum_{(i, j) \in \kappa} \eta_{i, j}}{\ge 0, \quad \label{eq:sum_eta>=0}}{\forall \kappa \in \gC}
\addConstraint{p_{i, j}}{\ge 0, \quad \label{eq:dual_cycle_p>=0}}{\forall i, j \in \gL}
\addConstraint{\omega}{\ge 0. \label{eq:dual_cycle_omega>=0}}
\end{mini!}

Strong duality holds between the primal problem \eqref{eq:primal_cycle} and the dual problem \eqref{eq:dual_cycle} because the primal objective \eqref{eq:primal_cycle_obj} is concave and Slater's condition can be verified by the following solution.
Let $w_\kappa = \varepsilon > 0$ for all $\kappa \in \gC$ (strictly satisfying \eqref{eq:primal_cycle_w>=0}).
Let $y_{i, j} = \sum_{\kappa \in \gC} \one{(i, j) \in \kappa} w_\kappa$ for all $(i, j) \in \gL^2$ (satisfying \eqref{eq:primal_cycle_decomposition}).
Let $x_{i, j} = y_{i, j} / 2$ for all $(i, j) \in \gL^2$ (strictly satisfying \eqref{eq:primal_cycle_demand}).
Finally, let $\varepsilon$ be sufficiently small such that $\sum_{i, j \in \gL} d_{i, j} y_{i, j} \le m$ (strictly satisfying \eqref{eq:primal_cycle_total_supply}).

\paragraph{Difference Between Primal and Dual Objectives}

Given any feasible primal solution $(\vx, \vy, \vw)$ and any feasible dual solution $(\vp, \omega, \veta)$, the difference between the primal and dual objectives can be written as:
\begin{equation} \label{eq:diff_primal_dual}
  \sum_{i, j \in \gL} \int_{x_{i, j}}^{q_{i, j}(p_{i, j})} (v_{i, j}(s) - p_{i, j}) \dl s %
  + \sum_{\kappa \in \gC} w_\kappa \sum_{(i, j) \in \kappa} \eta_{i, j} %
  + \sum_{i, j \in \gL} p_{i, j} (y_{i, j} - x_{i, j})
  + \omega \Bigl( m - \sum_{i, j \in \gL} d_{i, j} y_{i, j} \Bigr).
\end{equation}
The first term of \eqref{eq:diff_primal_dual} is weakly positive because for every trip $(i, j) \in \gL^2$, one of the following three cases must hold: %
\begin{itemize}
  \item When $x_{i, j} = q_{i, j}(p_{i, j})$, the integral is zero.
  \item When $x_{i, j} < q_{i, j}(p_{i, j})$, the fact that $v\subij(\cdot)$ is strictly decreasing implies that for all $s \in [x_{i, j}, \ab q_{i, j}(p_{i, j}))$, $v_{i, j}(s) > p_{i, j}$. The integral is therefore strictly positive.
  \item If $x_{i, j} > q_{i, j}(p_{i, j})$, $v_{i, j}(s) < p_{i, j}$ for all $s \in (q_{i, j}(p_{i, j}), x_{i, j}]$ holds again due to the strict monotonicity of $v\subij(\cdot)$. The integral is therefore also strictly positive.
\end{itemize}
The %
remaining three terms are also always weakly positive due to the primal and dual constraints. As a result, \eqref{eq:diff_primal_dual} is zero (meaning that both the primal and the dual solutions are optimal) if and only if each of the four terms is zero.

\medskip

We now show an equivalence between the dual problems of the two formulations. Note that the equivalence between the two primal formulations and the two dual formulations imply that strong duality also holds between \eqref{eq:primal} and \eqref{eq:dual}.

\begin{claim} \label{claim:flow-cycle_dual_equivalence}
  The dual problem of the cycle formulation \eqref{eq:dual_cycle} is equivalent to the dual problem of the flow formulation \eqref{eq:dual} in the sense that for any \emph{optimal} solution to one problem, there exists an optimal solution to the other problem with the same objective value.
  Moreover, variables $\vp$ and $\omega$ are the same in the two dual problems, and the correspondence between $\veta$ and $\vphi$ is given by $\eta_{i, j} = \phi_j - \phi_i$ for all $(i, j) \in \gL^2$.
\end{claim}

\begin{proof} %
  Let $(\vp, \omega, \vphi)$ be an optimal solution to the flow dual problem \eqref{eq:dual}.
  We can construct a solution $(\vp, \omega, \veta)$ to the cycle dual problem \eqref{eq:dual_cycle} with $\vp$ and $\omega$ being the same as in the flow dual problem \eqref{eq:dual}, and with $\eta_{i, j} = \phi_j - \phi_i$ for all $(i, j) \in \gL^2$.
  We can verify that constraint \eqref{eq:p=c+dpi-eta} holds by construction; constraint \eqref{eq:sum_eta>=0} holds because $\sum_{(i, j) \in \kappa} \eta_{i, j} = \sum_{(i, j) \in \kappa} (\phi_j - \phi_i) = 0$ for all $\kappa \in \gC$; constraints \eqref{eq:dual_cycle_p>=0} and \eqref{eq:dual_cycle_omega>=0} hold since $(\vp, \omega, \vphi)$ is feasible solution to \eqref{eq:dual}.
  The objective functions %
  are the same as in the flow dual problem \eqref{eq:dual},
  therefore, $(\vp, \omega, \veta)$ is a feasible solution to the cycle dual problem \eqref{eq:dual_cycle} with the same (hence optimal) objective value.

  For the other direction, let $(\vp, \omega, \veta)$ be an optimal solution to the cycle dual problem \eqref{eq:dual_cycle}.
  We %
  construct a solution $(\vp, \omega, \vphi)$ to the flow dual problem \eqref{eq:dual} with $\vp$ and $\omega$ being the same as in the cycle dual problem \eqref{eq:dual_cycle}, and with $\phi_i$ for each $i \in \gL$ being the length of the shortest path from node $n$ to node $i$ in the complete directed graph where each edge $(i, j) \in \gL^2$ has weight $\eta_{i, j}$.
  The shortest paths are well-defined because, by \eqref{eq:sum_eta>=0}, there is no negative cycle in %
  this directed graph.
  The rest of this shows that $\{\eta_{i, j}\}_{i, j \in \gL}$ has the structure such that it can be written as $\eta_{i, j} = \phi_j - \phi_i$ for all $(i, j) \in \gL^2$.
  Then, \eqref{eq:dual_price_structure} holds by construction, and as a result, $(\vp, \omega, \vphi)$ is a feasible solution to the flow dual problem \eqref{eq:dual} with the same (hence optimal) objective value.

  First, observe that for any pair of nodes $i$ and $j$, the length of the shortest path from $n$ to $j$ is at most the length of the shortest path from $n$ to $i$ plus the length of edge $(i, j)$, which is $\eta_{i, j}$.
  Therefore we have $\phi_j \le \phi_i + \eta_{i, j}$, which implies $\eta_{i, j} -  \phi_j + \phi_i \ge 0$.

  Since $(\vp, \omega, \veta)$ is an optimal solution to the cycle dual problem \eqref{eq:dual_cycle}, we know from strong duality that for any optimal solution $(\vx, \vy, \vw)$ to the cycle primal problem \eqref{eq:primal_cycle}, the difference between the primal and dual objectives \eqref{eq:diff_primal_dual} is zero.
  This implies that every term in \eqref{eq:diff_primal_dual} is zero.
  In particular, we have (i) $x_{i, j} = q_{i, j}(p_{i, j})$ for all $(i, j) \in \gL^2$ and (ii) $\sum_{(i, j) \in \kappa} \eta_{i, j} = 0$ for all $\kappa \in \gC$ with $w_\kappa > 0$.
  (i) implies that $y_{i, j} \ge x_{i, j} > 0$ for all $(i,j) \in \gL^2$, \ie the driver flow on every edge is strictly positive (recall that $q\subij(r) > 0$ for all $r \geq 0$ since $q\subij$ is assumed to be strictly decreasing for all $r \geq 0$).
  Since $\vw$ is a cycle decomposition of $\vy$, for all $(i, j) \in \gL^2$, we have $\sum_{\kappa: (i, j) \in \kappa \in \gC} w_\kappa = y_{i, j} > 0$.
  Therefore, there exists a cycle $\kappa \in \gC$ such that $(i, j) \in \kappa$ and $w_\kappa > 0$.
  Applying %
  (ii), we know that for all $(i, j) \in \gL^2$, there exists a cycle $\kappa$ covering $(i, j)$ such that $\sum_{(i', j') \in \kappa} \eta_{i', j'} = 0$.
  For this cycle $\kappa$, we know that
  \[\sum_{(i', j') \in \kappa} (\eta_{i', j'} - \phi_{j'} + \phi_{i'}) = \sum_{(i', j') \in \kappa} \eta_{i', j'} + \sum_{(i', j') \in \kappa} (\phi_{i'} - \phi_{j'}) = 0,\]
  and that $\eta_{i', j'} - \phi_{j'} + \phi_{i'} \ge 0$ for all $(i', j') \in \kappa$.
  Therefore, it must be the case that $\eta_{i', j'} = \phi_{j'} - \phi_{i'}$ for all $(i', j') \in \kappa$, %
  and this completes the proof that $\eta_{i, j} = \phi_j - \phi_i$ for all $(i, j) \in \gL^2$.
\end{proof}

\subsection{Proof of Lemma \ref{lem:welfare_theorem}}
\label{appx:proof_welfare_thm}

\lemWelfareThm*

\begin{proof}

  The following two claims prove the  equivalence between welfare-optimal and competitive equilibrium outcomes.

  \begin{claim} \label{claim:optimal=>CE}
    Given any optimal solution $(\vx, \vy)$ to the primal problem \eqref{eq:primal}, there exists $\vp$ such that $(\vx, \vy, \vp)$ forms a competitive equilibrium.
  \end{claim}

  \begin{claim} \label{claim:CE=>optimal}
    Given any competitive equilibrium $(\vx, \vy, \vp)$, there exist $\vw$, $\omega$ and $\veta$ such that $(\vx, \vy, \vw)$ is an optimal solution to the primal problem \eqref{eq:primal_cycle} and $(\vp, \omega, \veta)$ is an optimal solution to the dual problem \eqref{eq:dual_cycle}.
  \end{claim}

  To establish the structure of the CE prices,
  note that for any %
  CE outcome $(\vx, \vy, \vp)$, by \Cref{claim:CE=>optimal}, there exist $\vw$, $\omega$ and $\veta$ such that $(\vx, \vy, \vw)$ is an optimal solution to the primal problem \eqref{eq:primal_cycle} and $(\vp, \omega, \veta)$ is an optimal solution to the dual problem \eqref{eq:dual_cycle}, with $p_{i, j} = c_{i, j} + d_{i, j} \omega - \eta_{i, j}$ for all $(i,j) \in \gL^2$.%

    Moreover, the equivalence between the dual problems \eqref{eq:dual} and \eqref{eq:dual_cycle} (Claim \ref{claim:flow-cycle_dual_equivalence}) imply that there exists $\vphi \in \R^n$ such that $(\vp, \omega, \vphi)$ forms an optimal solution to the dual problem \eqref{eq:dual}, and that $\eta_{i, j} = \phi_j - \phi_i$ holds for all $(i, j) \in \gL^2$.
    As a result, the CE prices must be of the form $p_{i, j} = c_{i, j} + d_{i, j} \omega - \phi_j + \phi_i$ for all $(i,j) \in \gL^2$.
  
\end{proof}

We now provide the proof for the two claims.

\begin{proof}[Proof of \Cref{claim:optimal=>CE}]
  Intuitively, given any set of optimal driver and rider flow $(\vx, \vy)$, we derive trip prices from the optimal dual solution, and argue that the equality between primal and dual objectives implies all rider best-response and driver best-response conditions in \Cref{def:CE}.

  Formally, for any decomposition $\vw$ of the driver flow $\vy$, $(\vx, \vy, \vw)$ is an optimal solution to the primal problem \eqref{eq:primal_cycle}.
  Let $(\vp, \omega, \veta)$ be any optimal solution to the dual problem~\eqref{eq:dual_cycle}. To prove that $(\vx, \vy, \vp)$ forms a competitive equilibrium, we first derive a useful property of the optimal $\omega$.
  For all $\kappa \in \gC$, constraints \eqref{eq:p=c+dpi-eta} and \eqref{eq:sum_eta>=0} imply that
  \[
    \omega \sum_{(i, j) \in \kappa} d_{i, j} - \sum_{(i, j) \in \kappa} (p_{i, j} - c_{i, j}) = \sum_{(i, j) \in \kappa} \eta_{i, j} \ge 0.
  \]
  Rearranging the terms, we get $\omega \ge \sum_{(i, j) \in \kappa} (p_{i, j} - c_{i, j}) \big/ \sum_{(i, j) \in \kappa} d_{i, j}$ for all $\kappa \in \gC$. Moreover, constraint \eqref{eq:dual_cycle_omega>=0} ensures that $\omega \ge 0$.
  Since $\omega$ appears in the dual objective that shall be minimized, in every optimal solution, we must have
  \begin{equation} \label{eq:omega=max}
    \omega = \max \biggl\{ \max_{\kappa \in \gC} \frac{ \sum_{(i, j) \in \kappa} (p_{i, j} - c_{i, j})}{\sum_{(i, j) \in \kappa} d_{i, j}}, \, 0 \biggr\} = \max \Bigl\{ \max_{\kappa \in \gC} \mu_{\kappa}, \, 0 \Bigr\},
  \end{equation}
  where $\mu_\kappa = \sum_{(i, j) \in \kappa} (p_{i, j} - c_{i, j}) \big/ \sum_{(i, j) \in \kappa} d_{i, j}$ is the surplus rate of cycle $\kappa$ as defined in \eqref{eq:cycle_surplus_rate}.

  Strong duality implies that given optimal primal and solutions, the difference between the primal and dual objectives \eqref{eq:diff_primal_dual} is zero, and thus all four terms must be zero since each term is non-negative. This allows us to verify all conditions for competitive equilibria:
  \begin{itemize}[leftmargin=*]
    \item %
          To establish rider best-response \ref{cond:CE_rider_br}, observe that $\sum_{i, j \in \gL} \int_{x_{i, j}}^{q_{i, j}(p_{i, j})} (v_{i, j}(s) - p_{i, j}) \dl s = 0$ %
          implies that $x_{i, j} = q_{i, j}(p_{i, j})$ must hold for all $i, j \in \gL$. %
          As we've discussed in Appendix~\ref{appx:cycle_formulation} above, the integral $\int_{x_{i, j}}^{q_{i, j}(p_{i, j})} (v_{i, j}(s) - p_{i, j}) \dl s$ is strictly positive when $x_{i, j} \neq q_{i, j}(p_{i, j})$.
    \item The third term $\sum_{i, j \in \gL} p_{i, j} (y_{i, j} - x_{i, j}) = 0$ implies that either $p\subij = 0$ or $y_{i, j} = x_{i, j}$, since both $p\subij \geq 0$ and $y\subij - x\subij \geq 0$ hold for all $(i,j) \in \gL^2$. This gives us \ref{cond:CE_relocation}, \ie trips with relocation flows have zero prices.
    \item We now show \ref{cond:CE_exhaust_supply}, that all drivers are used up when some cycle has a positive surplus.
          This is because $ \omega \bigl( m - \sum_{i, j \in \gL} d_{i, j} y_{i, j} \bigr) = 0$ together with \eqref{eq:omega=max} and primal feasibility imply that whenever $\max_{\kappa \in \gC} \mu_{\kappa} > 0$, we must have $\omega > 0$ and thus $ m - \sum_{i, j \in \gL} d_{i, j} y_{i, j} = 0$.
              \item %
          Finally, the second term $\sum_{\kappa \in \gC} w_\kappa \sum_{(i, j) \in \kappa} \eta_{i, j} = 0$, together with constraint \eqref{eq:sum_eta>=0}, imply that for each cycle $\kappa \in \gC$, either $w_\kappa = 0$ and $\sum_{(i, j) \in \kappa} \eta_{i, j} = 0$.
          For each cycle $\kappa$ with $w_\kappa > 0$, it must be the case that $\sum_{(i, j) \in \kappa} \eta_{i, j} = 0$.
          Given \eqref{eq:p=c+dpi-eta}, we know that $\omega = \sum_{(i, j) \in \kappa} (p_{i, j} - c_{i, j}) \big/ \sum_{(i, j) \in \kappa} d_{i, j} = \mu_\kappa$ must hold. %
          \begin{itemize}[leftmargin=*]
            \item %
                  When $\max_{\kappa \in \gC} \mu_\kappa \ge 0$, $\omega = \max_{\kappa \in \gC} \mu_\kappa$ given  \eqref{eq:omega=max}. As a result,
                  for each cycle $\kappa$ with $w_\kappa > 0$, we have $\mu_\kappa = \omega = \max_{\kappa' \in \gC} \mu_{\kappa'}$. %
                  This gives us \ref{cond:cycle_max}, that all drivers get the highest possible surplus rate among all cycles.
            \item When $\max_{\kappa \in \gC} \mu_\kappa < 0$, we need to show that $y\subij = 0$ for all $i,j \in \gL$ so that \ref{cond:CE_do_nothing} holds.
                  Note that \eqref{eq:omega=max} implies $\omega = 0$ and therefore $\omega > \mu_\kappa$ for all $\kappa \in \gC$.
                  Since $\omega = \mu_\kappa$ for all $\kappa \in \gC$ with $w_\kappa > 0$, %
                  it must be the case that $w_\kappa = 0$ for all $\kappa \in \gC$, implying that all driver flows are zero.%
                  \footnote{Note that this also implies that given any pair of optimal solutions, $\max_{\kappa \in \gC} \mu_\kappa < 0$ is not possible under the assumption that $q\subij(\cdot)$ is strictly decreasing for all $r \geq 0$, for all $(i,j) \in \gL^2$.
                      This is because these conditions imply that $q\subij(r) > 0$ for all $r \geq 0$, %
                      meaning that when rider best-response is satisfied (as we have established above), $x\subij > 0$ must hold for all $(i,j) \in \gL^2$, implying $y\subij > 0$ as well due to primal feasibility.
                  }
          \end{itemize}
  \end{itemize}

  This completes the proof that the outcome $(\vx, \vy, \vp)$ constructed above is a competitive equilibrium.
  Additionally, we emphasize that
  since the difference of the primal and dual objective is zero for \emph{any} optimal primal and dual solution, this above argument holds for any non-negative decomposition of $\vy$.
  As a result, we have a (seemingly) stronger version of condition \ref{cond:cycle_max} with ``$\forall \vw$'' replacing ``$\exists \vw$'':
  \begin{equation} \label{eq:cycle_max_strong}
    \forall \vw, \ \forall \kappa \in \gC, \quad w_\kappa > 0 \implies \mu_\kappa = \max_{\kappa' \in \gC} \mu_{\kappa'}.
  \end{equation}
\end{proof}

\begin{proof}[Proof of \Cref{claim:CE=>optimal}]
  Given any CE outcome $(\vx, \vy, \vp)$, for any cycle decomposition of the driver flow $\vw$,
  we know $(\vx, \vy, \vw)$ is a feasible solution to the primal problem \eqref{eq:primal_cycle}.
  We will prove the optimality of $(\vx, \vy, \vw)$ and $(\vp, \omega, \veta)$ by constructing
  $\omega$ and $\veta$ such that $(\vp, \omega, \veta)$ forms a feasible %
  dual solution, and
  showing that the difference between the primal and dual objectives \eqref{eq:diff_primal_dual} is zero.

  Let $\mu_\kappa$ be the surplus rate of each cycle $\kappa \in \gC$ as defined in \eqref{eq:cycle_surplus_rate}.
  We construct $\omega$ and $\veta$ as follows:
  \begin{align}
    \omega      & \triangleq \max \biggl\{ \max_{\kappa \in \gC} \frac{\sum_{(i, j) \in \kappa} (p_{i, j} - c_{i, j})}{\sum_{(i, j) \in \kappa} d_{i, j}}, \, 0 \biggr\} = \max \Bigl\{ \max_{\kappa \in \gC} \mu_{\kappa}, \, 0 \Bigr\}, \label{eq:let_omega} \\
    \eta_{i, j} & \triangleq c_{i, j} + d_{i, j} \omega - p_{i, j}, \qquad \forall (i,j) \in \gL^2. \label{eq:let_eta}
  \end{align}

  We first argue that $(\vp, \omega, \veta)$ forms a feasible solution to the dual problem \eqref{eq:dual_cycle}.
  The feasibility constraints \eqref{eq:p=c+dpi-eta} and \eqref{eq:dual_cycle_omega>=0} hold by construction.
  Constraint \eqref{eq:dual_cycle_p>=0} holds because trip prices are non-negative (as in condition \ref{cond:feasibility_nonneg_price}).
  To see \eqref{eq:sum_eta>=0}, observe that by construction of $\omega$, given any cycle $\kappa \in \gC$, we have $\omega \ge \sum_{(i, j) \in \kappa} (p_{i, j} \ab - c_{i, j}) \big/ \sum_{(i, j) \in \kappa} d_{i, j}$.
  This implies
  \begin{align}
    \sum_{(i, j) \in \kappa} \eta_{i, j} = \sum_{(i, j) \in \kappa} (c_{i, j} + d_{i, j} \omega - p_{i, j}) = \omega \sum_{(i, j) \in \kappa} d_{i, j} - \sum_{(i, j) \in \kappa} (p_{i, j} - c_{i, j}) \geq 0.
  \end{align}

  We now prove that the difference between the primal and dual objectives is zero. Consider each of the four terms in \eqref{eq:diff_primal_dual}:
  \begin{itemize}[leftmargin=*]
    \item The rider best response condition \ref{cond:CE_rider_br} implies that $x\subij = q_{i, j}(p_{i, j})$ for all $(i,j) \in \gL^2$, thus the first term
          $\sum_{i, j \in \gL} \int_{x_{i, j}}^{q_{i, j}(p_{i, j})} (v_{i, j}(s) - p_{i, j}) \dl s $ is zero.
    \item Condition \ref{cond:CE_relocation}, \ie trip price must be zero when the driver flow exceeds the rider flow, implies that the third term $ \sum_{i, j \in \gL} p_{i, j} (y_{i, j} - x_{i, j})$ is zero.
    \item For the fourth term $\omega \bigl( m - \sum_{i, j \in \gL} d_{i, j} y_{i, j} \bigr)$, note that by construction \eqref{eq:let_omega}, $\omega$ is strictly positive only if $\mu_\kappa > 0$ for some $\kappa \in \gC$, \ie there exists some cycle with strictly positive surplus. In this case, condition \ref{cond:CE_exhaust_supply} implies that all drivers are used up, \ie $m - \sum_{i, j \in \gL} d_{i, j} y_{i, j} = 0$.
    \item What is left to prove that the second term of \eqref{eq:diff_primal_dual}, \ie $\sum_{\kappa \in \gC} w_\kappa \sum_{(i, j) \in \kappa} \eta_{i, j}$, is also zero. Intuitively, $\sum_{(i, j) \in \kappa} \eta_{i, j}$ is the total earnings a driver on cycle $\kappa$ loses per cycle, relative to making a surplus of $\omega$ per unit of time. Driver best response therefore requires that either $\sum_{(i, j) \in \kappa} \eta_{i, j} = 0$ (\ie drivers on this cycle is making the highest possible surplus rate) or $w_\kappa = 0$ (\ie no driver is on the cycle).
          Formally, we need to discuss the following two cases:
          \begin{itemize}[leftmargin=*]
            \item
                  If $\max_{\kappa \in \gC} \mu_\kappa < 0$, then by condition \ref{cond:CE_do_nothing}, $y_{i, j} = 0$ for all trips $(i, j) \in \gL$. In this case, $w_{\kappa} = 0$ for all cycles $\kappa \in \gL$ thus $\sum_{\kappa \in \gC} w_\kappa \sum_{(i, j) \in \kappa} \eta_{i, j} = 0$ trivially holds.%
                  \footnote{Similar to the discussion in the proof of Claim~\ref{claim:optimal=>CE}, we note here that this case is not possible for any competitive equilibrium under the assumption that $q\subij(r)$ is strictly decreasing for all $r\geq 0$ for all $(i,j) \in \gL^2$.}
            \item
                  If $\max_{\kappa \in \gC} \mu_\kappa \ge 0$, $\omega = \max_{\kappa \in \gL} \mu_\kappa$ given \eqref{eq:let_omega}.
                  Condition \ref{cond:cycle_max} then implies that for all cycles $\kappa \in \gC$ with $w_\kappa > 0$, we must have $\mu_\kappa = \omega$. %
                  This implies that for all $\kappa \in \gC$ \st $w_\kappa > 0$, $\omega \sum_{(i, j) \in \kappa} d_{i, j} = \sum_{(i, j) \in \kappa} (p_{i, j} - c_{i, j})$,
                  which further implies $\sum_{(i, j) \in \kappa} \eta_{i, j} = \sum_{(i, j) \in \kappa} c_{i, j} + d_{i, j} \omega - p_{i, j}
                    = 0$. %
          \end{itemize}
  \end{itemize}
  This completes the proof that $(\vx, \vy, \vw)$ and $(\vp, \omega, \veta)$ are optimal solutions to the primal problem \eqref{eq:primal_cycle} and the dual problem \eqref{eq:dual_cycle}, respectively.
\end{proof}

\subsection{Uniqueness of Optimal Dual Solution}\label{appx:proof_dual_uniqueness}

\begin{proposition}\label{prop:dual_uniqueness}

  There exists $\omega^\ast \geq 0$ and $\vphi^\ast \in \R^n$ with $\phi_n^\ast = 0$ such that any optimal solution $(\vp, \omega, \vphi)$ of the dual problem \eqref{eq:dual} satisfies
  \begin{enumerate}[label=(\roman*)]
    \item $\omega = \omega^\ast$,
    \item $\phi_i - \phi_n = \phi_i^\ast$ for all $i \in \gL$.
  \end{enumerate}
\end{proposition}

\begin{proof}
  Let $(\vx, \vy)$ be an optimal solution to the primal problem \eqref{eq:primal}, and $(\vp, \omega, \vphi)$ and optimal solution to the dual problem \eqref{eq:dual}.
  Similar to the proof of Lemma~\ref{lem:welfare_theorem}, we can decompose the difference of the primal objective \eqref{eq:primal-obj} and dual objective \eqref{eq:dual_obj} in way similar to \eqref{eq:diff_primal_dual}, and prove that in order for the difference in the objectives to be zero, we must have $x\subij = q\subij (p\subij)$ for all $i,j \in \gL$.

  Note that the difference in the objectives is zero between any pair of optimal primal and dual solutions. The strict monotonicity of $q\subij$ for all $i,j \in \gL$ therefore implies that all optimal dual solutions must have the same prices $\vq$.
  This immediately pin down the unique $\omega^\ast$ and gives us (i), since %
  $p_{n,n} = c_{n,n} + \omega d_{n,n}$ must be the same for all optimal dual solutions thus it is impossible for two optimal dual solutions to have different $\omega$'s.

  Now consider the prices $p_{i,n}$ for all $i \neq n$. The fact that $p_{i,n} = c_{i,n} + \omega d_{i,n} + \phi_i - \phi_n$ must be the same across all optimal dual solutions, together with the uniqueness of optimal $\omega$, imply that $\phi_i - \phi_n$ must be the same for all optimal $\phi$'s.
  This finishes the proof of part (ii) and concludes the proof of this proposition.
\end{proof}

\subsection{Proof of Lemma~\ref{lem:unique}} \label{appx:proof_lem_uniqueness}

\newcommand{\rhohat}{\hat{\rho}}
\newcommand{\vrhohat}{\hat{\vrho}}
\newcommand{\piushort}{\ushort\pi}
\newcommand{\vpiushort}{\ushort\vpi}
\newcommand{\pibar}{\widebar\pi}
\newcommand{\balance}{b}
\newcommand{\vbalance}{\vb}

\newcommand{\gFhat}{\hat{\gF}}

\newcommand{\vpiast}{\vpi^{\ast}}
\newcommand{\piast}{\pi^\ast}

\newcommand{\Shat}{\hat{S}}

\newcommand{\anInt}{a}
\newcommand{\supInd}{^{(\anInt)}}

\noindent{}%
We first state and prove the following result, which we use in the proofs of \Cref{lem:unique} and \Cref{thm:differentiable}.

\begin{claim} \label{lem:inverse-positive}
  Suppose $\mM \in \R^{n\times n}$ is a strictly diagonally dominant matrix (\ie $|M_{i, i}| > \sum_{j \ne i} |M_{i, j}|$ for all $i = 1, 2, \dots, n$) with strictly positive diagonal entries and strictly negative non-diagonal entries.
  Then $\mM$ is invertible and all entries of $\mM^{-1}$ are strictly positive.
\end{claim}

\begin{proof}
  Assume without loss of generality that $\max_i M_{i, i} < 1$ (otherwise scale $\mM$ by a positive factor).
  Let $\mW = \mI - \mM$, where $\mI$ is the identity matrix. %
  We know $\mW$ is a matrix with positive entries, and for each $i$,
  \[\sum_{j} |W_{i, j}| = W_{i, i} + \sum_{j \ne i} W_{i, j} = 1 - M_{i, i} + \sum_{j \ne i} |M_{i, j}| < 1.\]
  This implies $\|\mW\|_\infty < 1$.
  Therefore, as $k$ goes to infinity, $\mW^k$ converges to the zero matrix, and $\mI + \mW + \mW^2 + \cdots + \mW^k$ converges to $(\mI - \mW)^{-1} = \mM^{-1}$.
  Since all entries of $\mW$ are positive, entries of $\mM^{-1}$ are also positive.
\end{proof}

Define $\qhat_{i, j}(r) \triangleq q_{i, j}(r) + \qtilde_{i, j}(r)$ for all $i, j \in \gL$ and $r \in \R$.
This is the need of drivers on trip $(i, j)$ to satisfy both demand and relocation at price $r$.
The feasibility constraints \ref{cond:feasibility_nonneg_price}--\ref{cond:feasibility_flow_balance}, the price structure, and the market clearing conditions \ref{cond:mc_rider_reloc_flow} and \ref{cond:mc_total_supply} together impose the following constraints on $\vpi$ and $\vphi$:
\begin{align*}
  \sum_{j \in \gL} \qhat_{i, j}(c_{i, j} + d_{i, j} \pi_i + \phi_i - \phi_j)             & = \sum_{j \in \gL} \qhat_{j, i}(c_{j, i} + d_{j, i} \pi_j + \phi_j - \phi_i), \quad \forall i \in \gL, \\
  \sum_{i, j \in \gL} d_{i, j} \qhat_{i, j}(c_{i, j} + d_{i, j} \pi_i + \phi_i - \phi_j) & = m,                                                                                                   \\
  c_{i, j} + d_{i, j} \pi_i + \phi_i - \phi_j                                            & \ge 0, \quad \forall i, j \in \gL.
\end{align*}

We prove \Cref{lem:unique}, which states that the origin-based market-clearing outcome exists and is unique.

\lemUniqueness*

\begin{proof}[Proof of \Cref{lem:unique}]
    The fact that $\vpi$ clears the market means that
  trip prices are non-negative, %
  the driver flow is balanced for each location, and %
  exactly $m$ units of drivers are dispatched.

  We first prove the uniqueness by showing that the set of surge multipliers that achieves flow-balance is ``ordered'', \ie if both $\vpi$ and $\vpi'$ lead to flow-balanced outcomes, $\pi_i ' < \pi_i$ for some $i \in \gL$ implies that $\pi_j ' < \pi_j$ for all $j \in \gL$.
  Then, we prove that there exists a vector $\vpi \in \R^n$ that satisfies all three conditions above.

  \subparagraph{(i) Orderedness/Uniqueness.}
  We first prove that for some fixed $\vphi \in \R^n$, if both $\vpi$ and $\vpi'$ lead to flow balanced outcomes, then $\exists i \in \gL \text{ s.t. } \pi_i > \pi_i' \implies \pi_j > \pi_j',~\forall j \in \gL$. %
  Equivalently, we show that $\gJ \ne \varnothing \implies \gJ = \gL$, where $\gJ \triangleq \{i \in \gL \,|\, \pi_i' < \pi_i\}$.
  Assume towards a contradiction, that $\gJ \ne \varnothing$ and $\gJ \ne \gL$. %
  The total flow from $\gJ$ to $\gL \setminus \gJ$ is strictly higher under $\vpi'$ in comparison to that under $\vpi$, \ie
  \begin{equation}
    \sum_{i \in \gJ, \, j \in \gL \setminus \gJ} \qhat_{i, j}(c\subij + d\subij \pi_i' + \phi_i - \phi_j) \ab > \sum_{i \in \gJ, \, j \in \gL \setminus \gJ} \qhat_{i, j}(c\subij + d\subij \pi_i + \phi_i - \phi_j). \label{eq:flow_ooU_1}
  \end{equation}
  At the same time, the total flow from $\gL \setminus \gJ$ to $\gJ$ is weakly lower under $\vpi'$, \ie
  \begin{equation}
    \sum_{i \in \gL \setminus \gJ, \, j \in \gJ} \qhat_{i, j}(c\subij + d\subij \pi_i' + \phi_i - \phi_j) \ab \le \sum_{i \in \gL \setminus \gJ, \, j \in \gJ} \qhat_{i, j}(c\subij + d\subij \pi_i + \phi_i - \phi_j). \label{eq:flow_ooU_2}
  \end{equation}
  This implies that the flow-balance constraints for locations in $\gU$ cannot be satisfied under both $\vpi$ and $\vpi'$, since under flow balance, the left-hand-sides of \eqref{eq:flow_ooU_1} and \eqref{eq:flow_ooU_2} must be equal, and so do the right-hand-sides of the two inequalities. This %
  completes the proof of part (i).

  Note that this also implies the \emph{uniqueness} of origin-based market-clearing outcomes, \ie there can be at most one market-clearing $\vpi$, if any exists at all.
  This is because if $\vpi$ and $\vpi'$ both satisfy flow-balance and $\vpi \neq \vpi'$,
  either $\pi_i > \pi_i'$ for all $i \in \gL$ or $\pi_i < \pi_i'$ for all $i \in \gL$. The resulting outcomes %
  therefore cannot use the same number of drivers, thus \ref{cond:mc_total_supply} cannot hold under both $\vpi$ and $\vpi'$.

  \subparagraph{(ii) Existence.}
  Before proving that a market clearing $\vpi \in \R^n$ exists for any $\vphi \in \R^n$, we first introduce some notations.

  For simplicity, define for all $i, j \in \gL$ and all $\pi_i \in \R$ \st $c_{i, j} + d_{i, j} \pi_i + \phi_i - \phi_j \geq 0$,
  \begin{equation*}
    \rho_{i, j}(\pi_i) \triangleq \qhat_{i, j}(c_{i, j} + d_{i, j} \pi_i + \phi_i - \phi_j).
  \end{equation*}
  Note that $\rho\subij(\cdot)$ is defined for fixed $\vphi \in \R^n$, and represents the rate of driver flow on the $(i,j)$ trip as a function of the surge multiplier $\pi_i$.
  Moreover, for all $\pi_i \geq (\phi_j - \phi_i - c_{i, j})/d_{i, j}$, $\rho\subij(\cdot)$ is strictly positive, continuously differentiable and strictly decreasing in $\pi_i$.

  We also denote the number of drivers needed to satisfy the demand originating from each location $i \in \gL$ (as a function of $\pi_i$) as %
  \begin{equation*}
    S_i(\pi_i) \triangleq \sum_{j \in \gL} d_{i, j} %
    \rho_{i, j}(\pi_i).
  \end{equation*}
  Since trip prices must be non-negative, $S_i(\pi_i)$ is defined only for $\pi_i \geq \piushort_i$, where
  $\piushort_i$ is the lowest value of $\pi_i$ %
  such that $p_{i, j} \ge 0$ for all $j \in \gL$:
  \begin{equation}
    \piushort_i \triangleq \max_{j \in \gL} \frac{\phi_j - \phi_i - c_{i, j}}{d_{i, j}}, \quad \forall i \in \gL. \label{eq:defn_piushort}
  \end{equation}
  Given properties of $\{\rho\subij(\cdot) \}_{j \in \gL}$,
  $S_i(\pi_i)$ is also strictly positive, continuously differentiable, and strictly decreasing in $\pi_i$ for all $\pi_i \geq \piushort_i$.
  Finally, denote total the number of drivers needed to satisfy the driver flow %
  as:
  \begin{equation*}
    S(\vpi) \triangleq %
    \sum_{i \in \gL} S_i(\pi_i).
  \end{equation*}
  $S$ is well-defined for $\vpi \in \R^n$ s.t. $\pi_i \geq  \piushort_i$ for all $i\in \gL$. We denote the domain of $S$ as
  \begin{equation}
    \gD \triangleq \{\vpi \in \R^n | \pi_i \geq \piushort_i,~\forall i \in \gL \} \label{equ:domain_nonneg_price}
  \end{equation}

  To establish that a market-clearing outcome exists, what we need to prove is that there exists a $\vpi \in \gD$ such that %
  the resulting outcome is flow balanced and uses up all drivers: %
  \begin{align}
    \sum_{j \in \gL} %
    \rho_{k, j}(\pi_k) & = \sum_{i \in \gL}                                          %
    \rho_{i, k}(\pi_i), ~ \forall k \in \gL, \label{eq:existence_proof_flow_balance} \\
    S(\vpi)            & = m. \label{eq:existence_proof_total_supply}
  \end{align}
  We prove this by showing that the set of multipliers
  satisfying the flow balance constraints \eqref{eq:existence_proof_flow_balance}
  forms a one-dimensional manifold,
    and then proving that there exist a $\vpi$ in this manifold \st \eqref{eq:existence_proof_total_supply} is satisfied.

  A first challenge %
  is that it is not straightforward to prove that %
  there exists any $\vpi \in \gD$ \st \eqref{eq:existence_proof_flow_balance} is satisfied. %
  To overcome this issue, we construct an alternative $\vrhohat = (\rho\subij(\cdot))_{(i,j) \in \gL^2}$ \st
  \begin{enumerate}[label=(\Roman*)]
    \item flow-balance \eqref{eq:existence_proof_flow_balance} (stated in terms of $\vrhohat$ instead of $\vrho$) is satisfied by construction for $\vpiushort$, and
    \item $\vrhohat$ coincides with $\vrho$ for all outcomes that use $m$ units of drivers supply (\ie $\vrhohat(\vpi) = \vrho(\vpi)$ for all $\vpi \in \gD$ s.t. \eqref{eq:existence_proof_total_supply} %
          holds under $\vrhohat$).
  \end{enumerate}
  First, observe that $S_i\left(\piushort_i\right) > m$ for all $i \in \gL$ since given \eqref{eq:defn_piushort}, $p\subij = 0$ for at least one destination $j \in \gL$,  %
  in which case the $(i,j)$ trip would use $d\subij \rho\subij(\piushort_i) = d\subij \qhat\subij(0) \geq m$ units of drivers. %
  Given that $S_i(\pi_i)$ is continuous and strictly decreasing in $\pi_i$, there must exist a unique
  $\pibar_i > \piushort_i$ \st %
  $S_i(\pibar_i) = m$.
  Let $\varepsilon > 0$ be a small enough constant such that $\piushort_i + \varepsilon < \pibar_i$ for all $i \in \gL$.
  Denote %
  \[
    K \triangleq \max_{i, j \in \gL} %
    \rho_{i, j}\left(\piushort_i + \varepsilon\right),
  \]
  and define $\rhohat_{i, j}(\cdot)$ for all $(i,j) \in \gL^2$ as
  \[\rhohat_{i, j}(\pi_i) \triangleq \rho_{i, j}(\pi_i) + \bigl( K - \rho_{i, j} \bigl( \piushort_i + \varepsilon \bigr) \bigr) \biggl(\frac{\max\{\pibar_i - \pi_i, 0\}}{\pibar_i - (\piushort_i + \varepsilon)} \biggr)^2.\]
  We know that $%
    \rhohat_{i, j}\left(\piushort_i + \varepsilon\right) = K$ for all $(i,j) \in \gL$, and that $%
    \rhohat_{i, j}(\pi_i) = %
    \rho_{i, j}(\pi_i)$ for all $\pi_i \ge \pibar_i$.
  See \Cref{fig:sketch} for an illustation.

  \begin{figure}[t]
    \centering

    \begin{tikzpicture}
      \begin{axis}[
          width=0.5\textwidth,
          height=0.3\textwidth,
          xlabel={$\pi_i$},
          xtick={0, 1, 2},
          xticklabels={$\piushort_i$, $\piushort_i + \varepsilon$, $\pibar_i$},
          ytick={1.25},
          yticklabels={$K$},
          domain=0:3,
          samples=100,
          legend pos=north east,
          ymax=3,
          ymin=0,
          axis lines=left,
        ]
        \draw[dotted] (axis cs:1,0) -- (axis cs:1,3); %
        \draw[dotted] (axis cs:0,1.25) -- (axis cs:3,1.25); %
        \addplot[smooth, solid, thick, color=blue] {1.5^(-x)};
        \addlegendentry{$\rho_{i, j}(\pi_i)$}
        \addplot[smooth, dashed, thick, color=red] {1.5^(-x) + 7 / 3 * (max(0, 1 - x / 2))^2};
        \addlegendentry{$\rhohat_{i, j}(\pi_i)$}
      \end{axis}
    \end{tikzpicture}

    \caption{Sketch of functions $\rho_{i, j}(\pi_i)$ and $\rhohat_{i, j}(\pi_i)$.}
    \label{fig:sketch}
  \end{figure}
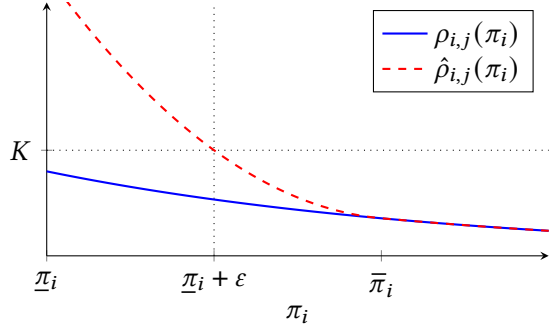

  It is straightforward to verify that (I) holds--- the flow balance constraints are trivially satisfied %
  since $\rhohat_{i, j}\left(\piushort_i\right) = K$ for all $(i,j) \in \gL^2$.
  We now argue that (II) also holds, by showing that %
  for any $\vpi \in \gD$ s.t. \eqref{eq:existence_proof_total_supply}, %
  we must have $\pi_i > \pibar_i$ for all $i \in \gL$
  since by definition of $\pibar_i$. $S_i(\pibar_i) = m$ for all $i\in \gL$ and for all other $j \in \gL$, $S_j(\pi_j)$ is strictly positive for all $\pi_j \geq \piushort_j$. %
  As a result, if there exists any $i \in \gL$ s.t. $\pi_i \leq \pibar_i$, $S(\vpi) > m$ must hold and \eqref{eq:existence_proof_total_supply} cannot be satisfied.

  Denote $ \Shat(\vpi) \triangleq \sum_{(i,j\in \gL)} d\subij \rhohat\subij(\pi_i)$,
  we now define %
  the following two subsets of $\gD$ where the flow-balance constraint and the total supply constraint are satisfied, respectively:
  \begin{align}
    \gF & = \Biggl\{ \vpi \in \gD \Biggm| \sum_{j \in \gL} %
    \rhohat_{k, j}(\pi_k) = \sum_{i \in \gL} %
    \rhohat_{i, k}(\pi_i), ~ \forall k \in \gL %
    \Biggr\} \label{eq:flow_balance_set},                  \\
    \gT & = \bigl\{ \vpi \in \gD \bigm|
    \Shat(\vpi) =  m
    \bigr\}. \label{eq:total_supply_set}
  \end{align}

  What is left to prove is that $\gT \cap \gF \neq \varnothing$.
  The fact that $\vrhohat = \vrho$ for all $\vpi \in \gT$ then implies that \eqref{eq:existence_proof_flow_balance} and \eqref{eq:existence_proof_total_supply} are both satisfied for any $\vrho \in \gT \cap \gF$, completing the proof of %
  of the existence of an origin-based market-clearing outcome.

  \medskip

  To prove $\gT \cap \gF \neq \varnothing$, we first define the following auxiliary function $\vbalance: \gD \to \R^{n-1}$ such that
  \[
    \balance_k(\vpi) = \sum_{i \in \gL} \rhohat_{i, k}(\pi_i) - \sum_{j \in \gL} \rhohat_{k, j}(\pi_k), \quad k = 1, \dots, n-1.
  \]
  Intuitively, $\balance_k$ represents the total amount of driver flow into location $k$ minus the total driver flow out of location $k$.  %
  As a result, $\vbalance(\vpi) = \vzero$ holds if and only if $\vpi \in \gF$, %
  meaning that $\gF$ is the preimage of $\vzero_{n-1}$ under function $\vbalance$.

  Given that $\rhohat_{i, j}(\cdot)$ is strictly positive, strictly decreasing, and continuously differentiable for each $(i,j) \in \gL^2$, we know that for each $k = 1, \dots, n-1$, $\balance_k(\vpi)$ is also continuously differentiable in $\vpi$.
  Now, consider the set %
  \[
    \gS \triangleq  \gF \cap \{\vpi \in \gD \,|\, \Shat(\vpi) \geq m \},
  \]
  we know $\vpiushort = \left(\piushort_i\right)_{i \in \gL} \in \gS$ since as we've discussed earlier, $\vbalance \left(\vpiushort \right) = \vzero$, and the corresponding outcome uses strictly more than $m$ units of drivers.
  We also know that $\gS$ is closed, since
  \begin{itemize}
    \item $\gD$ is closed,
    \item $\gF$ as the preimage of $\vzero$ under the continuous function $\vbalance$ is closed, and
    \item set $\{\vpi \in \gD \,|\, \Shat(\vpi) \geq m \}$ as the preimage of the closed set $[m, +\infty)$ under the continuous function $\Shat$ is also closed.
  \end{itemize}
  The set $\gS$ is also bounded above.
  Otherwise, if $\pi_i \to \infty$ for any $i \in \gL$, then the prices $p_{i, j} \to \infty$ for all $j \in \gL$ and the flow out of location $i$ must go to zero.
  To satisfy the flow balance constraint, the flow into location $i$ must also go to zero.
  Since $\vphi$ is fixed, this implies that $\pi_j \to \infty$ for all $j \in \gL$.
  This further implies that $\Shat(\vpi) \to 0$, contradicting the fact that $\Shat(\vpi) \geq m$.

  Let $\vpiast = (\piast_1, \dots, \piast_n)$ be the $\vpi \in \gS$ with the largest $\pi_n$ coordinate, \ie $\piast_n = \sup_{\vpi \in \gS} \pi_n$.
  We argue in the rest of the proof that %
  $\Shat(\vpiast) = m$ must hold, \ie $\vpiast \in \gF \cap \gT$.

  \medskip

  We first show that the mapping $\vbalance$ has no critical point. To see, this,
  divide the Jacobian matrix of function $\vbalance$ %
  into two blocks, with $\mA \in \R^{(n-1) \times (n-1)}$, $\vbeta \in \R^{n-1}$:
  \begin{equation*}
    \D \vbalance(\vpi) = \left[ \begin{array}{@{}c|c@{}} \begin{matrix} \diffp{\balance_1}{\pi_1}(\vpi) & \cdots & \diffp{\balance_1}{\pi_{n-1}}(\vpi) \\ \vdots & \ddots & \vdots \\ \diffp{\balance_{n-1}}{\pi_1}(\vpi) & \cdots & \diffp{\balance_{n-1}}{\pi_{n-1}}(\vpi) \end{matrix} & \begin{matrix} \diffp{\balance_1}{\pi_n}(\vpi) \\ \vdots \\ \diffp{\balance_{n-1}}{\pi_n}(\vpi) \end{matrix} \end{array} \right] \triangleq \left[ \begin{array}{@{}c|c@{}} \mA & \vbeta \end{array} \right].
  \end{equation*}
  For $k, j = 1, 2, \dots, n - 1$, %
  $j \ne k$, we have:
  \begin{align*}
    A_{k,k} = & \diffp{\balance_k}{\pi_k}(\vpi) = - \sum_{j \in \gL \setminus \{k\}} %
    \diff*{\rhohat_{k, j}(\pi_k)}{\pi_k} > 0,                                        \\
    A_{k,j} = & \diffp{\balance_k}{\pi_j}(\vpi) =                                    %
    \diff*{\rhohat_{j, k}(\pi_j)}{\pi_j} < 0.
  \end{align*}
  Moreover, observe that
  \[
    \left|A_{k,k} \right| = \sum_{j \in \gL \setminus \{k\}} \left| \diff*{\rhohat_{k, j}(\pi_k)}{\pi_k} \right| >
    \sum_{j \in \gL \setminus \{k, n\}} \left| %
    \diff*{\rhohat_{k, j}(\pi_k)}{\pi_k} \right| = \sum_{j \in \gL \setminus \{k, n\}} \left|A_{j,k} \right|.
  \]
  By \Cref{lem:inverse-positive}, $\mA$ is invertible, %
  and this holds at any $\vpi$, \ie $\vbalance$ does not have any critical point.
  By the implicit function theorem, there exists an open set $\gU \in \R$ containing $\pi_n^\ast$ such that there exists a unique continuously differentiable function $u: \gU \to \R^{n-1}$ with $u(\pi_n^\ast) = (\pi_1^\ast, \dots, \pi_{n-1}^\ast)$ and $\vbalance(u(\pi_n), \pi_n) = 0$ for all $\pi_n \in \gU$.

  We now assume towards a contradiction, that $S(\vpiast) \neq m$. By definition of $\gS$, we must have $S(\vpiast) > m$.
  Let $\{\pi_n\supInd \}_{\anInt = 1}^\infty$ be a sequence in $\gU$ that converges to $\pi_n^\ast$ \st $\pi_n\supInd > \pi_n^\ast$ for all $\anInt = 1, 2, \dots$, and denote $\vpi\supInd \triangleq (u(\pi_n\supInd), \pi_n\supInd)$.
  Given the continuity of $u$ %
  and %
  $\Shat$, we know that $\Shat(\vpi\supInd)$ converges to $\Shat(\pi^\ast)$ as $a \rightarrow +\infty$. This implies for $a$ that is sufficiently large, $\Shat(\vpi\supInd) > m$ holds, which implies that $\vpi\supInd \in \{\vpi \in \gD | \Shat(\vpi) \geq m \}$.
  Since $\pi_n\supInd \in \gU$, we also know that $\vpi\supInd \in \gF$. This implies that $\vpi\supInd \in \gS$, and this contradicts the fact that $\vpiast$ is the point in $\gS$ with the largest $\pi_n$.
  This completes the proof of this lemma.
  \end{proof}

\subsection{Proof of Lemma~\ref{thm:differentiable}} \label{appx:proof_thm_differentiable}

\thmDifferentiable*

\begin{proof}
  We first define an auxiliary function $\vg: \R^{2n-1} \to \R^n$, where
  \begin{align}
    g_k(\vpi, \vphi) & \triangleq \sum_{i \in \gL} \qhat_{i, k}(c_{i, k} + d_{i, k} \pi_i + \phi_i - \phi_k) - \sum_{j \in \gL} \qhat_{k, j}(c_{k, j} + d_{k, j} \pi_k + \phi_k - \phi_j), ~k = 1, 2, \dots, n-1, \label{eq:g_1} \\
    g_n(\vpi, \vphi) & \triangleq  m - \sum_{i, j \in \gL} d_{i, j} \qhat_{i, j}(c_{i, j} + d_{i, j} \pi_i + \phi_i - \phi_j). \label{eq:g_2}
  \end{align}
  Intuitively, given \ref{cond:mc_rider_reloc_flow}, the first $n - 1$ components of $\vg(\vpi, \vphi)$ correspond to the total flow of drivers into each location $k = 1, \dots, n-1$ minus the flow of drivers out of location $k$ per unit of time, when trip prices are given by $p\subij = c_{i, j} + d_{i, j} \pi_i + \phi_i - \phi_j$.
    The last component of $\vg(\vpi, \vphi)$ is the difference between the amount of available drivers $m$ and the %
  amount of drivers needed to fulfill the flow.
  $\vg(\vpi, \vphi) = \vzero$ thereby implies that the flow balance constraints
  and the total supply constraint \ref{cond:mc_total_supply} both hold, \ie $\vpi$ clears the market given destination-based adjustment $\vphi$ (recall that the first $n-1$ flow balance constraints imply the flow-balance at location $n$).
  Fixing any $\vphi \in \R^{n-1}$, Lemma~\ref{lem:unique} implies that there exists a unique $\vpi \in \R^n$ \st $\vg(\vpi, \vphi) = \vzero$, which is denoted as $\vPi(\vphi)$.

  The function $\vg$ is continuously differentiable since $\vqhat$ is continuously differentiable. The Jacobian matrix of $\vg$ can be partitioned into two blocks, representing the Jacobian of $\vg$ \wrt $\vpi$ and $\vphi$, respectively:
  \begin{equation} \label{eq:Jacobian_g_pi_phi}
    (\D\vg) (\vpi, \vphi) = \left[ \begin{array}{@{}c|c@{}}
        \begin{matrix}
          \diffp{g_1}{\pi_1}(\vpi, \vphi) & \cdots & \diffp{g_1}{\pi_n}(\vpi, \vphi) \\
          \vdots                          & \ddots & \vdots                          \\
          \diffp{g_n}{\pi_1}(\vpi, \vphi) & \cdots & \diffp{g_n}{\pi_n}(\vpi, \vphi)
        \end{matrix} &
        \begin{matrix}
          \diffp{g_1}{\phi_1}(\vpi, \vphi) & \cdots & \diffp{g_1}{\phi_{n-1}}(\vpi, \vphi) \\
          \vdots                           & \ddots & \vdots                               \\
          \diffp{g_n}{\phi_1}(\vpi, \vphi) & \cdots & \diffp{g_n}{\phi_{n-1}}(\vpi, \vphi)
        \end{matrix}
      \end{array} \right] \triangleq \left[ \begin{array}{@{}c|c@{}} \D_\vpi \vg & \D_\vphi \vg \end{array} \right].
  \end{equation}

  When $\D_\vpi \vg$ is invertable,
  the implicit function theorem%
  implies that $\vPi(\vphi)$ is continuously differentiable, with %
  \begin{equation}
    \D\vPi(\vphi) = - (\D_\vpi \vg)^{-1} \D_\vphi \vg. \label{eq:Jacobian_of_pi}
  \end{equation}
  The rest of this proof establishes that %
  $\D\vPi(\vphi)$ is invertable, and that the entries of $\D_\vpi \vg$ and $\D_\vphi \vg$ depend
  on the trip durations $\{ d\subij \}_{(i,j) \in \gL^2}$ and the local slope of the augmented demand function $\{\qhat_{i, j}'(p_{i, j})\}_{(i,j) \in \gL^2}$, both of which can be observed by the platform.

  \medskip

  We start from proving that $\D_\vpi \vg$ is invertable.
  First, partition $\D_\vpi \vg$ into four blocks, with $\mA \in \R^{(n-1) \times (n-1)}$, $\vbeta \in \R^{n-1}$, $\vgamma \in \R^{n-1}$, and $\lambda \in \R$:
  \begin{align}
    \D_\vpi \vg \triangleq \begin{bmatrix}
                             \mA          & \vbeta  \\
                             \vgamma^\trp & \lambda
                           \end{bmatrix}. \label{eq:Jac_g_wrt_pi}
  \end{align}
  When $A$ is invertable, the determinant of $\D_\vpi \vg$ can be written as:
  \[
    \det(\D_\vpi \vg) = \det(\mA) \det(\lambda - \vgamma^\trp \mA^{-1} \vbeta) \ne 0.
  \]
  We will prove that %
  $\det(\D_\vpi \vg) \neq 0$
  by proving that (i) $\mA$ is indeed invertable, and (ii) $\lambda - \bm\gamma^\trp \mA^{-1} \vbeta > 0$.
    
  To prove that (i) $\mA$ is indeed invertable, observe that the diagonal entries of $\mA$ is given by
  \begin{align}
    A_{k,k} = %
    \diffp{g_k}{\pi_k}(\vpi, \vphi) = - \sum_{j \in \gL \setminus \{k\}} d_{k, j} \qhat_{k, j}'(c_{k, j} + d_{k, j} \pi_k + \phi_k - \phi_j), ~\forall k = 1, \dots, n-1,   \label{eq:Jac_wrt_pi_diag}
  \end{align}
  and that the off-diagonal entries are of the form%
  \begin{align}
    A_{k,\ell} =
    \diffp{g_k}{\pi_\ell}(\vpi, \vphi) = d_{\ell, k} \qhat_{\ell, k}'(c_{\ell, k} + d_{\ell, k} \pi_\ell + \phi_\ell - \phi_k), ~\forall k = 1, \dots, n-1, ~\forall \ell = 1, \dots, n-1,~ \ell \neq k \label{eq:Jac_wrt_pi_off_diag}
  \end{align}
  Since $\qhat\subij(\cdot)$ is strictly decreasing for all $(i,j) \in \gL^2$, we know that $A_{k,k} > 0$ for all $k \leq n-1$, and that $A_{k,\ell} < 0$ for all $k \leq n-1$ and all $\ell \neq k$. Moreover, note that
  \begin{align*}
    |A_{k,k}| - \sum_{\ell \neq k} |A_{\ell,k}| = - d_{k,n} \qhat_{k, n}'(c_{k, n} + d_{k, n} \pi_k + \phi_k - \phi_n) > 0.
  \end{align*}
  It follows from \Cref{lem:inverse-positive} that $\mA$ is invertible and all entries of $\mA^{-1}$ are strictly positive.

  \smallskip

  To prove that (ii) $\lambda - \bm\gamma^\trp \mA^{-1} \bm\beta > 0$, observe that
  \begin{align}
    \beta_k       & = \diffp{g_k}{\pi_n}(\vpi, \vphi) = d_{\ell, k} \qhat_{n, k}'(c_{n, k} + d_{n, k} \pi_n + \phi_n - \phi_k) < 0, ~\forall k = 1, \dots, n-1, \label{eq:Jac_wrt_pi_beta}                                       \\
    \gamma_{\ell} & = \diffp{g_n}{\pi_\ell}(\vpi, \vphi) = -\sum_{j \in \gL} d_{\ell, j}^2 q_{\ell, j}'(c_{\ell, j} + d_{\ell, j} \pi_\ell + \phi_\ell - \phi_j) > 0, ~\forall \ell = 1, \dots, n-1, \label{eq:Jac_wrt_pi_gamma} \\
    \lambda       & = \diffp{g_n}{\pi_n}(\vpi, \vphi) = -\sum_{j \in \gL} d_{n, j}^2 q_{n, j}'(c_{n, j} + d_{n, j} \pi_h + \phi_n - \phi_j) > 0. \label{eq:Jac_wrt_pi_lambda}
  \end{align}
  As a result, $\bm\gamma^\trp \mA^{-1} \bm\beta < 0$, and thus $\lambda - \bm\gamma^\trp \mA^{-1} \bm\beta > 0$.

  \medskip

  What is left to show is that the entries of $\D_\vpi \vg$ and $\D_\vphi \vg$ depend
  on the trip durations and the local slope of the augmented demand function. For $\D_\vpi \vg$, this is shown in Equations \eqref{eq:Jac_wrt_pi_diag}--\eqref{eq:Jac_wrt_pi_lambda}.
  For $\D_\vphi \vg$, %
  it is convenient to partition the Jacobian matrix as
  \begin{align}
    \D_\vphi \vg \triangleq \begin{bmatrix}
                              \mB \\
                              \vtheta^\trp \label{eq:Jac_g_wrt_phi}
                            \end{bmatrix},
  \end{align}
  where $\mB \in \R^{(n-1) \times (n-1)}$ and $\vtheta \in \R^{(n-1)}$.
  For $k = 1, 2, \dots, n - 1$, the diagonal entries of $\mB$ are given by
  \begin{align}
    B_{k,k} = %
    \diffp{g_k}{\phi_k}(\vpi, \vphi) = - \sum_{i \in \gL \setminus \{k\}} \qhat_{i, k}'(c_{i, k} + d_{i, k} \pi_i + \phi_i - \phi_k) -
    \sum_{j \in \gL \setminus \{k\}} \qhat_{k, j}'(c_{k, j} + d_{k, j} \pi_k + \phi_k - \phi_j). \label{eq:Jac_wrt_phi_diag}
  \end{align}
  For $k = 1, 2, \dots, n - 1$ and $\ell = 1, 2, \dots, n - 1$ with $\ell \ne k$, the off-diagonal entry
  \begin{align}
    B_{k,\ell} = %
    \diffp{g_k}{\phi_\ell}(\vpi, \vphi) = \qhat_{\ell, k}'(c_{\ell, k} + d_{\ell, k} \pi_\ell + \phi_\ell - \phi_k) + %
    \qhat_{k, \ell}'(c_{k, \ell} + d_{k, \ell} \pi_k + \phi_k - \phi_\ell).
    \label{eq:Jac_wrt_phi_off_diag}
  \end{align}
  Finally, for $\ell = 1, 2, \dots, n-1$,
  \begin{align}
    \theta_\ell = %
    \diffp{g_n}{\phi_\ell}(\vpi, \vphi) = - \sum_{j \in \gL \setminus \{\ell\}} d_{\ell, j} \qhat_{\ell, j}'(c_{\ell, j} + d_{\ell, j} \pi_\ell + \phi_\ell - \phi_j) +
    \sum_{i \in \gL \setminus \{\ell\}} d_{i, \ell} \qhat_{i, \ell}'(c_{i, \ell} + d_{i, \ell} \pi_i + \phi_i - \phi_\ell).
    \label{eq:Jac_wrt_phi_theta}
  \end{align}
  This completes the proof of this lemma.
\end{proof}

\subsection{The Update Direction}
\label{appx:the_update_direction}

In this section, we prove the following result, which shows that any the market-clearing outcome corresponding to any $
  \vphi \in \R^{n-1}$, there exists a unique $\vphi'\in \R^{n-1}$ at which all $\pi_i$'s are equalized in the linear approximation of $\vPi$, \ie there exists $\xi \in \R$ \st  $\Pi(\vphi) + \D\vPi(\vphi)(\vphi'-\vphi) = \xi \vone_n$.
This is the direction towards which the INP mechanism updates the OD-based price adjustments week-over-week.

\begin{lemma} \label{clm:full_rank}
  The matrix $\begin{bmatrix} \D\vPi(\vphi) & -\vone \end{bmatrix} \in \R^{n \times n}$ has full rank at any $\vphi \in \R^{n-1}$.
\end{lemma}

\begin{proof}
  Recall from the proof of \Cref{thm:differentiable} that the Jacobian of $\vPi$ %
  is of the form $\D\vPi(\vphi) = - (\D_\vpi \vg)^{-1} \D_\vphi \vg$, where $\vg$ is given by \eqref{eq:g_1} and \eqref{eq:g_2} and $\D_\vpi \vg$ is full rank.
  The matrix $\begin{bmatrix} \D\vPi & -\vone \end{bmatrix}$ can therefore be written as
  \[
    \begin{bmatrix} \D\vPi & -\vone \end{bmatrix}
    = \begin{bmatrix} -(\D_\vpi \vg)^{-1} \D_\vphi \vg & -\vone \end{bmatrix}
    = - (\D_\vpi \vg)^{-1} \begin{bmatrix} \D_\vphi \vg & (\D_\vpi \vg) \vone \end{bmatrix}.
  \]

  As a result, it is sufficient to show that $\begin{bmatrix} \D_\vphi \vg & (\D_\vpi \vg) \vone \end{bmatrix}$ also has full rank. To prove this,
  we partition $\D_\vpi \vg$ and $\D_\vphi \vg$ as:
  \begin{align*}
    \D_\vpi \vg  & =
    \begin{bmatrix}
      \barmA \\
      \barvgamma^\trp
    \end{bmatrix},   \\
    \D_\vphi \vg & =
    \begin{bmatrix}
      \mB \\
      \vtheta^\trp
    \end{bmatrix}.
  \end{align*}
  Here, $\mB \in \R^{(n-1) \times (n-1)}$ and $\vtheta \in \R^{n-1}$ are exactly the same as in the partition \eqref{eq:Jac_g_wrt_phi}, given by \eqref{eq:Jac_wrt_phi_diag}--\eqref{eq:Jac_wrt_phi_theta}.
  On the other hand,
  $\barmA \in \R^{(n-1) \times n}$ and $\barvgamma \in \R^n$ %
  are given by
  \begin{align*}
    \barmA          & = \begin{bmatrix} \mA & \vbeta \end{bmatrix},           \\
    \barvgamma^\trp & = \begin{bmatrix} \vgamma^\trp & \lambda \end{bmatrix},
  \end{align*}
  where $\mA$, $\vbeta$, $\vgamma$, and $\lambda$ are components of the partition \eqref{eq:Jac_g_wrt_pi} of $\D_\vpi \vg$ (see \eqref{eq:Jac_wrt_pi_diag}--\eqref{eq:Jac_wrt_pi_lambda}).

  \smallskip

  It is straightforward to verify that
  \begin{itemize}
    \item the row sum of $\barmA$ is equal to $\vtheta$, \ie $\barmA \vone = \vtheta$, following from \eqref{eq:Jac_wrt_pi_diag}--\eqref{eq:Jac_wrt_pi_beta} and \eqref{eq:Jac_wrt_phi_theta},
    \item $\barvgamma^\trp \vone > 0$ because $\bargamma_i > 0$ for all $i \in \gL$ given \eqref{eq:Jac_wrt_pi_gamma} and \eqref{eq:Jac_wrt_pi_lambda}, and
    \item $\mB$ is symmetric, given \eqref{eq:Jac_wrt_phi_off_diag}.
  \end{itemize}

  \smallskip

  The matrix $\begin{bmatrix} \D_\vphi \vg & (\D_\vpi \vg) \vone \end{bmatrix}$ can therefore be written as: %
  \[
    \begin{bmatrix} \D_\vphi \vg & (\D_\vpi \vg) \vone \end{bmatrix} =
    \begin{bmatrix}
      \mB          & \barmA \vone          \\
      \vtheta^\trp & \barvgamma^\trp \vone
    \end{bmatrix} =
    \begin{bmatrix}
      \mB          & \vtheta               \\
      \vtheta^\trp & \barvgamma^\trp \vone
    \end{bmatrix} =
    \begin{bmatrix}
      \mI    & \vtheta / \barvgamma^\trp \vone \\
      \vzero & 1
    \end{bmatrix}
    \begin{bmatrix}
      \mB - \vtheta \vtheta^\trp / \barvgamma^\trp \vone & \vzero                \\
      \vzero^\trp                                        & \barvgamma^\trp \vone
    \end{bmatrix}
    \begin{bmatrix}
      \mI                                  & \vzero \\
      \vtheta^\trp / \barvgamma^\trp \vone & 1
    \end{bmatrix}.\]
  It suffices to show that $\mB - \vtheta \vtheta^\trp / \barvgamma^\trp \vone$ is positive definite.
  To this end, we show that for any $n-1$ dimensional vector $\vzeta  \in \R^{n-1}$ such that $\vzeta \ne \vzero_{n-1}$,
  $\vzeta^\trp (\barvgamma^\trp \vone \mB - \vtheta \vtheta^\trp) \vzeta  > 0$.

  \medskip

  In the rest of this proof, we make use of the structure of $\barmA$ and $\mB$ to show that
  \begin{align}
    \vzeta^\trp (\barvgamma^\trp \vone \mB - \vtheta \vtheta^\trp) \vzeta = %
    \vzeta^\trp \barvgamma^\trp \vone \mB \vzeta  - \vzeta^\trp \vtheta \vtheta^\trp \vzeta = %
    \barvgamma^\trp \vone (\vzeta^\trp  \mB \vzeta)  - (\vzeta^\trp \vtheta)^2 > 0, %
    ~ \forall \vzeta  \in \R^{n-1},~\vzeta \ne \vzero_{n-1}, \label{eq:no_zero_eigen_val}
  \end{align}
  or we have a counter-example of the Cauchy–Schwarz inequality.
  For simplicity of %
  notation, we define $\zeta_n \triangleq 0$, despite the fact that $\vzeta$ is a vector with dimension $n - 1$.
  First, given \eqref{eq:Jac_wrt_pi_gamma} and \eqref{eq:Jac_wrt_pi_lambda} (recall that $\bargamma_n = \lambda$), the element-wise sum of $\barvgamma$ can be written as
  \[
    \barvgamma^\trp \vone = - \sum_{i=1}^n \sum_{j=1}^n d_{i, j}^2 \qhat_{i, j}'.
  \]
  Combining \eqref{eq:Jac_wrt_phi_diag} and \eqref{eq:Jac_wrt_phi_off_diag}, $\vzeta^\trp \mB \vzeta$ %
  can be rewritten as:
  \begin{align*}
    \vzeta^\trp \mB \vzeta
     & = \sum_{i=1}^{n-1} \sum_{j=1}^{n-1} \zeta_i \zeta_j \biggl( \qhat_{j, i}' + \qhat_{i, j}' - \I[j = i] \biggl( \sum_{k=1}^n \qhat_{k, i}' + \sum_{\ell=1}^n \qhat_{i, \ell}' \biggr) \biggr)   \\
     & = \sum_{i=1}^{n} \sum_{j=1}^{n} \zeta_i \zeta_j \biggl( \qhat_{j, i}' + \qhat_{i, j}' - \I[j = i] \biggl( \sum_{k=1}^n \qhat_{k, i}' + \sum_{\ell=1}^n \qhat_{i, \ell}' \biggr) \biggr)       \\
     & = - \sum_{i=1}^{n} \zeta_i^2 \biggl( \sum_{k=1}^n \qhat_{k, i}' + \sum_{\ell=1}^n \qhat_{i, \ell}' \biggr) + \sum_{i=1}^{n} \sum_{j=1}^{n} \zeta_i \zeta_j (\qhat_{i, j}' + \qhat_{j, i}')    \\
     & = - \biggl( \sum_{i=1}^n \sum_{j=1}^{n} \zeta_j^2 \qhat_{i, j}' + \sum_{i=1}^{n} \sum_{j=1}^n \zeta_i^2 \qhat_{i, j}' \biggr) + 2 \sum_{i=1}^{n} \sum_{j=1}^{n} \zeta_i \zeta_j \qhat_{i, j}' \\
     & = - \sum_{i=1}^n \sum_{j=1}^n (\zeta_i - \zeta_j)^2 \qhat_{i, j}',
  \end{align*}
  and
  $\vzeta^\trp \vtheta$ %
  can be written as
  \begin{align*}
    \vzeta^\trp \vtheta
     & = \sum_{i=1}^{n-1} \zeta_i \biggl( - \sum_{j=1}^n d_{i, j} \qhat_{i, j}' + \sum_{k=1}^n d_{k, i} \qhat_{k, i}' \biggr)      \\
     & = \sum_{i=1}^{n} \zeta_i \biggl( - \sum_{j=1}^n d_{i, j} \qhat_{i, j}' + \sum_{k=1}^n d_{k, i} \qhat_{k, i}' \biggr)        \\
     & = - \sum_{i=1}^{n} \sum_{j=1}^n \zeta_i d_{i, j} \qhat_{i, j}' + \sum_{i=1}^n \sum_{j=1}^{n} \zeta_j d_{i, j} \qhat_{i, j}' \\
     & = - \sum_{i=1}^n \sum_{j=1}^n (\zeta_i - \zeta_j) d_{i, j} \qhat_{i, j}'.
  \end{align*}
  By Cauchy–Schwarz inequality, $(\barvgamma^\trp \vone) (\vzeta^\trp \mB \vzeta) \ge (\vzeta^\trp \vtheta)^2$.
  Moreover, the equality is achieved if and only if the ratios between the $d_{i, j}$'s and $\zeta_i - \zeta_j$'s are the same for all $i, j = 1, \dots, n$.
  This is impossible because while $d_{i, j}$ is always positive, $\zeta_i - \zeta_j$ is zero when $i = j$ and $\zeta_i - \zeta_n = \zeta_i \ne 0$ for some $i \leq n-1$.
  Therefore, $\vzeta^\trp (\barvgamma^\trp \vone \mB - \vtheta \vtheta^\trp) \vzeta > 0$ must hold.
  This completes the proof %
  that
  $\begin{bmatrix} \D_\vphi \vg & (\D_\vpi \vg) \vone \end{bmatrix}$ is positive definite, and also the proof of this lemma.
\end{proof}

\subsection{Properties of the Lyapunov Function}
\label{appx:proof_lem_alternative_obj}

\noindent{}We establish in this section useful properties of the alternative objective function $f$ as defined in \eqref{eq:quadratic_objective}.
We use $\vzero_\ell$ and $\vone_\ell$ denote $\ell$-dimensional vectors of all zeros and all ones, respectively, and drop the subscript when the dimension is clear from the context.

\begin{restatable}{lemma}{lemObjFn} \label{lem:alternative_obj}
  The function $f$ as defined in \eqref{eq:quadratic_objective} is %
  continuously differentiable,
  has a single critical point corresponding to the unique global minimum, and is coercive (\ie every sublevel set is bounded).
\end{restatable}

\newcommand{\phiStar}{\phi^\star}
\newcommand{\vphiStar}{\vphi^\star}
\newcommand{\vyStar}{\vy^\star}
\newcommand{\omegStar}{\omega^\star}
\newcommand{\pStar}{p^\star}
\newcommand{\vpStar}{\vp^\star}

\begin{proof}[Proof of Lemma~\ref{lem:alternative_obj}]
We first prove that $f$ has a unique minimizer $\vphiStar \in \R^{n-1}$, achieving $f(\vphiStar) = 0$.
Given any augmented economy $(\vd, \vc, \vq, m, \vqtilde)$,
consider the following %
optimization problem, where $\hat v(\cdot) = \qhat^{-1}(\cdot)$ is the inverse function of $\qhat(\cdot)$ (recall that %
$\qhat\subij(\cdot) = q\subij(\cdot) + q\subij(\cdot)$ is strictly decreasing for all $(i,j) \in \gL^2$):
\begin{maxi!}
{\vy \in \R^{n^2}}{\sum_{i, j \in \gL} \biggl( \int_0^{y_{i, j}} \hat v_{i, j}(s) \dl s - c_{i, j} y_{i, j} \biggr)}{\label{eq:primal_no_relocation}}{}
\addConstraint{\sum_{j \in \gL} y_{k, j}}{= \sum_{i \in \gL} y_{i, k}, \quad}{\forall k = 1, 2, \dots, n - 1 \label{eq:primal_no_reloc_flow_balance}}
\addConstraint{\sum_{i, j \in \gL} d_{i, j} y_{i, j}}{= m \label{eq:primal_no_reloc_total_supply}}
\end{maxi!}

Let $\vphi = (\phi_1, \dots, \phi_{n-1})$ denote the dual variables corresponding to %
\eqref{eq:primal_no_reloc_flow_balance},
and let $\omega$ be the dual variable corresponding to the total supply constraint \eqref{eq:primal_no_reloc_total_supply}, and fix $\phi_n \triangleq 0$.%
\footnote{Note that the first $n-1$ flow balance constraints \eqref{eq:primal_no_reloc_flow_balance} imply that location $n$ is also flow balanced. For convenience of notation, we drop the %
  redundant flow balance
  constraint, which is equivalent to fixing the %
  dual variable at $\phi_n = 0$.}
The dual of \eqref{eq:primal_no_relocation} can be written as
\begin{mini!}
{\omega \in \R, \, \vphi \in \R^{n-1}}{m \omega + \sum_{i, j \in \gL} \int_{c_{i, j} + d_{i, j} \omega + \phi_i - \phi_j}^\infty \hat q_{i, j}(r) \dl r, \label{eq:dual_obj_no_relocation}}{%
}{}
\end{mini!}
and difference between primal and dual objectives %
is of the form
\begin{equation} \label{eq:diff_primal_dual_no_relocation}
  \sum_{i, j \in \gL} \int_{y_{i, j}}^{\hat q_{i, j}(c_{i, j} + d_{i, j} \omega + \phi_i - \phi_j)} (\hat v_{i, j}(s) - (c_{i, j} + d_{i, j} \omega + \phi_i - \phi_j)) \dl s + \omega \Bigl( m - \sum_{i, j \in \gL} d_{i, j} y_{i, j} \Bigr) + \sum_{k \in \gL} \phi_k \Bigl( \sum_{i \in \gL} y_{i, k} - \sum_{j \in \gL} y_{k, j} \Bigr).
\end{equation}

Let $\vyStar \in \R^{n^2}$ and $(\vphiStar, \omegStar)$ be optimal solutions of the primal \eqref{eq:primal_no_reloc_total_supply} and dual \eqref{eq:dual_obj_no_relocation}, respectively. %
The primal-dual difference \eqref{eq:diff_primal_dual_no_relocation} is zero for $\vyStar \in \R^{n^2}$ and $(\vphiStar, \omegStar)$
since the primal problem is convex.
It is straightforward to verify that \eqref{eq:diff_primal_dual_no_relocation} being zero implies
that the surge multiplier $\vpi = \omegStar \vone_n$ (\ie $\pi_1 = \dots = \pi_n = \omegStar$) clears the market when the OD-based adjustments are given by $\vphiStar$.
In other words, $\vPi(\vphiStar) = \omegStar \vone$, implying $f(\vphiStar) = 0$. This is a global minimum, since $f$ is non-negative by construction.

\smallskip

We now prove that the global minimum is unique.
Suppose %
$\vphi' \in \R^{n-1}$ %
also minimize \eqref{eq:quadratic_objective}, \ie $f(\vphi') = 0$.
We know that the market-clearing surge multipliers $\vpi' = \vPi(\vphi')$ must also satisfy $\pi'_1 = \dots = \pi'_n$.
Let $\omega' = \pi'_1 $, and let $\vy'$ be the resulting driver flow under this market-clearing outcome. $\vy'$ is a feasible primal solution. We can verify that the market-clearing conditions in Definition~\ref{def:MC} imply that the difference between the primal objective achieved by  $\vy'$ and the dual objective achieved by $(\vphi', \omega')$ is zero, meaning that both solutions are optimal.
With an argument similar to the proof of \Cref{prop:dual_uniqueness}, we can prove that the optimal dual solution is unique.
This implies that $\vphiStar = \vphi'$, \ie \eqref{eq:quadratic_objective} has a unique minimum.

\bigskip

To complete the proof of this theorem, we establish the following  results.

\begin{restatable}%
  {proposition}{clmCriticalPt} \label{prop:critical_point_is_global_min}
  Every stationary point of %
  $f$ is a global minimum.
\end{restatable}

\begin{restatable}{proposition}{lemCoercive} \label{prop:coercive}
  $f$ as defined in \eqref{eq:quadratic_objective} is coercive, \ie %
  every sublevel set of %
  $f$ is bounded. Formally,
  $\forall U \in \R$, $\exists M > 0$, $\forall \vphi \in \R^{n-1}$, $f(\vphi) \leq U \implies \|\vphi\|_\infty \leq M$.
\end{restatable}
\end{proof}

\subsubsection{Proof of Proposition~\ref{prop:critical_point_is_global_min}}

\clmCriticalPt*

\begin{proof}
  First note that the alternative objective function $f$ as defined in  \eqref{eq:quadratic_objective} can be written as: %
  \begin{equation*} %
    f(\vphi) = \sum_{i \in \gL} \Bigl( \Pi_i(\vphi) - \frac{1}{n} \sum_{j \in \gL} \Pi_j(\vphi) \Bigr)^2 = \frac{1}{2 n} \sum_{i, j \in \gL} (\Pi_i(\vphi) - \Pi_j(\vphi))^2.
  \end{equation*}
  Define $\bm\xi: \R^{n-1} \to \R^n$ such that for each $i = 1, 2, \dots, n$,
  \[
    \xi_i(\vphi) \triangleq 2 \Pi_i(\vphi) - \frac{2}{n} \sum_{j \in \gL} \Pi_j(\vphi).
  \]
  The gradient of $f$ can be written as:
  \begin{align*}
    \nabla f(\vphi)
     & = \frac{1}{n} \sum_{i, j \in \gL} (\Pi_i(\vphi) - \Pi_j(\vphi)) (\nabla \Pi_i(\vphi) - \nabla \Pi_j(\vphi))     \\
     & = \frac{2}{n} \sum_{i, j \in \gL} (\Pi_i(\vphi) - \Pi_j(\vphi)) \nabla \Pi_i(\vphi)                             \\
     & = \sum_{i \in \gL} \Bigl( 2 \Pi_i(\vphi) - \frac{2}{n} \sum_{j \in \gL} \Pi_j(\vphi) \Bigr) \nabla \Pi_i(\vphi) \\
     & = \D \bm\Pi(\vphi)^\trp \bm\xi(\vphi).                                                                          %
  \end{align*}
  Observe that the entries of $\bm\xi(\vphi)$ always sum to zero, \ie $\vone^\trp \bm\xi(\vphi) = 0$, $\forall \vphi \in \R^{n-1}$.
  As a result,
  \[
    \nabla f(\vphi) = \D \bm\Pi(\vphi)^\trp \bm\xi(\vphi) = \vzero_{n-1} \iff \begin{bmatrix} \D \bm\Pi(\vphi) & -\vone_n \end{bmatrix}^\trp \bm\xi(\vphi) = \vzero_n.
  \]
  By \Cref{clm:full_rank}, the matrix $\begin{bmatrix} \D \bm\Pi(\vphi) & -\vone_n \end{bmatrix}$ has full rank at all $\vphi \in \R^{n-1}$.
  $\nabla f(\vphi) = \vzero$ then implies that $\bm\xi(\vphi) = \vzero$, %
  meaning that for all $i \in \gL$,
  \[
    \Pi_i(\vphi) = \frac{1}{n} \sum_{j \in \gL} \Pi_j(\vphi),
  \]
  In other words, $\Pi_1(\vphi) = \cdots = \Pi_n(\vphi)$,
  in which case $f(\vphi) = 0$.
  Since $f$ is non-negative, this completes the proof that every stationary point of $f$ is a global minimum.
\end{proof}

\bigskip

\subsubsection{Proof of Proposition~\ref{prop:coercive}} \label{appx:proof_coercive}

\newcommand{\pbar}{\bar{p}}
\newcommand{\Mbar}{\bar{M}}

\lemCoercive*

\begin{proof}%
  Assume toward contradiction that %
  $f$ has an unbounded sublevel set, \ie
  \begin{equation}
    \exists U^\ast \in \R, \ \forall M > 0, \ \exists \vphi, \ (f(\vphi) \le U^\ast) \wedge (\|\vphi\|_\infty > M). \label{eq:coercive_assumptpion}
  \end{equation}
  We know it must be the case that $U^\ast > 0$, since (as we have established earlier in this section) %
  $f$ is non-negative and achieves the global minimum $f = 0$ at a unique $\vphiStar \in \R^{n-1}$.

  The first term in \eqref{eq:coercive_assumptpion}, $f(\vphi) \le U^\ast$, implies that $\bigl( \Pi_i(\vphi) - \frac{1}{n} \sum_{j \in \gL} \Pi_j(\vphi) \bigr)^2 \leq U^\ast$ holds for all $i \in \gL$. As a result, at any $\vphi \in \R^{n-1}$ \st $f(\vphi) \le U^\ast$, we must have
  \begin{equation} \label{eq:Pi's_are_roughly_equal}
    \max_{i \in \gL} \Pi_i(\vphi) - \min_{i \in \gL} \Pi_i(\vphi) = \Bigl( \max_{i \in \gL} \Pi_i(\vphi) - \frac{1}{n} \sum_{i \in \gL} \Pi_i(\vphi) \Bigr) + \Bigl( \frac{1}{n} \sum_{i \in \gL} \Pi_i(\vphi) - \min_{i \in \gL} \Pi_i(\vphi) \Bigr) \le 2 \sqrt{U^\ast}.
  \end{equation}

  The second term in \eqref{eq:coercive_assumptpion}, $\|\vphi\|_\infty > M$, implies that there exists some $i \leq n-1$ s.t. $|\phi_i| = |\phi_i - \phi_n| > M $.
  Assume w.l.o.g.\@ that $|\phi_1 - \phi_n| > M$ (otherwise re-label the locations). Consider the prices of the two trips between location $1$ and location $n$:
  \begin{align*}
    p_{1,n}  & = c_{1, n} + d_{1, n} \Pi_{1}(\vphi) + \phi_{1} - \phi_{n}, \\
    p_{n, 1} & = c_{n, 1} + d_{n, 1} \Pi_{n}(\vphi) + \phi_{n} - \phi_{1}.
  \end{align*}
  Trip prices being non-negative implies that
  \begin{align*}
    \Pi_{1}(\vphi) & \ge \frac{\phi_{n} - \phi_{1} - c_{1, n}}{d_{1, n}}, \\
    \Pi_{n}(\vphi) & \ge \frac{\phi_{1} - \phi_{n} - c_{n, 1}}{d_{n, 1}}.
  \end{align*}
  Since $|\phi_1 - \phi_n| > M$, we have
  \begin{align*}
    \max_i \Pi_i(\vphi) & \ge \max\{\Pi_1(\vphi), ~\Pi_n(\vphi)\} %
    \geq \min \Bigl\{ \frac{M - c_{1, n}}{d_{1, n}}, \, \frac{M - c_{n, 1}}{d_{n, 1}} \Bigr\}.
  \end{align*}
  Furthermore, applying \eqref{eq:Pi's_are_roughly_equal},
  we get
  \begin{align}
    \min_i \Pi_i(\vphi) & \ge \max_i \Pi_i(\vphi) -  2 \sqrt{U^\ast} \geq  \min \Bigl\{ \frac{M - c_{1, n}}{d_{1, n}}, \, \frac{M - c_{n, 1}}{d_{n, 1}} \Bigr\} - 2 \sqrt{U^\ast}. %
    \label{eq:limit_on_min_pi}
  \end{align}

  \smallskip

  The rest of the proof applies \eqref{eq:limit_on_min_pi} to establish that there exists a sufficiently large $\Mbar > 0$ \st for all $\vphi \in \R^{n-1}$ such that $\|\vphi\|_\infty > \Mbar$,
  the corresponding market-clearing outcome cannot satisfy
  $f(\vphi) \le U^\ast$.
  We start from constructing one such $\Mbar$.

  First, for each $(i,j) \in \gL$, we define $\pbar\subij \geq 0$ as the smallest non-negative price for the $(i,j)$ trip \st %
  the driver flow for $(i,j)$ trip is no greater than $m/(2n^2 \max_{i,j\in \gL}d\subij)$:
  \[
    \pbar\subij = \inf \left\lbrace
    r \in \R_{\geq 0} ~\left|~ \qhat\subij(r) \leq %
    \varepsilon
    \right.
    \right\rbrace,
  \]
  where
  \begin{align*}
    \varepsilon \triangleq \frac{m}{(n^3 + n^2 + n) \max_{i,j\in \gL} \{d\subij \} }.
  \end{align*}
  If $\qhat\subij(0) \leq \varepsilon$, we know $\pbar\subij = 0$. If $\qhat\subij(0) > \varepsilon$, since $\qhat\subij$ is continuous, strictly decreasing, and converges to $0$ as the price approaches infinity, $\pbar\subij$ is the unique prices at which $\qhat\subij(\pbar\subij) = \varepsilon$.

  For each origin, let $\pibar_i$ be the smallest $\pi_i$ such that $c\subij + d\subij \pi_i$ is at least $\pbar \subij$ for all destinations $i\in \gL$:
  \begin{equation}
    \pibar_i \triangleq \max_{j \in \gL} \left\lbrace
    \left(\pbar\subij - c\subij \right) / d\subij
    \right\rbrace, \label{eq:coercive_proof_pibar}
  \end{equation}
  and let $\Mbar$ be sufficiently large such that \eqref{eq:limit_on_min_pi} guarantees that $\min_i \Pi_i(\vphi)$ is no smaller than the maximum $\pibar_i$ among all $i\in \gL$:
  \begin{equation}
    \Mbar \triangleq \max\left\lbrace
    \left( \max_{i \in \gL} \pibar_i + 2 \sqrt{U^\ast} \right) d_{1,n} + c_{1,n},~
    \left( \max_{i \in \gL} \pibar_i + 2 \sqrt{U^\ast} \right) d_{n,1} + c_{n,1}
    \right\rbrace. \label{eq:coercive_proof_Mibar}
  \end{equation}

  Now consider a $\vphi \in \R^{n-1}$ \st $\|\vphi\|_\infty > \Mbar$, under which $f(\vphi) \le U^\ast$ is satisfied (\eqref{eq:coercive_assumptpion} guarantees that such $\vphi$ exists).
  Applying \eqref{eq:limit_on_min_pi}, \eqref{eq:coercive_proof_pibar} and \eqref{eq:coercive_proof_Mibar}, %
  we know that for all $i, j \in \gL$ \st $\phi_i \geq \phi_j$, %
  the trip price under the market-clearing outcome must satisfy
  \begin{align*}
    p_{i, j} = c_{i, j} + d_{i, j} \Pi_{i}(\vphi) + \phi_{i} - \phi_{j} \ge c_{i, j} + d_{i, k} \min_i \Pi_i(\vphi) \geq \pbar \subij. %
  \end{align*}
  This implies that for all $i,j\in \gL$ such that $\phi_i \geq \phi_j$, the driver flow rate %
  \[
    y\subij = \qhat(p\subij) \leq \qhat(\pbar\subij) \leq \varepsilon.
  \]
  Relabel the locations \st $\phi_1 \geq \phi_2 \geq \dots \geq \phi_n$ and apply \Cref{clm:bound_on_total_supply} below, we know that the total amount of drivers dispatched under the market-clearing outcome is bounded by:
  \begin{align*}
    \sum_{i,j \in \gL} d\subij %
    y\subij & \leq \max_{i,j\in \gL}\{d\subij\} \sum_{i,j \in \gL} y\subij \leq
    \max_{i,j\in \gL}\{d\subij\}(n + n^2 + n^3) %
    \varepsilon = m.
  \end{align*}
  This is a contradiction, since given \ref{cond:mc_total_supply}, any market-clearing outcome must use exactly $m$ units of drivers.
  This completes the proof that $f$ is coercive.
  \if
    As a result, the total amount of drivers dispatched under the market-clearing outcome can be bounded by
    \begin{align}
      \sum_{i,j \in \gL} d\subij \qhat \subij(p\subij) & \leq \max_{i,j\in \gL} \{d\subij \} \left(
      \sum_{i\in \gL} \qhat_{i,i}(p_{i,i}) +
      \sum_{i,j\in \gL: \phi_i > \phi_j} \qhat\subij(p\subij) +
      \sum_{i,j\in \gL: \phi_i < \phi_j} \qhat\subij(p\subij)
      \right) \notag                                                                                \\
                                                       & \leq \max_{i,j\in \gL} \{d\subij \} \left(
      \sum_{i\in \gL} \qhat_{i,i}(p_{i,i}) +
      2\sum_{i,j\in \gL: \phi_i > \phi_j} \qhat\subij(p\subij)
      \right)  \notag %
      \\
                                                       & \leq \max_{i,j\in \gL} \{d\subij \} \left(
      \sum_{i\in \gL} \qhat_{i,i}(\pbar_{i,i}) +
      2 \sum_{i,j\in \gL: \phi_i > \phi_j} \qhat\subij(\pbar\subij)
      \right) \notag                                                                                \\
                                                       & \leq \max_{i,j\in \gL} \{d\subij \} \left(
      \sum_{i\in \gL}  \frac{m}{2n^2 \max_{i,j\in \gL} d\subij} +
      2 \sum_{i,j\in \gL: \phi_i > \phi_j} \frac{m}{2n^2 \max_{i,j\in \gL} d\subij}
      \right)  \notag                                                                               \\
                                                       & \leq m/2. \notag
    \end{align}
    The second inequality makes use of the fact that
    when the market-clearing outcome is flow-balanced, we must have
    \[
      \sum_{i,j\in \gL: \phi_i > \phi_j} \qhat\subij(\pbar\subij) =
      \sum_{i,j\in \gL: \phi_i < \phi_j} \qhat\subij(p\subij).
    \]
    This completes the proof that $f$ is coercive.
  \fi
    \end{proof}

\begin{claim}\label{clm:bound_on_total_supply}
  Suppose driver flow $\vy \in \R^{n^2}$ is balanced (\ie $\sum_{i \in \gL}y_{i,k} = \sum_{j\in \gL} y_{k,j}$ for all $k \in \gL$) and satisfies
  \begin{equation*}
    y\subij \leq \varepsilon,~\forall i, j \in \gL ~\mathrm{s.t.}~ i \leq j, %
  \end{equation*}
  the sum of driver flow over all trips is bounded by
  \[
    \sum_{i,j\in \gL} y_{i,j} < (n + n^2 + n^3) \varepsilon.
  \]
\end{claim}
\begin{proof}
  First, we decompose the sum $\sum_{i,j\in \gL} y_{i,j}$ into two parts:
  \begin{align*}
    \sum_{i,j\in \gL} y_{i,j} = %
    \sum_{i,j\in \gL, i \leq j} y\subij + \sum_{i,j\in \gL, i > j} y\subij.
  \end{align*}
  Since $y\subij \leq \varepsilon$ if $i\leq j$, the first term is bounded by
  \[
    \sum_{i,j \in \gL, i \leq j} y\subij \leq \varepsilon \frac{n(n+1)}{2} < (n + n^2)\varepsilon.
  \]
  Due to the flow-balance constraint, for each $k \geq 2$ we must have $\sum_{i \geq k, j < k} y\subij = \sum_{i \geq k, j < k} y_{j,i}$. As a result,
  the second term can be written as
  \[
    \sum_{i,j\in \gL, i > j} y\subij = \sum_{k =2}^n \sum_{i \geq k, j < k} y\subij \leq \sum_{k=2}^{n} (k-1)(n-k+1) \varepsilon  < n^3 \varepsilon.
  \]
  This completes the proof of this claim.
  \end{proof}

\subsection{Proof of Theorem~\ref{thm:INP}} \label{appx:proof_main_thm}

\thmINP*

The following proposition is a direct corollary of Proposition 1.2.1 of \cite{bertsekas2016nonlinear}.

\begin{proposition} \label{prop:limit_point_is_stationary}
  Let $f$ be the objective function defined in \eqref{eq:quadratic_objective}, and %
  $(\vphi^{[k]})_{k \in \sN}$ be a sequence generated by: %
  \[
    \vphi^{[k+1]} = \vphi^{[k]} + \nu^{[k]} \alpha^{[k]} \vdelta^{[k]}.
  \]
  Assume the followings: %
  \begin{enumerate}[label = (\roman*)]
    \item $(\alpha^{[k]} \vdelta^{[k]})_{k \in \sN}$ is gradient related, \ie
          for any subsequence $(\vphi^{[k]})_{k \in \gK}$ (for some $\gK \subseteq \sN$, $|\gK| = \infty$) that converges to a non-stationary point of $f$, the corresponding subsequence
          $(\alpha^{[k]} \vdelta^{[k]})_{k \in \gK}$ is bounded and satisfies %
          \[
            \limsup_{k \to \infty, \, k \in \gK} \nabla f(\vphi^{[k]})^\trp \alpha^{[k]}
            \vdelta^{[k]} < 0.
          \]
    \item $\nu^{[k]}$ is chosen by the Armijo rule, \ie $\nu^{[k]} = \beta^{N_k}$, where $\beta \in (0, 1)$, $\sigma \in (0, 1)$, and $N_k$ is the smallest non-negative integer $N$ for which
          \[
            f(\vphi^{[k]}) - f(\vphi^{[k]} + \beta^{N} \alpha^{[k]} \vdelta^{[k]}) > - \sigma \beta^{N} \nabla f(\vphi^{[k]})^\trp \alpha^{[k]} \vdelta^{[k]}.
          \]
  \end{enumerate}

  Then, every limit point of $(\vphi^{[k]})_{k \in \sN}$ is a stationary point of $f$. %
\end{proposition}

Recall that $t_0, t_1, t_2, \dots$ are the values that $t'$ in Line~\ref{line:great_step} of \Cref{alg:Newton_Armijo} ever takes. For each integer $k \ge 0$, define
\begin{align*}
  \vphi^{[k]}   & \triangleq \vphi^{(t_k)},     \\
  \vpi^{[k]}    & \triangleq \vpi^{(t_k)},      \\
  \vdelta^{[k]} & \triangleq \vdelta^{(t_k+1)}, \\ \alpha^{[k]} & \triangleq \alpha^{(t_k+1)}.
\end{align*}
We know that
\[\vphi^{[k+1]} = \vphi^{[k]} + \beta^{N_k} \alpha^{[k]} \vdelta^{[k]},\]
where $N_k = t_{k+1} - t_k - 1$ is the smallest integer $N \ge 0$ such that
\begin{equation} \label{eq:Armijo_in_reduction}
  f(\vphi^{[k]} + \beta^{N} \alpha^{[k]} \vdelta^{[k]}) < f(\vphi^{[k]}) + \sigma \nabla f(\vphi^{[k]})^\trp \beta^{N} \alpha^{[k]} \vdelta^{[k]}.
\end{equation}

\begin{proposition} \label{prop:coarse_convergence}
  $(\vphi^{[k]})_{k \in \sN}$ converges to $\vphi^\star$.
\end{proposition}

\begin{proof}%

  We %
  first apply \Cref{prop:limit_point_is_stationary} and show that every limit point of $(\vphi^{[k]})_{k \in \sN}$ is a stationary point of $f$.
  By construction, the backtracking factor is chosen by the Armijo rule thus condition (ii) of \Cref{prop:limit_point_is_stationary} is satisfied.
  What is left to verify is condition (i), that $(\alpha^{[k]} \vdelta^{[k]})_{k \in \sN}$ is gradient related.
  By definition of the INP mechanism (see Algorithm~\ref{alg:Newton_Armijo}), we know that
  \begin{align*}
    \vdelta^{[k]} & = \bigl( \begin{bmatrix} -\D\vPi(\vphi^{[k]}) & \vonen \end{bmatrix}^{-1} \vpi^{[k]} \bigr)_{1:n-1}, \\
    \alpha^{[k]}  & = \min \left\lbrace 1, \, \frac{\tau}{\|\D\vPi(\vphi^{[k]}) \vdelta^{[k]} \|_\infty} \right\rbrace,
  \end{align*}
  Moreover, the gradient of $f$ \wrt $\vphi$ at $\vphi^{[k]}$ is of the form:
  \begin{align}
    \nabla f(\vphi^{[k]}) & = 2 \bigl( \D\vPi(\vphi^{[k]}) \bigr)^\trp \Bigl( \vpi^{[k]} - \Bigl( \frac{1}{n} \vonen^\trp \vpi^{[k]} \Bigr) \vonen \Bigr) \\
                          & = 2 \bigl( \D\vPi(\vphi^{[k]}) \bigr)^\trp \Bigl( \mI - \frac{1}{n} \vonen \vonen^\trp \Bigr) \vpi^{[k]}.
  \end{align}
  Since $\vonen ^\trp  \bigl( \mI - \frac{1}{n} \vonen \vonen^\trp \bigr) \vpi^{[k]} = 0$, we %
  can rewrite $\nabla f$ %
  with a zero appended at the end %
  as
  \[
    \begin{bmatrix} \nabla f(\vphi^{[k]}) \\ 0 \end{bmatrix}
    =
    2 \begin{bmatrix} \D\vPi(\vphi^{[k]}) & -\vonen \end{bmatrix}^\trp
    \Bigl( \mI - \frac{1}{n} \vonen \vonen^\trp \Bigr) \vpi^{[k]}.
  \]
  Therefore,
  \begin{align}
    \nabla f(\vphi^{[k]})^\trp \vdelta^{[k]}
     & =  \begin{bmatrix} \nabla f(\vphi^{[k]}) \\ 0 \end{bmatrix}^\trp \begin{bmatrix} -\D\vPi(\vphi^{[k]}) & \vonen \end{bmatrix}^{-1} \vpi^{[k]} \notag                                                  \\
     & = 2 \bigl( \vpi^{[k]} \bigr)^\trp \Bigl( \mI - \frac{1}{n} \vonen \vonen^\trp \Bigr)^\trp \begin{bmatrix} \D\vPi(\vphi^{[k]}) & -\vonen \end{bmatrix} \begin{bmatrix} -\D\vPi(\vphi^{[k]}) & \vonen \end{bmatrix}^{-1} \vpi^{[k]} \notag \\
     & = - 2 \bigl( \vpi^{[k]} \bigr)^\trp \Bigl( \mI - \frac{1}{n} \vonen \vonen^\trp \Bigr) \vpi^{[k]} \notag                                                                                                                                 \\
     & \le 0. \label{eq:Newton_decreases_f}
  \end{align}
  The last inequality holds with equality if and only if $\pi_1^{[k]} = \pi_2^{[k]} = \dots = \pi_n^{[k]}$, at which point $f = 0$ achieves the global minimum.
  As a result, $\nabla f(\vphi^{[k]})^\trp \vdelta^{[k]} < 0$ (\ie $\vdelta^{[k]}$ is a descent direction) for $f$ whenever $\vphi^{[k]}$ is not optimal.
  What this implies is that the last term of \eqref{eq:Armijo_in_reduction}, $\sigma \nabla f(\vphi^{[k]})^\trp \beta^{N} \alpha^{[k]} \vdelta^{[k]}$, is weakly negative, and that the objective value $f(\vphi^{[k]})$ is weakly decreasing in $k$.
  As a result, every point in the sequence $(\vphi^{[k]})_{k \in \sN}$ must be contained in the sublevel set $\{\vphi \in \sR^{n-1} \mid f(\vphi) \le f(\vphi^{(0)})\}$, which is bounded and compact given \Cref{prop:coercive}.

  Observe that both $\vdelta^{[k]}$ and $\alpha^{[k]}$ can be written as functions of $\vphi^{[k]}$. The functions are continuous because $\vPi$ is continuously differentiable. Since every continuous function over a compact space must be bounded, $(\alpha^{[k]} \vdelta^{[k]})_{k \in \sN}$ is bounded.
  For the same reason, $(\|\D\vPi(\vphi^{[k]}) \vdelta^{[k]}\|_\infty)_{k \in \sN}$ is bounded, and thus $(\alpha^{[k]})_{k \in \sN}$ is bounded away from zero.

  Now consider a subsequence $(\vphi^{[k]})_{k \in \gK}$ that converges to a non-stationary point.
  We know that as $k \rightarrow \infty$, the limits of both $\alpha^{[k]}$ and $\nabla f(\vphi^{[k]})^\trp \vdelta^{[k]}$ exist due to continuity.
  $\lim_{k \to \infty, \, k \in \gK} \alpha^{[k]} > 0$ since for each $k \in \sN$, $\alpha^{[k]}$ is positive and bounded away from zero. Moreover, %
  $\lim_{k \to \infty, \, k \in \gK} \nabla f(\vphi^{[k]})^\trp \vdelta^{[k]} < 0$ since $\lim_{k \to \infty, \, k \in \gK} \vphi^{[k]}$ is not a stationary point of $f$. This implies that
  \[
    \limsup_{k \to \infty, \, k \in \gK} \nabla f(\vphi^{[k]})^\trp \alpha^{[k]} \vdelta^{[k]} < 0,
  \]
  and completes the proof of condition (i), that $(\alpha^{[k]} \vdelta^{[k]})_{k \in \sN}$ is gradient related.

  \medskip

  It follows from \Cref{prop:limit_point_is_stationary} that every limit point of $(\vphi^{[k]})_{k \in \sN}$ is a stationary point of $f$.
  By \Cref{lem:alternative_obj}, there is a single stationary point corresponding to the unique global minimum of $f$. %
  It then follows that every limit point of $(\vphi^{[k]})_{k \in \sN}$ is the unique global minimum, and therefore $(\vphi^{[k]})_{k \in \sN}$ converges to the global minimum, completing the proof of this proposition.

  To see the last statement, assume toward contradiction that the sequence $(\vphi^{[k]})_{k \in \sN}$ does not converge to the global minimum $\vphi^\star$. In this case, there exists $\varepsilon > 0$ and $\gK \subseteq \sN$ with $|\gK| = \infty$
  such that $\vphi^{[k]}$ resides outside of the open ball $\gB_\varepsilon(\vphi^\star) \triangleq \{\vphi \in \sR^{n-1} \mid \| \vphi - \vphi^\star \|_2 < \varepsilon \}$ for all $k \in \gK$. (Here $\|\cdot\|_2$ denotes the 2-norm for vectors). %
  Observe that the infinite subsequence $(\vphi^{[k]})_{k \in \gK}$ is contained in the compact set $\{\vphi \in \sR^{n-1} \mid f(\vphi) \le f(\vphi^{(0)})\} \setminus \gB_\varepsilon(\vphi^\star)$, thus %
  must have a limit point in that set, which does not contain the global minimum. This contradicts with the fact that every limit point of $(\vphi^{[k]})_{k \in \sN}$ is the global minimum.
\end{proof}

\begin{proposition} \label{prop:delta->0}
  $(\vdelta^{[k]})_{k \in \sN}$ converges to $\vzero_{n-1}$.
\end{proposition}

\begin{proof}
  We prove this proposition by showing that the 2-norm of $\vdelta^{[k]}$ converges to zero.
  By definition, $\vdelta^{[k]}$ is the first $n-1$ entries of
  $\bigl( \begin{bmatrix} -\D\vPi(\vphi^{[k]}) & \vonen \end{bmatrix}^{-1} \vpi^{[k]} \bigr)$.
  As a result, for each $k \ge 0$, there exist $\xi^{[k]} \in \R$ \st %
  \[
    \begin{bmatrix} \vdelta^{[k]} \\ \xi^{[k]} \end{bmatrix} = \begin{bmatrix} - \D\vPi(\vphi^{[k]}) & \vonen \end{bmatrix}^{-1} \vpi^{[k]},
  \]
  which gives us
  \[
    \begin{bmatrix} - \D\vPi(\vphi^{[k]}) & \vonen \end{bmatrix} \begin{bmatrix} \vdelta^{[k]} \\ \xi^{[k]} \end{bmatrix} = \vpi^{[k]}.
  \]
  Subtracting $\frac{1}{n} \vonen^\trp \vpi^{[k]} \vonen$ from both sides of the equation, we have
  \[
    \begin{bmatrix} - \D\vPi(\vphi^{[k]}) & \vonen \end{bmatrix} \begin{bmatrix} \vdelta^{[k]} \\ \xi^{[k]} - \frac{1}{n} \vonen^\trp \vpi^{[k]} \end{bmatrix}
    = \vpi^{[k]} - \frac{1}{n} \vonen^\trp \vpi^{[k]} \vonen,
  \]
  and equivalently,
  \[
    \begin{bmatrix} \vdelta^{[k]} \\ \xi^{[k]} - \frac{1}{n} \vonen^\trp \vpi^{[k]} \end{bmatrix}
    = \begin{bmatrix} - \D\vPi(\vphi^{[k]}) & \vonen \end{bmatrix}^{-1} \Bigl(\vpi^{[k]} - \frac{1}{n} \vonen^\trp \vpi^{[k]} \vonen \Bigr).
  \]
  Let $\|\cdot\|_2$ denote the 2-norm for vectors and the induced 2-norm for matrices. We have
  \begin{align*}
    \|\vdelta^{[k]}\|_2^2
    \le \biggl\| \begin{bmatrix} \vdelta^{[k]} \\ \xi^{[k]} - \frac{1}{n} \vonen^\trp \vpi^{[k]} \end{bmatrix} \biggr\|_2^2
     & \le \bigl\| \begin{bmatrix} - \D\vPi(\vphi^{[k]}) & \vonen \end{bmatrix}^{-1} \bigr\|_2^2 \Bigl\| \vpi^{[k]} - \frac{1}{n} \vone^\trp \vpi^{[k]} \vonen \Bigr\|_2^2 \\
     & = \bigl\| \begin{bmatrix} - \D\vPi(\vphi^{[k]}) & \vonen \end{bmatrix}^{-1} \bigr\|_2^2 \, f(\vphi^{[k]}).
  \end{align*}

  Given \Cref{prop:coarse_convergence} and the continuity of $\D\vPi$, we know that as $k \to \infty$, %
  $\bigl\| \begin{bmatrix} - \D\vPi(\vphi^{[k]}) & \vonen \end{bmatrix}^{-1} \bigr\|_2^2$ converges to $\bigl\| \begin{bmatrix} - \D\vPi(\vphi^\star) & \vonen \end{bmatrix}^{-1} \bigr\|_2^2$, and %
  $f(\vphi^{[k]})$ converges to 0.
  Therefore, $\|\vdelta^{[k]}\|_2^2$ converges to 0, completing the proof of this proposition.
\end{proof}

\section{Additional Examples} \label{appx:additional_discussions}

The following example demonstrates that the dual objective (as a function of the OD-based adjustments $\vphi$) and the Lyanupov function as defined in \eqref{eq:quadratic_objective} are not necessarily quasi-convex.

\begin{example} \label{exmp:dual_obj_non_convex}

  \begin{figure}[t]
    \centering
    \subcaptionbox{Dual objective.\label{fig:nonconvex_dual}}{\includegraphics{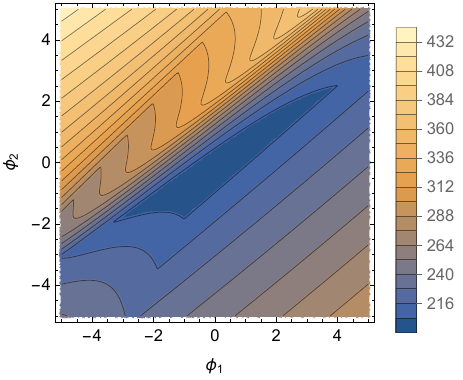}}
    \hfill
    \subcaptionbox{Lyanupov function $f$. \label{fig:nonconvex_f}}{\includegraphics{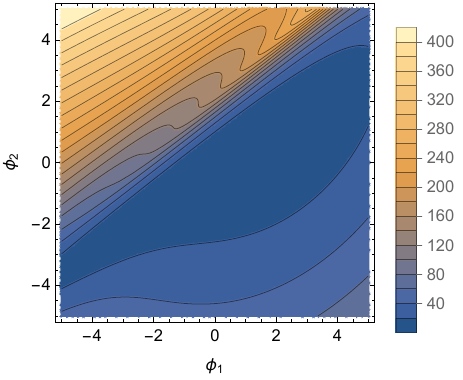}}
    \caption{Contour plots of the non-convex dual objective and the Lyapunov function $f$ for \Cref{exmp:dual_obj_non_convex}.}
    \label{fig:nonconvex}
  \end{figure}

  Consider an economy with three-locations, %
  a total of $m = 8.6$ units of drivers, and %
  rider demand functions
  \begin{align*}
    q_{1, 1}(r) = q_{1, 2}(r) = q_{2, 3}(r) = q_{3, 1}(r) = q_{3, 2}(r) = q_{3, 3}(r) & = \exp(-r / 30), \\
    q_{1, 3}(r) = q_{2, 1}(r) = q_{2, 2}(r)                                           & = 4 \exp(-r).
  \end{align*}
  The trip costs are zero, \ie  $c_{i, j} = 0$ for all $i,j \in \gL$, and the durations are given by $d_{i, j} = 1$ for all $i, j$.
  \Cref{fig:nonconvex} shows the dual objective and the Lyapunov function $f$ as a function of $\phi_1$ and $\phi_2$, fixing $\phi_3 = 0$.
  We can see that the contours are not convex, meaning that the dual objective and the function $f$ are not quasi-convex.

\end{example}

\clearpage

\section{Simulation Details and Additional Results} \label{appx:additiona_simulations}

In this section, we provide detailed descriptions of %
the simulation settings that are omitted from Section~\ref{sec:simulations} of the paper (Appendix~\ref{appx:simulation_details}).
Moreover, we describe the overall
market dynamics in Chicago (Appendix~\ref{appx:Chicago_market_dynamics}), and include additional simulation results on:
\begin{enumerate}[label = (\roman*)]
  \item the \emph{pure Newton} (as oppose to \emph{damped Newton}) method that does not limit the maximum allowed changes in the origin-based multipliers (\Cref{appx:undampened_Newton}),
  \item an iterative mechanism based on gradient descent \wrt the %
        $f$ (\Cref{appx:GD_simulations}), and
  \item a simple method that directly uses the multipliers as the direction for updating the OD-based adjustments (\Cref{appx:pi_as_direction}).
\end{enumerate}

\subsection{Chicago Market Dynamics} \label{appx:Chicago_market_dynamics}

In this section, we provide %
high level descriptions of the market dynamics in the City of Chicago, which provide support for the assumption that the market conditions are stationary during the

First, we present in Figure~\ref{fig:volume_by_how} the average (over the first 10 weeks of 2020) number of trips originating during each \emph{hour-of-week}. %
The $0^{\mathrm{th}}$ hour-of-week corresponds to midnight--1 a.m.\@ on Mondays, and the $1^\mathrm{st}$ hour-of-week corresponds to 1--2 a.m.\@ on Mondays, and so on. The gray stripes indicate the morning rush hours (7--10 a.m.\@) and the green stripes indicate the evening rush hours (5--8 p.m.).
We can observe clear weekly patterns: during weekdays, trip volume peaks during morning and evening rush hours; the market is the most busy during Friday and Saturday evenings (during which people take rides to commute from work as well as to and from restaurants and bars).

\begin{figure}[t!]
  \centering
  \includegraphics[scale = \PyplotScale]{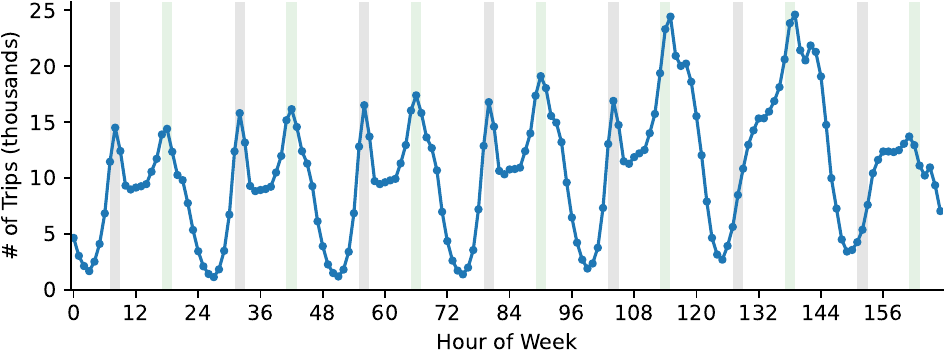}
  \caption{Average number of trips by hour-of-week. %
  }
  \label{fig:volume_by_how}
\end{figure}

\begin{figure}[t!]
  \centering
  \includegraphics[scale = \PyplotScale]{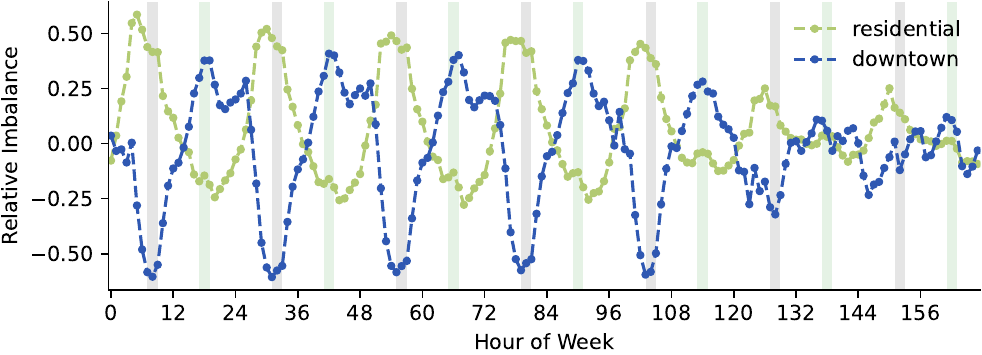}
  \caption{Average trip flow imbalance of the residential area and downtown, by hour-of-week. %
  }
  \label{fig:imbalance_by_how_two_locations}
\end{figure}

In Figure~\ref{fig:imbalance_by_how_two_locations}, we illustrate the average flow imbalance of the two locations discussed in Section~\ref{sec:intro}, the residential area (Lake View) and the downtown area (The Loop). %
The residential area sees more trips originating than ending in the area during the weekday morning rush hours, and the opposite during the evening rush hours. Downtown, on the other hand, sees more trips ending in the area in the morning, and leaving the area in the evening.
Importantly, observe that the trip flow imbalance of both locations remain relatively constant during the weekday morning rush hours (7--10 a.m.\@) and the evening rush hours (5--8 p.m.).

\begin{figure}[t!]
  \centering
  \subcaptionbox{Histogram of trip duration.\label{fig:2020-pre-covid_trip_seconds}}{\includegraphics[scale = \PyplotScale]{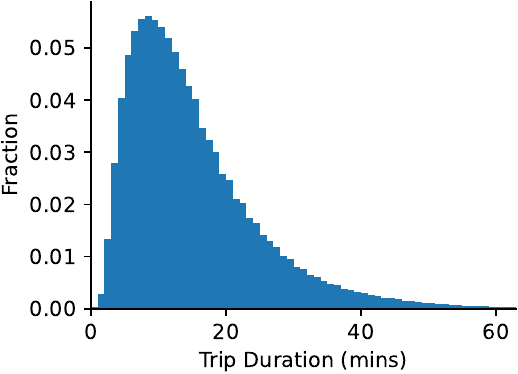}}
  \hspace{2em}
  \subcaptionbox{Cumulative distribution of trip duration.\label{fig:2020-pre-covid_trip_seconds_cumulative}}{\includegraphics[scale = \PyplotScale]{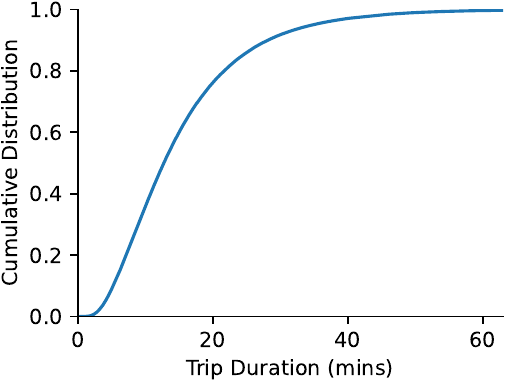}}
  \caption{Distribution of trip duration (in minutes).}
  \label{fig:trip_duration_distribution}
\end{figure}

We present the distribution of trip duration (in minutes) in Figure~\ref{fig:trip_duration_distribution}, for all trips during the first 10 weeks of 2020 (the distribution is not qualitatively different if we focus only on trips that take place during the morning rush hours).
We can see that the vast majority of trips take fewer than 60 minutes, and roughly 80\% of all trips are shorter than 20 minutes.
Given that the flow imbalance of the two locations remain relatively constant for three hours during the morning rush (as shown in Figure~\ref{fig:imbalance_by_how_two_locations}), %
the market condition should not change substantially during the time a driver completes an average trip.

\subsection{Simulation Details} \label{appx:simulation_details}

We conduct two types of simulations in this paper:  a stationary setting where the rider demand and driver supply are stationary over time and correspond to the average market condition during the first 10 weeks of 2020, and a dynamic setting where the rider demand and driver supply change over time according to the realized trip flows from the beginning of 2019 to the end of 2020.

\subsubsection{Stationary Setting} \label{appx:stationary_setting}

This section describes the way we construct the %
economy for the stationary setting, using data from the first 10 weeks of 2020.
This economy is used in \Cref{sec:stationary_simulation}, \Cref{appx:undampened_Newton}, \Cref{appx:GD_simulations}, and \Cref{appx:stationary_simple_direction}.

\subparagraph{Trip Duration and Cost}

We consider the $n = 77$ community areas in the City of Chicago as our set of locations.
Among all $n^2 = 5929$ origin-destination (OD) pairs, %
$5826$ of them have at least one observed trip.
For each of these OD pairs $(i, j)$, we calculate the average duration of the trips and use it as $d_{i, j}$.
For the remaining 16 edges for which we did not observe any trip, we impute the duration with the length of the shortest path using the edges with observed trips.%
\footnote{
  There is little or no rider demand between these extremely sparse origin-destination pairs. We nevertheless estimate the duration of these trips such that we can properly account for the time it takes a driver to relocate without a rider.
}
We assume that all trips cost \$20/hour, \ie $c_{i, j} = 20 d_{i, j}$ for all $i,j \in \gL$.

\subparagraph{Rider Demand}

Among the %
approximately 18 million recorded trips, %
$148$ thousands started during 7--8 a.m.\@ on Wednesdays. For each OD pair $(i, j) \in \gL^2$, we calculate the average number of observed trips per hour, and denote it as $\xobs_{i, j}$.
Among all $n^2 = 5929$ OD pairs,
$3331$
are traversed by at least one trip during the time window that we focus on.
For each $(i,j) \in \gL^2$ such that we observed no trips in the data, we assume $q\subij(r) = 0$ for all $r \geq 0$. %
A constant zero demand function does not satisfy our assumption that $q_{i, j}(\cdot)$ is strictly decreasing for all $i,j \in \gL$. This assumption is, however, sufficient but not necessary for establishing our theoretical results, %
and having zero demand for these OD pairs does not lead to any issues in our simulations.

For each $(i,j) \in \gL^2$ for which we observe at least one trip, we calculate the average price (the sum of the \emph{Fare} and \emph{Additional Charges} columns) of the recorded trips, and denote it as $\pobs_{i, j}$. %
To construct the rider demand functions, we assume that (i) the riders' average willingness to pay for a trip is \$1 per minute (\$60/hour), and (ii) the willingness to pay follows an exponential distribution. Technically, this is assuming that %
the demand function is of the form
\begin{equation*} %
  q_{i, j}(r) = Q_{i, j} \exp(- r / (60 d_{i, j}))
\end{equation*}
for some constants $(Q_{i, j})_{(i,j) \in \gL^2}$.
Intuitively, $Q_{i, j}$ can be interpreted as the total amount of riders interested in traveling from $i$ to $j$.
For each $(i,j) \in \gL^2$ \st at least one trip is recoreded, we determine $Q_{i, j}$ by solving $\xobs_{i, j} = Q_{i, j} \exp(- \pobs_{i, j} / (60 d_{i, j}))$.

\subparagraph{Driver Supply}

For the total number of drivers $m$, we use the minimum number that can fulfill the observed rider flow $(\xobs_{i, j})_{(i, j) \in \gL^2}$ and maintain flow balance.
Formally, we assume
\begin{mini}
  {\vy \in \R^{n^2}}{\sum_{i, j \in \gL} d_{i, j} y_{i, j}}{}{m =}
  \addConstraint{y_{i, j}}{\ge \xobs_{i, j}, \quad}{\forall i, j \in \gL}
  \addConstraint{\sum_{j \in \gL} y_{i, j}}{= \sum_{j \in \gL} y_{j, i}, \quad}{\forall i \in \gL.}
  \label{eq:find_min_supply_level}
\end{mini}
This gives us $m = \BalancedDrivers$, which is almost certainly an underestimation of the number of online drivers since in practice it's very unlikely for drivers to relocate in the most efficient manner.
In comparison, the observed total supply hour that's spent ``on trip'' is $\sum_{i, j \in \gL} d_{i, j} \xobs_{i, j} = \ObservedDrivers$.
This implies that in practice the efficiency (\ie fraction of drivers' time spent on trip) is most likely lower than $\ObservedDrivers / \BalancedDrivers = \ObservedDriversRatio$.
The platforms may observe seemingly a higher level of utilization, which indicates that some drivers might be relocating offline and not being counted towards the total supply hours.

\subparagraph{Driver Relocation}

To coordinate the driver relocation flows, we use phantom demand functions $\qtilde_{i, j}(r) = 500 (\max\{0, 1 - r / 3\})^4$ for all $i, j \in \gL$.%
\footnote{
  This phantom demand function does not satisfy our assumption \ref{cond:phantom_zero_price} that $d_{i, j} \qtilde_{i, j}(0) > m$.
  The assumption is used to prove existence of the market clearing surge multiplier $\vpi$ for any $\vphi$.
  In practice, we are only interested in a limited range of reasonable values of $\vphi$, so the assumption is only sufficient but not necessary.
}
Thus, relocation occurs only when the price is below \$3.

\subsubsection{Non-Stationary Setting} \label{appx:nonstationary_setting}

This section describes the non-stationary setting used in \Cref{sec:nonstationary_simulations} and \Cref{appx:dynamic_simple_direction}, representing Chicago's market conditions for 7--8 a.m.\@ on Wednesdays, from the beginning of 2019 through the end of 2020.
There are a total of 139 million trips %
during the two years.
We removed %
a total of $43$ trips with distances that are more than 30 times the median distance of trips with the same origin and destination--- these are most likely due to errors in the data recording/processing proces.

\subparagraph{Trip Duration and Cost}

Among all $n^2 = 5929$ OD pairs, 5925 of them have at least one trip in 2019 or 2020.
For these OD pairs, we compute the average trip duration over the two years and denote it as $d_{i, j}$ for OD pair $(i, j)$.
The remaining 4 OD pairs have no trip in 2019 and 2020, and we set their trip duration to be the shortest path distance using the duration of the observed trips.
Similar to the stationary setting, we set the trip cost to be the average trip duration multiplied by the driver cost rate $\$20$ per hour.

\subparagraph{Rider Demand}
The rider demand functions are constructed in the same way we construct the rider demand function for the stationary setting, except that we use only recorded trips from each week to determine the demand function for each week $\vq\supt$. %
More specifically, for each OD pair $(i, j) \in
  \gL^2$ and each week $t$, let $\xobs_{i, j, t}$ be the number of recorded trips from $i$ to $j$ that started during 7--8 a.m.\@ on Wednesday.
If at least one trip is recorded, \ie if $\xobs_{i, j, t} > 0$, we use the average price of these trips as $\pobs_{i, j, t}$.
The rider demand functions are defined as
\[
  q_{i, j}^{(t)}(r) = Q_{i, j}^{(t)} \exp(-r / (60 d_{i, j})),
\]
where $Q_{i, j}^{(t)}$ is the solution of
$\xobs_{i, j, t} = Q_{i, j}^{(t)} \exp(-\pobs_{i, j, t} / (60 d_{i, j}))$.

\subparagraph{Driver Supply}

Similar to the stationary setting, for each week $t$, we compute the minimum number of drivers needed to serve all observed trips while maintaining flow balance, and use it as the number of drivers $m^{(t)}$ in the simulation.
Technically, we solve \eqref{eq:find_min_supply_level} for each week $t$, using $\xobs_{i, j,t}$ instead of $\xobs_{i, j}$.

\subparagraph{Driver Relocation}

To coordinate the flow of driver relocation, we use phantom demand functions $\qtilde_{i, j}(r) = 1000 (\max\{0, 1 - r / 4\})^4$ for all OD pairs $(i, j)$.
Relocation occurs only when the price is below \$4.

\subparagraph{Events at McCormic Place}

There are a total of 101 Wednesdays that are working days during 2019 and 2020.
As we have briefly explained in \Cref{sec:nonstationary_simulations}, we excluded a total of \EventDays{} of them from our simulations,
during which it is very likely that major events were held at McCormick Place, which is located in the South Side, the community area just south of The Loop. %
In particular, we count the trips from North Side and Downtown to South Side during 7--8 a.m.\@ (see \Cref{fig:inflow_33}), and excluded days for which this number greater than 300.

\begin{figure}[t]
  \centering
  \includegraphics[scale = \PyplotScale]{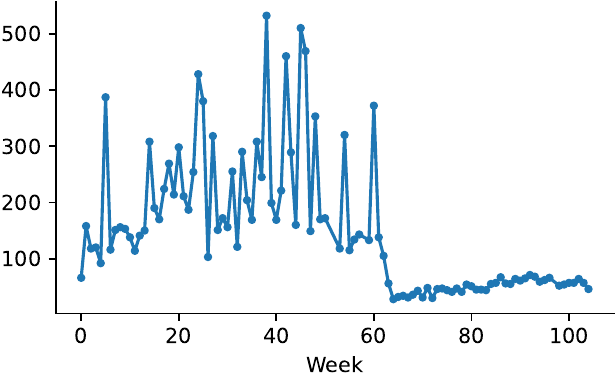}
  \caption{Total trip flow into South Side over weeks. %
  }
  \label{fig:inflow_33}
\end{figure}

During a typical ``non-event'' day, the South Side community area has an imbalance of around $0.4$, meaning that more trips originate from the area than end in the area (\ie the area is a popular origin).
During the event days, however, the imbalance can drop to as low as $-0.4$, when substantially more riders request trips ending in the area.
This leads to a very different flow of driver supply in space, thus %
a ridesharing platform should adjust the marginal values of drivers for this area as well as the neighboring areas. This is practically feasible since major events are typically planned months ahead of time.

\paragraph{Optimal OD-Based Adjustments and the Na\"ive Origin-Based Multipliers}

\begin{figure}[t]
  \centering
  \subcaptionbox{Optimal OD-based additive adjustments $\vphi^{(t)}$.\label{fig:optimal_phi}}{\includegraphics[scale = \PyplotScale]{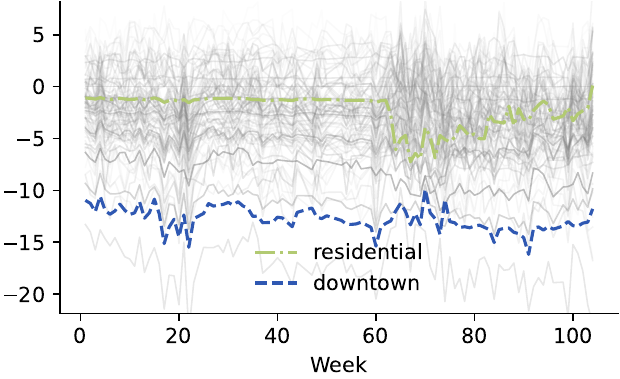}}
  \hfill
  \subcaptionbox{Origin-based multipliers $\vpi^{(t)}$.\label{fig:naive_pi}}{\includegraphics[scale = \PyplotScale]{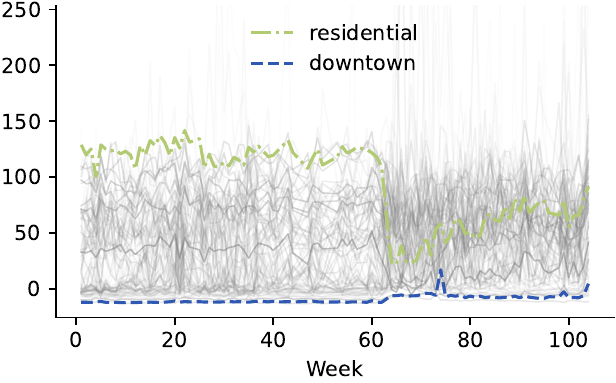}}
  \caption{Optimal OD-based adjustments and na\"ive origin-based multipliers $\vphi^{(t)}$ and $\vpi^{(t)}$ for each week.}
  \label{fig:optimal_phi_and_naive_pi_weekly}
\end{figure}

\Cref{fig:optimal_phi_and_naive_pi_weekly} shows the optimal OD-based adjustments $\vphi^{(t)}$ and the na\"ive origin-based multipliers $\vpi^{(t)}$ for each week $t$.
In \Cref{fig:optimal_phi}, the optimal adjustment of the residential area stays relatively constant over time, until the pandemic hits in March 2020, after which the adjustment drops substantially.
Overall, both the optimal adjustments and the na\"ive origin-based multipliers become more volatile after the pandemic hits.

\subsection{The Pure Newton Method} \label{appx:undampened_Newton}

Under the INP mechanism defined in Section~\ref{sec:INP_definition}, the step size \eqref{eq:INP_stepsize} is determined such that the (linear approximation of the) origin-based multipliers will not change week-over-week by more than some constant $\tau > 0$ for any location.
Since the update direction \eqref{eq:INP_direction} corresponds to the one under the Newton method for solving systems of equations, the INP mechanism %
can be considered %
as a ``damped Newton'' algorithm.
If we remove the stepsize limit (\ie setting $\alpha\supt = 1$ for all $t$) %
and run the \emph{pure Newton} method, the results for the stationary setting (same as the setting studied in Section~\ref{sec:stationary_simulation}) are shown in \Cref{fig:pure_Newton_pi_phi} and \Cref{fig:pure_Newton_obj}.

\begin{figure}[t]
  \centering
  \subcaptionbox{The price adjustments $\vphi^{(t)}$.\label{fig:pure_Newton_phi}}{\includegraphics[scale = \PyplotScale]{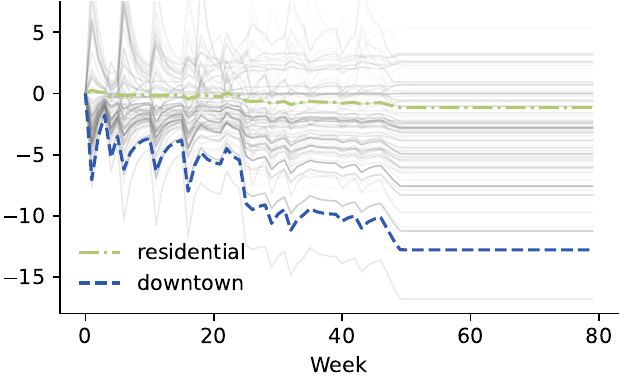}}
  \hfill
  \subcaptionbox{The multipliers $\vpi^{(t)}$.\label{fig:pure_Newton_pi}}{\includegraphics[scale = \PyplotScale]{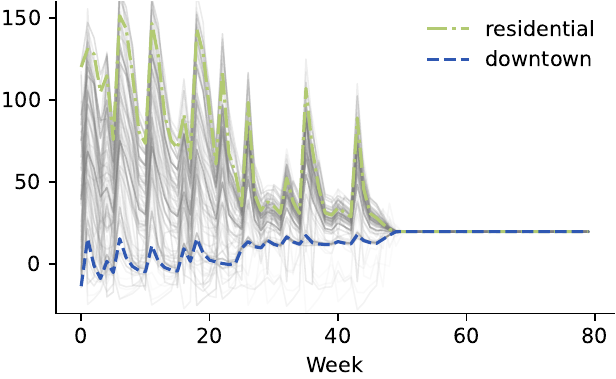}}
  \caption{$\vphi^{(t)}$ and $\vpi^{(t)}$ over iterations $t$ for the pure Newton method in the stationary setting.}
  \label{fig:pure_Newton_pi_phi}
\end{figure}

\begin{figure}[t]
  \centering
  \subcaptionbox{Social welfare (\$/hour).\label{fig:pure_Newton_welfare}}{\includegraphics[scale = \PyplotScale]{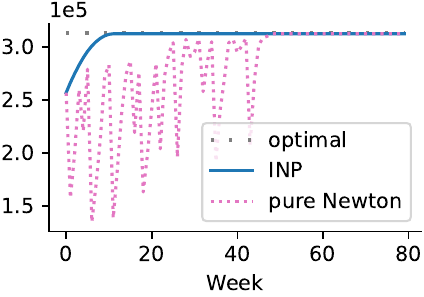}}
  \hfill
  \subcaptionbox{Dual objective.\label{fig:pure_Newton_dual}}{\includegraphics[scale = \PyplotScale]{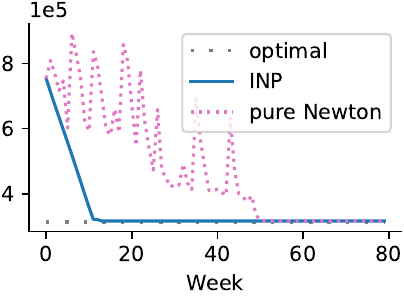}}
  \hfill
  \subcaptionbox{Function $f$.\label{fig:pure_Newton_log_f}}{\includegraphics[scale = \PyplotScale]{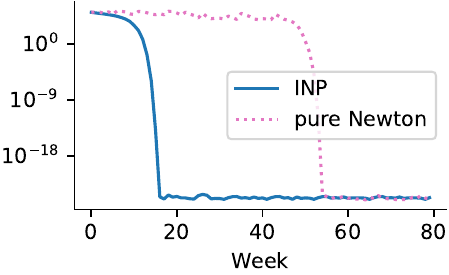}}
  \caption{Social welfare, dual objective, and the Lyanupov function $f$ over iterations $t$ for the pure Newton method in the stationary setting.}
  \label{fig:pure_Newton_obj}
\end{figure}

We can see that despite the fact that the pure Newton algorithm is still able to converge, the very large updates and frequent backtracking leads to very volatile market-clearing multipliers and social welfare (which are highly undesirable in practice).
The convergence also takes %
more iterations in comparison to the results presented in Section~\ref{sec:stationary_simulation}, where the properly ``paced'' INP mechanism never has to backtrack.

\subsection{Gradient Descent on the Lyapunov Function} \label{appx:GD_simulations}

Given that the gradient of the Lyanupov function $f$ at a market clearing outcome can be computed using information available in the market, %
a natural alternative (to INP) to consider is the gradient descent algorithm \wrt $f$.
More specifically, instead of using the Newton direction %
\eqref{eq:INP_direction}, we can use %
$-\nabla f$ at the previous market clearing outcome as the update direction. In conjunction with backtracking line search (recall that $f$ as a function of $\vphi$ is not necessarily convex), the algorithm is guaranteed to converge to the unique global minimum of $f$ where all origin-based multipliers are equal.

\begin{figure}[t]
  \centering
  \subcaptionbox{The price adjustments $\vphi^{(t)}$.\label{fig:gradient_phi}}{\includegraphics[scale = \PyplotScale]{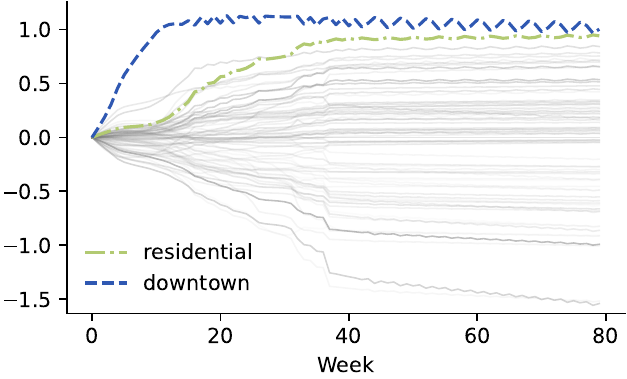}}
  \hfill
  \subcaptionbox{The multipliers $\vpi^{(t)}$.\label{fig:gradient_pi}}{\includegraphics[scale = \PyplotScale]{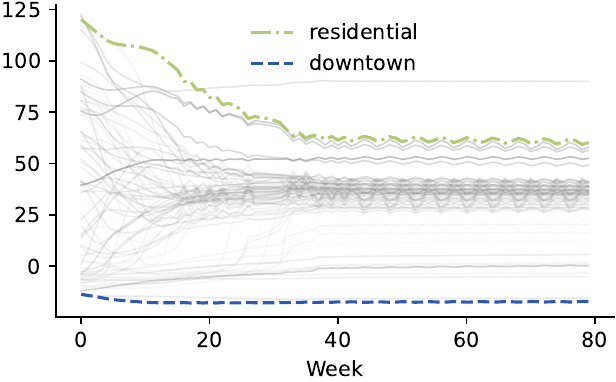}}
  \caption{$\vphi^{(t)}$ and $\vpi^{(t)}$ over iterations $t$ for the gradient descent method in the static setting.}
  \label{fig:gradient_pi_phi}
\end{figure}

\begin{figure}[t]
  \centering
  \subcaptionbox{Social welfare (\$/hour).\label{fig:gradient_welfare}}{\includegraphics[scale = \PyplotScale]{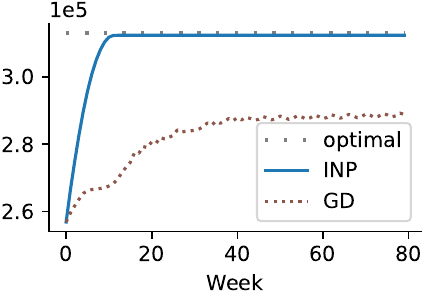}}
  \hfill
  \subcaptionbox{Dual objective.\label{fig:gradient_dual}}{\includegraphics[scale = \PyplotScale]{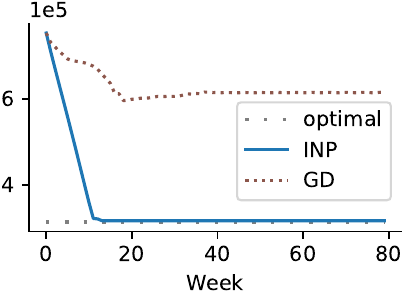}}
  \hfill
  \subcaptionbox{Function $f$ in linear scale.\label{fig:gradient_f}}{\includegraphics[scale = \PyplotScale]{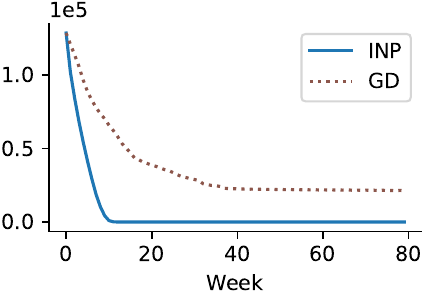}}
  \caption{Social welfare, dual objective, and $f$ over iterations $t$ for the gradient descent method in the static setting.}
  \label{fig:gradient_obj}
\end{figure}

In \Cref{fig:gradient_pi_phi} and \Cref{fig:gradient_obj}, we present the outcome under this %
algorithm (labeled GD) for the same stationary setting as in Section~\ref{sec:stationary_simulation}.
Similar to the INP mechanism, we limit the stepsize such that the expected week-over-week change in the origin-based multipliers is at most $\tau = 10$, and use the same parameters for backtracking line search.
We can see that despite the theoretical convergence guarantees, the convergence is very slow--- after 80 iterations, the outcome is still far from optimal (and not much progress was made even after we run the algorithm for a total of $500$ iterations). %

  The cyclic oscillations in $\vphi\supt$ and $\vpi\supt$ starting around the 20th iteration is not a result of backtracking line search, and %
  a smaller stepsize limit $\tau$ or different backtracking parameters does not help.
  This suggests that the gradient of $f$ is %
  not an effective direction for updating the OD-based adjustments. Also note that this method does not require less information than INP, since the local price sensitivity of rider demand is also necessary for computing $\nabla f$.

\subsection{A Simple Update Direction} \label{appx:pi_as_direction}

In this section, we present simulation results for a %
very simple mechanism, where we update the OD-based adjustments using only the origin-based multipliers.
Specifically, we adjust $\phi_i$ for each $i \in \gL$ by $\pi_i \suptmo - \pi_n \suptmo$ (multiplied by some constant), and without using the observed trip volume or the price sensitivity.\footnote{Whether to subtract $\pi_n \suptmo$ has no impact on the market-clearing outcome since prices depend only on the differences in the OD-based adjustments. The current forms, however, keeps $\phi_n\supt$ at zero for all $t$, making it easier to observe the changes of $\vphi\supt$ over time.}
This method does not necessarily have a theoretical guarantee, since the update direction may not be a descent direction of the Lyanupov function $f$ as defined in \eqref{eq:quadratic_objective}.
It works well and is quite robust, %
however, since by increasing the $\phi_i$ for a location with high $\pi_i$, we reduce the prices of trips ending in this location (relative to other locations), thereby increasing the driver supply for this location and reducing the multiplier for subsequent iterations.

\subsubsection{Stationary Market Condition} \label{appx:stationary_simple_direction}

\begin{figure}[t!]
  \centering
  \subcaptionbox{The price adjustments $\vphi^{(t)}$.\label{fig:simple_phi}}{\includegraphics[scale = \PyplotScale]{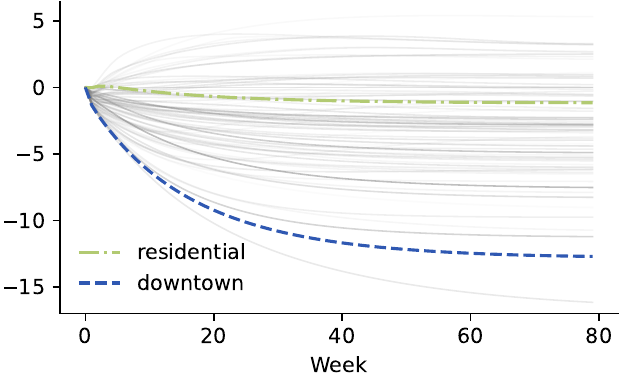}}
  \hfill
  \subcaptionbox{The multipliers $\vpi^{(t)}$.\label{fig:simple_pi}}{\includegraphics[scale = \PyplotScale]{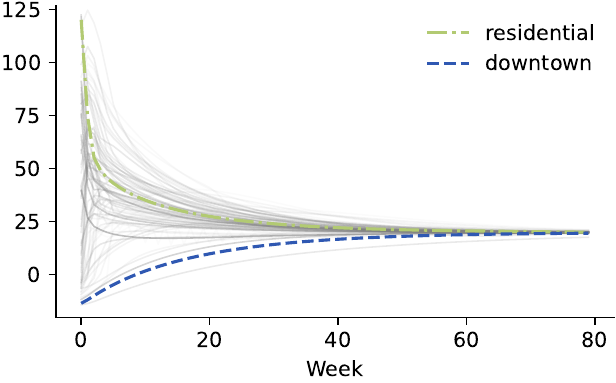}}
  \caption{$\vphi^{(t)}$ and $\vpi^{(t)}$ over iterations $t$ for the simple direction in the static setting.}
  \label{fig:simple_pi_phi}
\end{figure}

\begin{figure}[t!]
  \centering
  \subcaptionbox{Social welfare (\$/hour).\label{fig:simple_welfare}}{\includegraphics[scale = \PyplotScale]{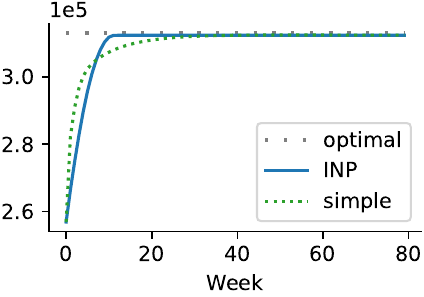}}
  \hfill
  \subcaptionbox{Dual objective.\label{fig:simple_dual}}{\includegraphics[scale = \PyplotScale]{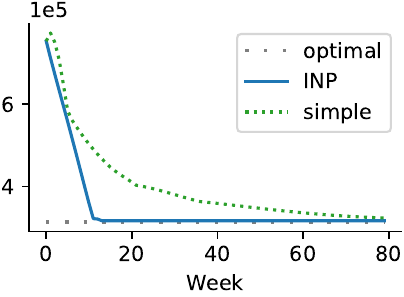}}
  \hfill
  \subcaptionbox{Function $f$.\label{fig:simple_log_f}}{\includegraphics[scale = \PyplotScale]{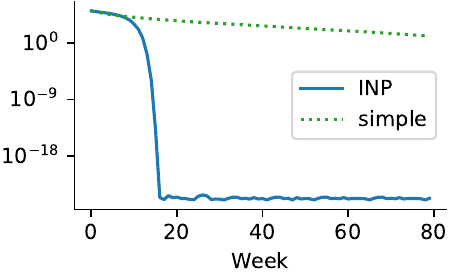}}
  \caption{Social welfare, dual objective, and $f$ over iterations $t$ for the simple direction in the static setting.}
  \label{fig:simple_obj}
\end{figure}

We first consider the simple direction in the static setting studied in Section~\ref{sec:stationary_simulation}, with %
$$
  \vphi\supt = \vphi\suptmo + 0.01 \vpi \suptmo.
$$
\Cref{fig:simple_pi_phi} plots the price adjustments $\vphi^{(t)}$ and the multipliers $\vpi^{(t)}$ over iterations $t$, and %
\Cref{fig:simple_obj} compares the social welfare, the dual objective, and the function $f$ over iterations $t$ with those under the INP mechanism.
The multipliers $\vpi^{(t)}$ converge to the same value and the social welfare converges to \SimpleWelfareRatio{} of the optimal social welfare.
The convergence rate appears to be linear, as shown in \Cref{fig:simple_log_f}.

\subsubsection{Non-Stationary Market Conditions} \label{appx:dynamic_simple_direction}

\begin{figure}[t]
  \centering
  \subcaptionbox{The OD-based adjustments $\vphi^{(t)}$.}{\includegraphics[scale = \PyplotScale]{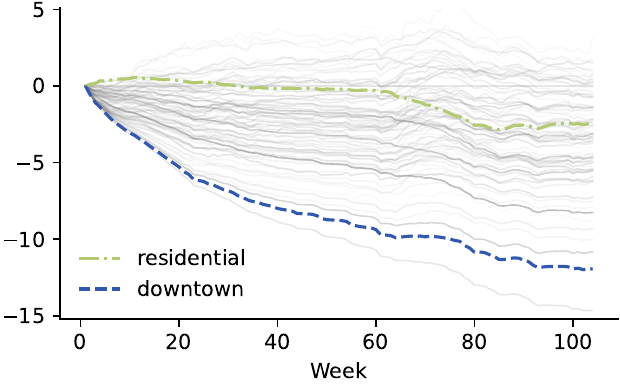}}
  \hfill
  \subcaptionbox{The origin-based multipliers $\vpi^{(t)}$.}{\includegraphics[scale = \PyplotScale]{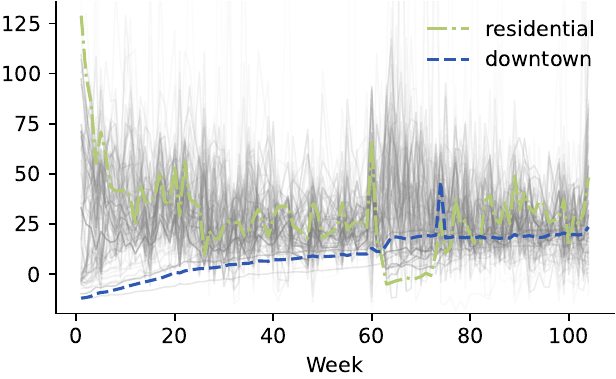}}
  \caption{$\vphi\supt$ and $\vpi\supt$ over weeks $t$, for the simple update direction in the dynamic setting.}
  \label{fig:dynamic_simple_phi_pi}
\end{figure}

\begin{figure}[t]
  \centering
  \subcaptionbox{Social welfare.}{\includegraphics[scale = \PyplotScale]{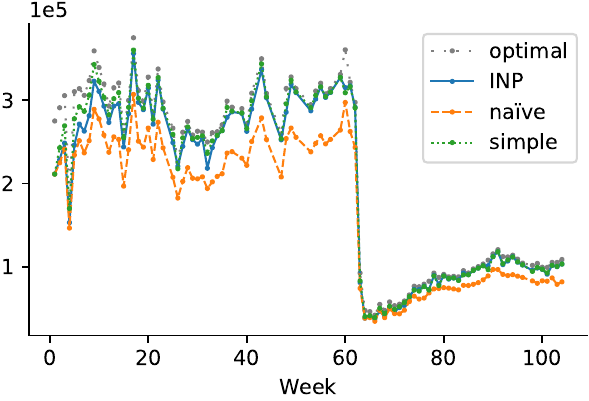}}
  \hfill
  \subcaptionbox{Welfare ratio to the optimum.\label{fig:dynamic_simple_welfare_ratio}}{\includegraphics[scale = \PyplotScale]{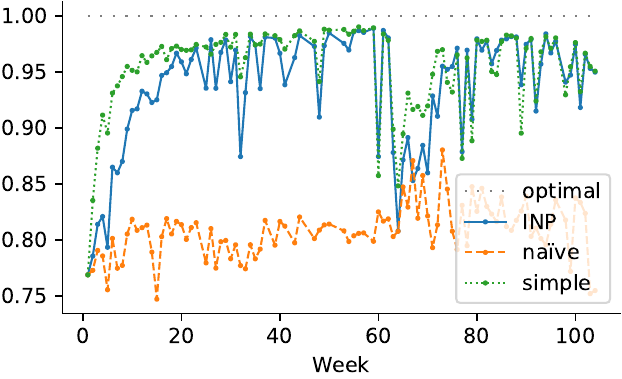}}
  \caption{Social welfare and welfare ratio over iterations $t$ for the simple direction in the dynamic setting.}
  \label{fig:dynamic_simple_welfare}
\end{figure}

We also test the simple update direction in the non-stationary setting, with %
$\vphi\supt = \vphi\suptmo + 0.005 \vpi \suptmo$.
The results are shown in \Cref{fig:dynamic_simple_phi_pi} and \Cref{fig:dynamic_simple_welfare}.
Despite having a worse empirical local convergence rate compared to the INP mechanism (linear vs.\ superlinear as suggested by \Cref{fig:simple_log_f}), the simple direction performs quite well in the dynamic setting.
In particular, it appears to be more robust to the non-stationarity of the demand than the INP mechanism, resulting in less fluctuations in \Cref{fig:dynamic_simple_welfare_ratio}.

\end{document}